\documentclass[prx,twocolumn,english,superscriptaddress,floatfix,longbibliography]{revtex4-2}
\usepackage[letterpaper,top=2cm,bottom=2cm,left=3cm,right=3cm,marginparwidth=1.75cm]{geometry}

\usepackage{amsthm,amsmath,bbm,dsfont}
\usepackage{graphicx}
\usepackage{dcolumn}
\usepackage[colorlinks,citecolor=blue,linkcolor=blue,urlcolor=blue]{hyperref} 
\usepackage{amssymb}
\usepackage[english]{babel}
\usepackage{indentfirst}
\usepackage{comment}
\usepackage{caption,subcaption}
\usepackage{tikz}

\newtheorem{theorem}{Theorem}
\newtheorem{corollary}{Corollary}[theorem]
\newtheorem{lemma}[theorem]{Lemma}
\newtheorem{prop}{Proposition}
\theoremstyle{definition}
\newtheorem{definition}{Definition}

\begin{document}

\preprint{APS/123-QED}
\author{Yi-Xuan Wang}
\affiliation{Institute of Theoretical Physics and IQST, Albert-Einstein Allee 11, Ulm University, 89081 Ulm, Germany}
\affiliation{Department of Modern Physics, University of Science and Technology of China, Hefei 230026, China}
\author{Yuval Gefen}
\email{yuval.gefen@weizmann.ac.il}
\affiliation{Department of Condensed Matter Physics, Weizmann Institute of Science, Rehovot 7610001, Israel}

\title{State Engineering of Unsteerable Hamiltonians}
\begin{abstract}
Lindbladian dynamics of open systems may be employed to steer a many-body system towards a non-trivial ground state of a local Hamiltonian. Such protocols provide us with tunable platforms facilitating the engineering and study of non-trivial many-body states. Steering of a quantum system towards a degenerate ground state manifold  provides us with a protected platform to employ many-body states as a resource for quantum information processing. Notably, ground states of frustrated local Hamiltonians have been known not to be amenable to steering protocols. 
Revisiting this intricate physics we report two new results: (i) we find a broad class of (geometrically) frustrated local Hamiltonians for which steering of the ground state manifold is possible through a sequence of discrete steering steps. Following the steering dynamics, states within the degenerate ground-state manifold keep evolving in a non-stationary manner. (ii) For the class of Hamiltonians with ground states which are non-steerable through local superoperators, we derive a "glass floor" on how close to the ground state one can get implementing a steering protocol.  This is expressed invoking the concept of cooling-by-steering (a lower bound of the achievable temperature), or through an upper bound of the achievable fidelity. Our work provides a systematic outline for studying quantum state manipulation of a broad class of strongly correlated states.
\end{abstract}


\maketitle

\tableofcontents
\newpage

\section{Introduction}\label{sec:Introduction}
Constructing the ground state (GS) of a non-trivial Hamiltonian is an important and timely challenge, and concerns low-temperature phases of matter, controlled investigation of many-body physics~\cite{verstraete2009quantum,georgescu2014quantumsim,mcardle2020quantumchem}, mathematically equivalent computational problems ~\cite{kempe2006complexity,lucas2014ising}, quantum information processing~\cite{kitaev2003fault,FTQCPhysRevA.57.127}. Steering of quantum states  ~\cite{Measurement-inducedsteeringPhysRevResearch.2.033347,PRXQuantum.4.020347,PhysRevResearch.6.013244}  provides us with a tool of many-body quantum state steering. A relatively simple paradigm of such protocols is passive steering (e.g. "blind" steering~\cite{Measurement-inducedsteeringPhysRevResearch.2.033347}), which provides us with an automated approach of pushing the system's state towards a desired target state. Note that in the continuum-time (weak interaction) limit, the dynamics of such protocols can be cast in the form of Lindblad dynamics, giving rise to dissipative GS engineering. Unfortunately, such an approach to steering has limited applicability. It has been argued that only the ground states of Frustration-Free Hamiltonians are amenable to passive steering~\cite{verstraete2009quantum,ticozzi2012stabilizing,Measurement-inducedsteeringPhysRevResearch.2.033347}. 

We recall that passive steering of quantum states through generalized measurements ~\cite{Measurement-inducedsteeringPhysRevResearch.2.033347}  is a protocol consisting of three steps, cf. Fig.~\ref{fig:steer}. 
(i) Initializing the detector ancilla in a state which is independent of the system’s running state; 
(ii) coupling  (part of) the system with the detector, and evolving this composite with a local Hamiltonian, $H_{\mathrm{s-d}}$, and then decoupling the two. The system-detector coupling strength is tunable; 
(iii) performing the blind measurement, namely tracing out over the detector’s subspace, remaining with the system whose dynamics can be effectively described by a superoperator. “Passive”, or, in the present case, “blind”, refers to a protocol which is pre-determined, with no reference to information about the detectors’ readouts. A paradigmatic example is the Affleck-Kennedy-Lieb-Tasaki (AKLT) model~\cite{PhysRevLett.59.799AKLT,Measurement-inducedsteeringPhysRevResearch.2.033347}, where one can employ a passive protocol to steer towards the GS manifold. Generalizations to active protocols ~\cite{PhysRevResearch.6.013244,PRXQuantum.4.020347} and post-selective protocols~\cite{PhysRevA.105.L010203} are directions of active study. 

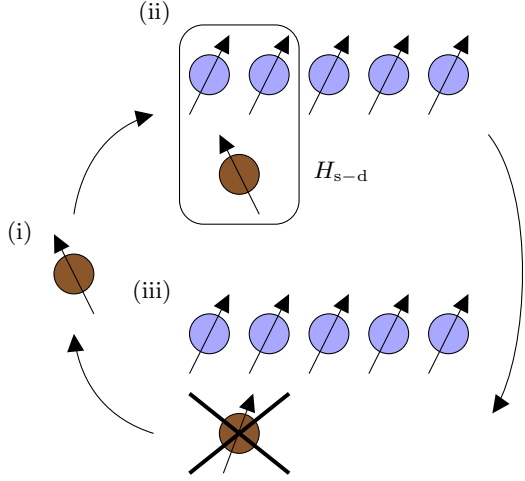
\begin{figure}[htbp]
    \centering
    \begin{tikzpicture}[x=0.75pt,y=0.75pt,yscale=-1,xscale=1]

\draw  [fill={rgb, 255:red, 168; green, 168; blue, 255 }  ,fill opacity=1 ] (118,70) .. controls (118,64.48) and (122.48,60) .. (128,60) .. controls (133.52,60) and (138,64.48) .. (138,70) .. controls (138,75.52) and (133.52,80) .. (128,80) .. controls (122.48,80) and (118,75.52) .. (118,70) -- cycle ;
\draw    (118,90) -- (136.66,52.68) ;
\draw [shift={(138,50)}, rotate = 116.57] [fill={rgb, 255:red, 0; green, 0; blue, 0 }  ][line width=0.08]  [draw opacity=0] (8.93,-4.29) -- (0,0) -- (8.93,4.29) -- cycle    ;
\draw  [fill={rgb, 255:red, 168; green, 168; blue, 255 }  ,fill opacity=1 ] (148,70) .. controls (148,64.48) and (152.48,60) .. (158,60) .. controls (163.52,60) and (168,64.48) .. (168,70) .. controls (168,75.52) and (163.52,80) .. (158,80) .. controls (152.48,80) and (148,75.52) .. (148,70) -- cycle ;
\draw    (148,90) -- (166.66,52.68) ;
\draw [shift={(168,50)}, rotate = 116.57] [fill={rgb, 255:red, 0; green, 0; blue, 0 }  ][line width=0.08]  [draw opacity=0] (8.93,-4.29) -- (0,0) -- (8.93,4.29) -- cycle    ;
\draw  [fill={rgb, 255:red, 168; green, 168; blue, 255 }  ,fill opacity=1 ] (178,70) .. controls (178,64.48) and (182.48,60) .. (188,60) .. controls (193.52,60) and (198,64.48) .. (198,70) .. controls (198,75.52) and (193.52,80) .. (188,80) .. controls (182.48,80) and (178,75.52) .. (178,70) -- cycle ;
\draw    (178,90) -- (196.66,52.68) ;
\draw [shift={(198,50)}, rotate = 116.57] [fill={rgb, 255:red, 0; green, 0; blue, 0 }  ][line width=0.08]  [draw opacity=0] (8.93,-4.29) -- (0,0) -- (8.93,4.29) -- cycle    ;
\draw  [fill={rgb, 255:red, 168; green, 168; blue, 255 }  ,fill opacity=1 ] (208,70) .. controls (208,64.48) and (212.48,60) .. (218,60) .. controls (223.52,60) and (228,64.48) .. (228,70) .. controls (228,75.52) and (223.52,80) .. (218,80) .. controls (212.48,80) and (208,75.52) .. (208,70) -- cycle ;
\draw    (208,90) -- (226.66,52.68) ;
\draw [shift={(228,50)}, rotate = 116.57] [fill={rgb, 255:red, 0; green, 0; blue, 0 }  ][line width=0.08]  [draw opacity=0] (8.93,-4.29) -- (0,0) -- (8.93,4.29) -- cycle    ;
\draw  [fill={rgb, 255:red, 168; green, 168; blue, 255 }  ,fill opacity=1 ] (238,70) .. controls (238,64.48) and (242.48,60) .. (248,60) .. controls (253.52,60) and (258,64.48) .. (258,70) .. controls (258,75.52) and (253.52,80) .. (248,80) .. controls (242.48,80) and (238,75.52) .. (238,70) -- cycle ;
\draw    (238,90) -- (256.66,52.68) ;
\draw [shift={(258,50)}, rotate = 116.57] [fill={rgb, 255:red, 0; green, 0; blue, 0 }  ][line width=0.08]  [draw opacity=0] (8.93,-4.29) -- (0,0) -- (8.93,4.29) -- cycle    ;
\draw  [fill={rgb, 255:red, 139; green, 87; blue, 42 }  ,fill opacity=1 ] (133,120) .. controls (133,114.48) and (137.48,110) .. (143,110) .. controls (148.52,110) and (153,114.48) .. (153,120) .. controls (153,125.52) and (148.52,130) .. (143,130) .. controls (137.48,130) and (133,125.52) .. (133,120) -- cycle ;
\draw    (153,140) -- (134.34,102.68) ;
\draw [shift={(133,100)}, rotate = 63.43] [fill={rgb, 255:red, 0; green, 0; blue, 0 }  ][line width=0.08]  [draw opacity=0] (8.93,-4.29) -- (0,0) -- (8.93,4.29) -- cycle    ;
\draw   (113,57) .. controls (113,50.37) and (118.37,45) .. (125,45) -- (161,45) .. controls (167.63,45) and (173,50.37) .. (173,57) -- (173,133) .. controls (173,139.63) and (167.63,145) .. (161,145) -- (125,145) .. controls (118.37,145) and (113,139.63) .. (113,133) -- cycle ;
\draw  [fill={rgb, 255:red, 168; green, 168; blue, 255 }  ,fill opacity=1 ] (118,200) .. controls (118,194.48) and (122.48,190) .. (128,190) .. controls (133.52,190) and (138,194.48) .. (138,200) .. controls (138,205.52) and (133.52,210) .. (128,210) .. controls (122.48,210) and (118,205.52) .. (118,200) -- cycle ;
\draw    (118,220) -- (136.66,182.68) ;
\draw [shift={(138,180)}, rotate = 116.57] [fill={rgb, 255:red, 0; green, 0; blue, 0 }  ][line width=0.08]  [draw opacity=0] (8.93,-4.29) -- (0,0) -- (8.93,4.29) -- cycle    ;
\draw  [fill={rgb, 255:red, 168; green, 168; blue, 255 }  ,fill opacity=1 ] (148,200) .. controls (148,194.48) and (152.48,190) .. (158,190) .. controls (163.52,190) and (168,194.48) .. (168,200) .. controls (168,205.52) and (163.52,210) .. (158,210) .. controls (152.48,210) and (148,205.52) .. (148,200) -- cycle ;
\draw    (148,220) -- (166.66,182.68) ;
\draw [shift={(168,180)}, rotate = 116.57] [fill={rgb, 255:red, 0; green, 0; blue, 0 }  ][line width=0.08]  [draw opacity=0] (8.93,-4.29) -- (0,0) -- (8.93,4.29) -- cycle    ;
\draw  [fill={rgb, 255:red, 168; green, 168; blue, 255 }  ,fill opacity=1 ] (178,200) .. controls (178,194.48) and (182.48,190) .. (188,190) .. controls (193.52,190) and (198,194.48) .. (198,200) .. controls (198,205.52) and (193.52,210) .. (188,210) .. controls (182.48,210) and (178,205.52) .. (178,200) -- cycle ;
\draw    (178,220) -- (196.66,182.68) ;
\draw [shift={(198,180)}, rotate = 116.57] [fill={rgb, 255:red, 0; green, 0; blue, 0 }  ][line width=0.08]  [draw opacity=0] (8.93,-4.29) -- (0,0) -- (8.93,4.29) -- cycle    ;
\draw  [fill={rgb, 255:red, 168; green, 168; blue, 255 }  ,fill opacity=1 ] (208,200) .. controls (208,194.48) and (212.48,190) .. (218,190) .. controls (223.52,190) and (228,194.48) .. (228,200) .. controls (228,205.52) and (223.52,210) .. (218,210) .. controls (212.48,210) and (208,205.52) .. (208,200) -- cycle ;
\draw    (208,220) -- (226.66,182.68) ;
\draw [shift={(228,180)}, rotate = 116.57] [fill={rgb, 255:red, 0; green, 0; blue, 0 }  ][line width=0.08]  [draw opacity=0] (8.93,-4.29) -- (0,0) -- (8.93,4.29) -- cycle    ;
\draw  [fill={rgb, 255:red, 168; green, 168; blue, 255 }  ,fill opacity=1 ] (238,200) .. controls (238,194.48) and (242.48,190) .. (248,190) .. controls (253.52,190) and (258,194.48) .. (258,200) .. controls (258,205.52) and (253.52,210) .. (248,210) .. controls (242.48,210) and (238,205.52) .. (238,200) -- cycle ;
\draw    (238,220) -- (256.66,182.68) ;
\draw [shift={(258,180)}, rotate = 116.57] [fill={rgb, 255:red, 0; green, 0; blue, 0 }  ][line width=0.08]  [draw opacity=0] (8.93,-4.29) -- (0,0) -- (8.93,4.29) -- cycle    ;
\draw  [fill={rgb, 255:red, 139; green, 87; blue, 42 }  ,fill opacity=1 ] (133,250) .. controls (133,244.48) and (137.48,240) .. (143,240) .. controls (148.52,240) and (153,244.48) .. (153,250) .. controls (153,255.52) and (148.52,260) .. (143,260) .. controls (137.48,260) and (133,255.52) .. (133,250) -- cycle ;
\draw    (135,270) -- (148.95,232.81) ;
\draw [shift={(150,230)}, rotate = 110.56] [fill={rgb, 255:red, 0; green, 0; blue, 0 }  ][line width=0.08]  [draw opacity=0] (8.93,-4.29) -- (0,0) -- (8.93,4.29) -- cycle    ;
\draw [line width=1.5]    (118,230) -- (168,270) ;
\draw [line width=1.5]    (118,270) -- (168,230) ;
\draw    (100,250) .. controls (77.37,243.01) and (63.56,226.15) .. (60.35,202.93) ;
\draw [shift={(60,200)}, rotate = 84.39] [fill={rgb, 255:red, 0; green, 0; blue, 0 }  ][line width=0.08]  [draw opacity=0] (8.93,-4.29) -- (0,0) -- (8.93,4.29) -- cycle    ;
\draw    (268,100) .. controls (292.25,130.15) and (288.27,209.62) .. (271.58,237.56) ;
\draw [shift={(270,240)}, rotate = 305.3] [fill={rgb, 255:red, 0; green, 0; blue, 0 }  ][line width=0.08]  [draw opacity=0] (8.93,-4.29) -- (0,0) -- (8.93,4.29) -- cycle    ;
\draw  [fill={rgb, 255:red, 139; green, 87; blue, 42 }  ,fill opacity=1 ] (50,170) .. controls (50,164.48) and (54.48,160) .. (60,160) .. controls (65.52,160) and (70,164.48) .. (70,170) .. controls (70,175.52) and (65.52,180) .. (60,180) .. controls (54.48,180) and (50,175.52) .. (50,170) -- cycle ;
\draw    (70,190) -- (51.34,152.68) ;
\draw [shift={(50,150)}, rotate = 63.43] [fill={rgb, 255:red, 0; green, 0; blue, 0 }  ][line width=0.08]  [draw opacity=0] (8.93,-4.29) -- (0,0) -- (8.93,4.29) -- cycle    ;
\draw    (60,140) .. controls (63.29,117.65) and (75.4,97.41) .. (97.22,90.76) ;
\draw [shift={(100,90)}, rotate = 166.35] [fill={rgb, 255:red, 0; green, 0; blue, 0 }  ][line width=0.08]  [draw opacity=0] (8.93,-4.29) -- (0,0) -- (8.93,4.29) -- cycle    ;

\draw (179,112.4) node [anchor=north west][inner sep=0.75pt]    {$H_{\mathrm{s-d}}$};
\draw (25,142) node [anchor=north west][inner sep=0.75pt]   [align=left] {(i)};
\draw (91,32) node [anchor=north west][inner sep=0.75pt]   [align=left] {(ii)};
\draw (88,172) node [anchor=north west][inner sep=0.75pt]   [align=left] {(iii)};

\end{tikzpicture}
    \caption{Schematic illustration of a steering protocol. (i) Initialize the detector (ancilla, red spin) in a given state, independent of the system’s running state; (ii) couple part of the system (blue spins) to a detector ancilla (red spin), employing a local Hamiltonian $H_{\mathrm{s-d}}$; (iii) trace out the detector.}
    \label{fig:steer}
\end{figure} 

The context of the present work is steering of quantum system towards many-body quantum states, specifically the ground states (GSs) of a quantum many-body Hamiltonian. Our ambition here is two-fold: first, to extend such passive steering protocols, employing \textit{local physical operations} (i.e., local superoperators a.k.a.  local quantum channels~\cite{preskill1998lecture}, consisting of local Kraus operators) beyond the Frustration-Free (FF) local Hamiltonian scenario, employing local Lindbladians. Second, for scenarios where  steering is not facilitated,   we  aim at  characterizing the distance one can get  from the target state. 

To address the first goal we revise the acceptable classification of local Hamiltonian models. Revising the present classification comprising (i) FF steerable (FFS) and (ii) non-FF non-steerable (NFFNS) classes, we introduce additional classes: (iii) Non-FF Steadily Steerable (NFFSS) and (iv) Non-FF Jittery Steerable (NFFJS), cf. Fig.~\ref{fig:classfication}. We establish certain criteria to determine whether ground states of class (iii) and (iv)  may or may not be steered toward employing passive steering. Our analysis extends the notion of the steerability of FF Hamiltonians to certain types of classically frustrated Hamiltonians.  Once steering (e.g., towards the GS) is successful, expectation values and correlation functions employing the target state as a resource are accessible. Note that as far as the FFS and NFFSS classes are concerned, our dynamical protocols steer towards a given state (a stationary state within the GS manifold). And as far as the NFFJS class is concerned, the Hamiltonian supports a degenerate GS.  Under the steering operations, any given state within the GS manifold evolves dynamically into other states within the space. As far as a second goal is concerned, for NFFNS Hamiltonians we find an upper bound on the asymptotic achievable fidelity (which, clearly, is less than 1). Alternatively, we find a lower bound on the temperature to which the system can be cooled.

\begin{figure}[htbp]
    \centering
    \begin{tikzpicture}[x=0.75pt, y=0.75pt, yscale=-1, xscale=1]

\definecolor{colorA}{rgb}{0.5, 0.6, 1.0}
\definecolor{colorA1}{rgb}{0.5, 0.5, 1.0}
\definecolor{colorB}{rgb}{0.3, 1.0, 0.3}
\definecolor{colorC}{rgb}{1.0, 0.3, 0.3}
\definecolor{colorD}{rgb}{0.7, 0.7, 0.7} 

\fill[colorA] (40,34) .. controls (40,20.75) and (50.75,10) .. (64,10) -- (136,10) .. controls (149.25,10) and (160,20.75) .. (160,34) -- (160,60) -- (40,60) -- cycle; 
\fill[colorA1] (80,36) .. controls (80,23) and (83,20) .. (96,20) -- (104,20) .. controls (117,20) and (120,23) .. (120,36) -- (120,34) .. controls (120,47) and (117,50) .. (104,50) -- (96,50) .. controls (83,50) and (80,47) .. (80,34) -- cycle;
\fill[colorD] (40,60) -- (100,60) -- (100,160) -- (64,160) .. controls (50.75,160) and (40,149.25) .. (40,136) -- cycle; 
\fill[colorB] (100,60) -- (160,60) -- (160,110) -- (100,110) -- cycle; 
\fill[colorC] (100,110) -- (160,110) -- (160,136) .. controls (160,149.25) and (149.25,160) .. (136,160) -- (100,160) -- cycle; 

\draw[line width=0.8pt] (40,34) .. controls (40,20.75) and (50.75,10) .. (64,10) -- (136,10) .. controls (149.25,10) and (160,20.75) .. (160,34) -- (160,136) .. controls (160,149.25) and (149.25,160) .. (136,160) -- (64,160) .. controls (50.75,160) and (40,149.25) .. (40,136) -- cycle;
\draw[line width=1pt] (80,36) .. controls (80,23) and (83,20) .. (96,20) -- (104,20) .. controls (117,20) and (120,23) .. (120,36) -- (120,34) .. controls (120,47) and (117,50) .. (104,50) -- (96,50) .. controls (83,50) and (80,47) .. (80,34) -- cycle;

\draw[line width=0.75pt] (40,60) -- (160,60);
\draw[line width=0.75pt] (100,60) -- (100,160);
\draw[line width=0.75pt] (100,110) -- (160,110);

\draw (100,35) node [anchor=center, font=\large]    {$A$};
\draw (140,35) node [anchor=center, font=\large]    {$B$};
\draw (130,85) node [anchor=center, font=\large]    {$C$};
\draw (130,135) node [anchor=center, font=\large]    {$D$};
\draw (70,110) node [anchor=center, font=\large]    {$\overline{C+D}$};

\draw (81,-15) node [anchor=north west,font=\large]   [align=left] {FFS};
\draw (171,25) node [anchor=north west,font=\large]   [align=left] {NFFSS};
\draw (171,75) node [anchor=north west,font=\large]   [align=left] {NFFJS};
\draw (171,125) node [anchor=north west,font=\large]   [align=left] {NFFNS};

\end{tikzpicture}
    \caption{Steerability of ground states of the underlying Hamiltonians (which are  made up of local operators). $A$: (Frustration Free steerable (FFS)) Hamiltonians: the Hamiltonian is  Frustration Free. $B$: (Non-Frustration-Free Steadily Steerable (NFFSS)) the Hamiltonian is not Frustration Free, but there exists a Frustration Free Parent Hamiltonian. $C$: (Non-FF Jittery Steerable (NFFJS)) and $\overline{C+D}$ (classification is unknown) satisfy the necessary conditions for non-FF steerability. $C$ are further known to be steerable. The dynamics involves steering towards a protected subspace. With continued steering,  jumps are induced among ground states that span this subspace. $D$: Hamiltonians which do not satisfy the necessary conditions, hence are Non-FF Non-Steerable (NFFNS). Hamiltonians belonging to $\overline{C+D}$ may turn out to be either  NFFJS or NFFNS. 
} 
    \label{fig:classfication}
\end{figure}
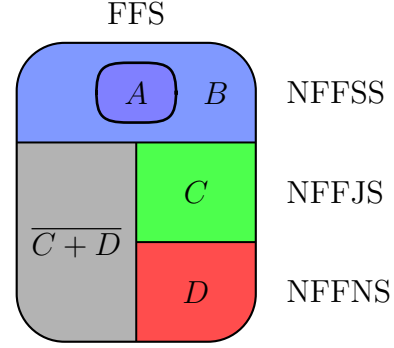

Note that this is not ”steering” in the historical sense of quantum mechanics~\cite{uola2020quantum}. There, ”steering” refers to  non-locality  exhibited by certain multipartite quantum states: a local operation on one part of the quantum state can change the conditioned reduced quantum state on other parts. Such a change of the reduced quantum state by a distant operation had been referred to as  ”quantum steering”.  Accepting the contextual meaning of "quantum steering" in reference to state engineering, we note that there are several important paradigms for the preparation of GSs employing such protocols, which are robust against decoherence~\cite{PRX15,arute2019quantum}: dissipative quantum state preparation~\cite{verstraete2009quantum}, steering of quantum states through generalized measurements~\cite{Measurement-inducedsteeringPhysRevResearch.2.033347}, or variational imaginary time evolution~\cite{mcardle2019variational}. These approaches have been studied experimentally as well~\cite{barreiro2011open,lin2013dissipative,barreiro2010experimental} and extend to further directions such as quantum cooling~\cite{Langbehn2023,Kishony2023,sugiura2025quantum} and quantum error correction~\cite{PhysRevResearch.6.013244}.

This work is structured as follows. In Sec.~\ref{sec:Definition of stabilizability} we provide the proper definition of steerability for a local Hamiltonian, classifying Hamiltonians into FFS, NFFSS, NFFJS, and NFFNS types. In Secs.~\ref{sec:Identifying an FFS Hamiltonian} and ~\ref{sec:Non-Frustration-Free Steadily Steerable Hamiltonians} we discuss the necessary and sufficient conditions, supplementing this discussion with examples, for FFS and NFFSS Hamiltonians. In Sec.~\ref{sec:Necessary conditions for NFFJS Hamiltonian} we explore the necessary conditions for NFFJS Hamiltonians. In doing so, we introduce the notion of conserved quantities in the target subspace and bipartite indistinguishability, where the negation of the latter enables local transitions between ground states. In Sec.~\ref{sec:NFFJS Examples} we examine the steerability of geometrically frustrated Hamiltonians composed of local Hamiltonians made up of commuting Pauli matrices—a standard template for geometric frustration. Such Hamiltonians are fully integrable but not FF, and are proven to belong to the NFFJS class. Our steering protocol employs the local Hamiltonians to construct local superoperators, and does not require the knowledge or otherwise the construction of the explicit GS(s). In Sec.~\ref{sec:NFFNS Hamiltonian} we focus on Hamiltonians that do not satisfy the necessary conditions for steerability (cf. Sec.~\ref{sec:Necessary conditions for NFFJS Hamiltonian})—namely, Hamiltonians that belong to the NFFNS class: we characterize the distance from the ground state that can be approached through passive steering, and derive a general lower bound on achievable effective temperature. In Sec.~\ref{sec:Implementations of NFFNS model}, we implement our analysis of non-steerable states to several examples of many-body Hamiltonians: the antiferromagnetic Heisenberg model, the Dirac/Majorana fermionic SYK model, and the Fermi-Hubbard model.

\begin{widetext}
\centering

\begin{table}[ht]
    \centering
    \makebox[\textwidth]{
    \begin{ruledtabular}
    \begin{tabular}{ccc} 
         Term used &  Abbreviation & Reference in the text\\ \hline 
         Ground State & GS     & Sec.~\ref{sec:Introduction}\\ 
         Steerable Subspace & / & Def.~\ref{def:stabilizability} in Sec.~\ref{sec:Definition of stabilizability}\\ 
 Frustration-Free&FF & Sec.~\ref{sec:Introduction}, Eq.~\eqref{eq:ff def} in Sec.~\ref{sec:Definition of stabilizability}\\ 
 Parent Hamiltonian & PH & Sec.~\ref{sec:Definition of stabilizability}\\ 
 Frustration-Free Steerable& FFS &Sec.~\ref{sec:Introduction},~\ref{sec:Identifying an FFS Hamiltonian}\\ 
 Non-Frustration-Free Steadily-Steerable & NFFSS & Sec.~\ref{sec:Introduction},~\ref{sec:Non-Frustration-Free Steadily Steerable Hamiltonians}\\  
 Non-Frustration-Free Jittery Steerable & NFFJS & Sec.~\ref{sec:Introduction},~\ref{sec:Necessary conditions for NFFJS Hamiltonian},~\ref{sec:NFFJS Examples}\\ 
 Non-Frustration-Free Non-Steerable & NFFNS & Sec.~\ref{sec:Introduction},~\ref{sec:NFFNS Hamiltonian}\\  
 Subspace Conserved Quantity & SCQ & Def.~\ref{def:SCQ} in Sec.~\ref{sec:Identifying an FFS Hamiltonian}\\ 
 Parent Hamiltonian Frustration-Free & PHFF & Def.~\ref{def:pHFF} in Sec.~\ref{sec:Necessary and Sufficient conditions for Non-Frustration-Free Steadily Steerable (NFFSS) Hamiltonian}\\ 
 Reduced Density Matrix & RDM & Sec.~\ref{sec:Identifying Non-Frustration-Free Steadily Steerable Hamiltonians}\\ 
 Bipartite Indistinguishability & / & Sec.~\ref{sec:Necessary conditions for NFFJS Hamiltonian} \\ 
 Fidelity of quantum states & / & Eq.~\eqref{eq:fid def} in Sec.~\ref{sec: Approximate steering and surrogate states}\\ 
 Surrogate State & / & Sec.~\ref{sec: Approximate steering and surrogate states}\\ 
 $p$-approximate Steerable Subspace & / & Def.~\ref{def:approx_stabilizability} in Sec.~\ref{sec: Approximate steering and surrogate states}\\ 
  presumed Surrogate State & / & Sec.~\ref{sec: Approximate steering and surrogate states}\\ 
 Distance to GS (fidelity, energy, temperature) & / & Sec.~\ref{sec: Approximate steering and surrogate states}
 \end{tabular}
 \end{ruledtabular}
 }
    \caption{List of terms and abbreviations.}
    \label{tab:abbreviations}
\end{table}

\end{widetext}


\section{Definition of steerability}\label{sec:Definition of stabilizability}
Before proceeding, we will briefly describe our steering protocol in formal terms. The first issue is how to define steerability.  By analogy with dissipative quantum state engineering~\cite{verstraete2009quantum}, we note that any (passive) steering protocol towards a (multi-dimensional) degenerate subspace comprises two stages: (i) steering an arbitrary initial state towards a desired, multidimensional target {\it space} $\mathcal{H}_{\mathrm{target}}$; (ii) keeping the running state within the target space. Formally speaking, we have the following definition.
\begin{definition}
    \label{def:stabilizability}
    A linear subspace $\mathcal{H}_{\mathrm{target}}\subseteq \mathcal{H}$ is called \textit{steerable subspace}, if there exists an infinite sequence of steering operators, represented as \textit{local} non-unitary superoperators $\{\mathcal{P}_1,\mathcal{P}_2\cdots\}$, such that one of the two following cases is true. Let us also denote a later time state $\rho_n(|\psi\rangle \langle \psi|)= \mathcal{P}_n (\mathcal{P}_{n-1}(\cdots \mathcal{P}_1(|\psi\rangle\langle \psi|)))$.
\begin{enumerate}
  \item[(1a)] For every initial state $|\psi\rangle$ and a sufficiently large $N$, $\rho_N(|\psi\rangle\langle\psi|)$ is arbitrarily close to  a (mixed or pure) state $\sigma\in \mathcal{H}_{\mathrm{target}}$, and $\forall n>N$, $\mathcal{P}_n(\sigma)=\sigma$ is invariant. All possible $\sigma$ spans the steerable subspace $\mathcal{H}_{\mathrm{target}}$.  
  \item[(1b)] For every initial state $|\psi\rangle$ and sufficient large $N$, $\rho_N(|\psi\rangle\langle\psi|)$ is arbitrarily close to  a state $\sigma\in \mathcal{H}_{\mathrm{target}}$, and $\forall n>N$ either $\mathcal{P}_n(\sigma)=\sigma$ or $\mathcal{P}_n(\sigma)=\gamma\neq\sigma$, $\gamma\in \mathcal{H}_{\mathrm{target}}$.  In other words, any initial state is steered towards $\mathcal{H}_{\mathrm{target}}$, then, within $\mathcal{H}_{\mathrm{target}}$, does not converge to a fixed state.
\end{enumerate}
\end{definition}
Specifically, we will consider $\mathcal{H}_{\mathrm{target}}$ to be the GS manifold $\mathcal{H}_{\mathrm{GS}}$ of a system's Hamiltonian $H$. We note that $H$ itself is \textit{not} part of the dynamics of the system, and is employed here only for the sake of defining  $\mathcal{H}_{\mathrm{target}}$. The dynamics here is solely due to local steering superoperators, that is only due to the system-detector Hamiltonian $H_{\mathrm{s-d}}$. More specifically:
(i) for case (1a),  as far as the steadily steerable  (FFS/NFFSS) classes are concerned, once the initial state $|\psi\rangle$ has reached a state within the target space represented by the density matrix  $\rho_{\mathrm{target}}$, the system is then kept stationary;  $\rho_{\mathrm{target}}$ is thus a steady state of our protocol, cf. Fig.~\ref{fig:Stabilizing}a. (ii) for the case (1b), the NFFJS class, approaching a state $\rho_{\mathrm{target}}$ within the target space still entails a dynamical evolution: there exists a subset of states in  $\mathcal{H}_{\mathrm{target}}$, for each of them, there is at least one steering operator, for which this state is not stationary, and therefore, under the action of steering operators, will be mapped onto another state within the target manifold $\mathcal{H}_{\mathrm{target}}$  (hence a "jumpy" behavior), cf. Fig. ~\ref{fig:Stabilizing}b.

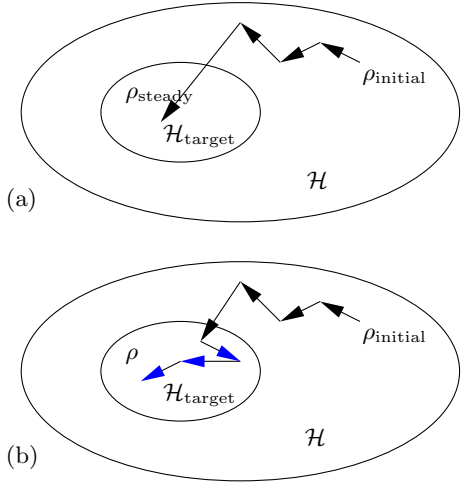
\begin{figure}[htbp]
    \centering
    \begin{tikzpicture}[x=0.75pt,y=0.75pt,yscale=-1,xscale=1]

\draw   (40,75) .. controls (40,44.62) and (89.25,20) .. (150,20) .. controls (210.75,20) and (260,44.62) .. (260,75) .. controls (260,105.38) and (210.75,130) .. (150,130) .. controls (89.25,130) and (40,105.38) .. (40,75) -- cycle ;
\draw   (80,75) .. controls (80,61.19) and (97.91,50) .. (120,50) .. controls (142.09,50) and (160,61.19) .. (160,75) .. controls (160,88.81) and (142.09,100) .. (120,100) .. controls (97.91,100) and (80,88.81) .. (80,75) -- cycle ;
\draw    (210,50) -- (191.79,40.89) ;
\draw [shift={(190,40)}, rotate = 26.57] [fill={rgb, 255:red, 0; green, 0; blue, 0 }  ][line width=0.08]  [draw opacity=0] (12,-3) -- (0,0) -- (12,3) -- cycle    ;
\draw    (150,30) -- (111.25,78.44) ;
\draw [shift={(110,80)}, rotate = 308.66] [fill={rgb, 255:red, 0; green, 0; blue, 0 }  ][line width=0.08]  [draw opacity=0] (12,-3) -- (0,0) -- (12,3) -- cycle    ;
\draw    (170,50) -- (151.41,31.41) ;
\draw [shift={(150,30)}, rotate = 45] [fill={rgb, 255:red, 0; green, 0; blue, 0 }  ][line width=0.08]  [draw opacity=0] (12,-3) -- (0,0) -- (12,3) -- cycle    ;
\draw    (190,40) -- (171.79,49.11) ;
\draw [shift={(170,50)}, rotate = 333.43] [fill={rgb, 255:red, 0; green, 0; blue, 0 }  ][line width=0.08]  [draw opacity=0] (12,-3) -- (0,0) -- (12,3) -- cycle    ;
\draw   (40,205) .. controls (40,174.62) and (89.25,150) .. (150,150) .. controls (210.75,150) and (260,174.62) .. (260,205) .. controls (260,235.38) and (210.75,260) .. (150,260) .. controls (89.25,260) and (40,235.38) .. (40,205) -- cycle ;
\draw   (80,205) .. controls (80,191.19) and (97.91,180) .. (120,180) .. controls (142.09,180) and (160,191.19) .. (160,205) .. controls (160,218.81) and (142.09,230) .. (120,230) .. controls (97.91,230) and (80,218.81) .. (80,205) -- cycle ;
\draw    (210,180) -- (191.79,170.89) ;
\draw [shift={(190,170)}, rotate = 26.57] [fill={rgb, 255:red, 0; green, 0; blue, 0 }  ][line width=0.08]  [draw opacity=0] (12,-3) -- (0,0) -- (12,3) -- cycle    ;
\draw    (150,160) -- (131.11,188.34) ;
\draw [shift={(130,190)}, rotate = 303.69] [fill={rgb, 255:red, 0; green, 0; blue, 0 }  ][line width=0.08]  [draw opacity=0] (12,-3) -- (0,0) -- (12,3) -- cycle    ;
\draw    (170,180) -- (151.41,161.41) ;
\draw [shift={(150,160)}, rotate = 45] [fill={rgb, 255:red, 0; green, 0; blue, 0 }  ][line width=0.08]  [draw opacity=0] (12,-3) -- (0,0) -- (12,3) -- cycle    ;
\draw    (150,200) -- (122,200) ;
\draw [shift={(120,200)}, rotate = 360] [fill={rgb, 255:red, 0; green, 0; blue, 255 }  ][line width=0.08]  [draw opacity=0] (12,-3) -- (0,0) -- (12,3) -- cycle    ;
\draw    (120,200) -- (101.79,209.11) ;
\draw [shift={(100,210)}, rotate = 333.43] [fill={rgb, 255:red, 0; green, 0; blue, 255 }  ][line width=0.08]  [draw opacity=0] (12,-3) -- (0,0) -- (12,3) -- cycle    ;
\draw    (130,190) -- (148.21,199.11) ;
\draw [shift={(150,200)}, rotate = 206.57] [fill={rgb, 255:red, 0; green, 0; blue, 255 }  ][line width=0.08]  [draw opacity=0] (12,-3) -- (0,0) -- (12,3) -- cycle    ;
\draw    (190,170) -- (171.79,179.11) ;
\draw [shift={(170,180)}, rotate = 333.43] [fill={rgb, 255:red, 0; green, 0; blue, 0 }  ][line width=0.08]  [draw opacity=0] (12,-3) -- (0,0) -- (12,3) -- cycle    ;

\draw (182,103.4) node [anchor=north west][inner sep=0.75pt]    {$\mathcal{H}$};
\draw (110,80) node [anchor=north west][inner sep=0.75pt]    {$\mathcal{H}_\mathrm{target}$};
\draw (210,52.4) node [anchor=north west][inner sep=0.75pt]    {$\rho_\mathrm{initial} $};
\draw (181,232.4) node [anchor=north west][inner sep=0.75pt]    {$\mathcal{H}$};
\draw (110,210) node [anchor=north west][inner sep=0.75pt]    {$\mathcal{H}_\mathrm{target}$};
\draw (31,112) node [anchor=north west][inner sep=0.75pt]   [align=left] {(a)};
\draw (31,241) node [anchor=north west][inner sep=0.75pt]   [align=left] {(b)};
\draw (91,61.4) node [anchor=north west][inner sep=0.75pt]    {$\rho _\mathrm{steady}$};
\draw (91,192.4) node [anchor=north west][inner sep=0.75pt]    {$\rho $};
\draw (210,182.4) node [anchor=north west][inner sep=0.75pt]    {$\rho_\mathrm{initial} $};

\end{tikzpicture}
    \caption{ Steering trajectories of two situations in Def.~\ref{def:stabilizability}. (a) Steering towards the GS manifold of an FFS/NFFSS Hamiltonian. Once it is mapped onto $\mathcal{H}_{\mathrm{target}}$, the system's states remain stationary. (b) Steering employing towards the GS manifold of an NFFJS Hamiltonian. Even within $\mathcal{H}_{\mathrm{target}}$, the system's state keeps "jumping" around.}
    \label{fig:Stabilizing}
\end{figure}

Note that for the latter case there exists at least one GS, that is invariant for each  $\mathcal{P}_n$. To show this we first note that each superoperator has at least one fixed point~\cite{Watrous_2018}.  We then argue that if all  the fixed points are not in $\mathcal{H}_{\mathrm{target}}$, states in $\mathcal{H}_{\mathrm{target}}$ may leak towards those exterior fixed points, rendering the probability to remain within $\mathcal{H}_{\mathrm{target}}$ less than 1, in contradiction to the definition of steerable subspace. 

As suggested by the names of our steerability classes, the notion of a \textit{local Frustration-Free} Hamiltonian is of central importance, and relies on the existence of a local \textit{Parent Hamiltonian}~\cite{parentHamiltonian}. Let us elaborate on these notions. First, a Hamiltonian $H$ is called a \textit{local Hamiltonian} if it can be expressed as a sum over a set of local operators, i.e.,  $H = \sum_i H_i$,  where each $H_i$ operates non-trivially only on a local fragment of the entire system~\cite{kempe2006complexity}. Second, given a local Hamiltonian $H = \sum_i H_i$, it is referred to as \textit{Frustration-Free} if \textit{all} local operators $H_i$ can be simultaneously minimized in the GS manifold of $H$. In other words, for any GS of $H$, denoted as $|\psi_{\mathrm{GS}}\rangle$, the global ground-state energy is the sum of the local ground-state energies, i.e.,
\begin{equation}
\label{eq:ff def}
\langle \psi_{\mathrm{GS}} | H | \psi_{\mathrm{GS}} \rangle = \inf_{|\phi\rangle} \langle \phi | H | \phi \rangle = \sum_i \inf_{|\phi\rangle} \langle \phi | H_i | \phi \rangle.
\end{equation}
Thus, $|\psi_{\mathrm{GS}}\rangle$ is the common GS of every local term of the Hamiltonian $H_i \otimes_{k} \mathbbm{1}_k$, such that this operator is proportional to the identity outside the finite range of $H_i$. Thirdly, given a set of target quantum states $\{|\psi_i\rangle\}$, we can define the corresponding parent Hamiltonian $H_{\mathrm{PH}}=\sum_i H_{\mathrm{PH},i}$ as the Hermitian operator, such that the GS manifold of $H_{\mathrm{PH}}$ is $\operatorname{span}\{|\psi_i\rangle\}$. 

Note that the notion of locality in the present analysis is crucial, both for local Hamiltonians and local steering operators. In principle, one can steer towards an arbitrary target state employing non-local ("global") steering operators. Physically this implies a detector (or detectors) that are coupled to an arbitrary non-local set of the system's degrees of freedom. In order to address a physically realizable protocol we confine ourselves here to local Hamiltonians.  Hereafter, our notion of \textit{locality} is associated with either one of the two: (i) an operator that acts on a subset of degrees of freedom, where this subset is defined by the (system size-independent) number of degrees of freedom coupled to a single steering degree of freedom (a single detector); or (ii) an operator acting on less than half of the system’s degrees of freedom. Without further clarification, we always refer to the first case. Also note that being a FF model does \textit{not} imply commutativity of the local terms, i.e., it is possible that $[H_i, H_j] \neq 0$.

In the following, the steerability of FFS, NFFSS and NFFJS classes will be discussed more formally.  

\section{Frustration-Free Steerable (FFS) Hamiltonians}\label{sec:Identifying an FFS Hamiltonian}

\subsection{Steerability of FF Hamiltonians}\label{sec:Steerability of FF Hamiltonian}

We next discuss and demonstrate the steerability for FFS Hamiltonians (class $A$ of Fig.~\ref{fig:classfication})  In the following, we consider a system whose dynamics is governed by a Hamiltonian $H$. An FFS Hamiltonian is defined when $H = \sum_i H_i$ is both local and frustration-free. As discussed in Sec.~\ref{sec:Definition of stabilizability}, it implies that the GS of $H$, denoted as $|\psi_{\text{GS}}\rangle$, is the common ground state of each local term $H_i$ of the Hamiltonian.

If a Hamiltonian $H$ is FFS, then steering towards its  GS manifold can be achieved in two ways~\cite{ticozzi2012stabilizing,verstraete2009quantum}. The first method involves the continuum time limit of generalized weak measurements, equivalent to coupling to an engineered dissipative bath. The system's dynamics is then effectively governed by a set of Lindblad operators $\{L_i\}$, obeying Lindblad dynamics. We can choose a suitable steering protocol such that the right-invariant space of $L_i$ equals the local GS subspace of $H_i$. In the long-time limit, the system converges to the common invariant space of all $\{L_i\}$, which corresponds to the common GS manifold of $\{H_i\}$, and thus to the GS manifold of $H$.  Alternatively, one may design a set of local steering operators $\{\mathcal{P}_i\}$ such that the invariant subspace of $\mathcal{P}_i$ equals the local GS manifold of $H_i$. After sufficiently many discrete applications of these steering operators, the only remaining invariant states are the GSs of $H$.

\subsection{Examples of FFS Hamiltonians}\label{sec:Examples of FFS Hamiltonian}

Let us demonstrate the idea of FFS Hamiltonian with two examples. 

(i) \textbf{\textit{  1d ferromagnetic Ising chain}}. Its Hamiltonian can be represented as $H=\sum_i H_i=-\sum_i J\sigma_i^z\sigma_{i+1}^z$($J>0$). Its two-fold GS manifold is computed to be $\{|+_1+_2+_3\cdots +_N\rangle, |-_1-_2-_3\cdots -_N\rangle\}$, where $\sigma^z_i|\pm_i\rangle=\pm|\pm_i\rangle$. Direct evaluation reveals that the GSs are also GSs of each local term, i.e. $-J\sigma^z_i\sigma^z_{i+1}$. Hence $H$ is an FFS Hamiltonian. 

(ii) \textbf{\textit{Affleck-Kennedy-Lieb-Tasaki(AKLT)  model}}~\cite{PhysRevLett.59.799AKLT}. Its Hamiltonian is defined on a one-dimensional spin-1 chain, as $H=\sum_i \boldsymbol{S}_i\cdot\boldsymbol{S}_{i+1}+ 1/3(\boldsymbol{S}_i\cdot\boldsymbol{S}_{i+1})^2$, where $\boldsymbol{S}_i$ are the spin operators for $i$th spin. The GS of each local term $\boldsymbol{S}_i\cdot\boldsymbol{S}_{i+1}+ 1/3(\boldsymbol{S}_i\cdot\boldsymbol{S}_{i+1})^2$ composed of spin-0,1 subspace for the composition of spin $i$ and spin $i+1$. The valence-bond construction of the Hamiltonian's GSs explicitly shows that the GSs are constrained in the spin-0,1 subspace, and hence $H$ is FFS Hamiltonian. It has been explicitly shown that the AKLT model is steerable~\cite{Measurement-inducedsteeringPhysRevResearch.2.033347}.

\section{Non-Frustration-Free Steadily Steerable Hamiltonians}\label{sec:Non-Frustration-Free Steadily Steerable Hamiltonians}

\subsection{Necessary and Sufficient conditions for Non-Frustration-Free Steadily Steerable (NFFSS) Hamiltonians}\label{sec:Necessary and Sufficient conditions for Non-Frustration-Free Steadily Steerable (NFFSS) Hamiltonian}

We next identify the necessary and sufficient conditions for NFFSS Hamiltonians. The latter corresponds to class $B$, cf. Fig.~\ref{fig:classfication}. NFFSS Hamiltonians are  Hamiltonians whose GS is steady steerable (i.e.,it can be designed to be a stationary state of a steering protocol), cf. Def.~\ref{def:stabilizability}; yet, at the same time are not of the FFS class, as they cannot be written as a sum of local Hamiltonians that share the same GS (one can find, though, a parent Hamiltonian that is FF, cf. PHFF, Def.~\ref{def:pHFF} below). Consider a system whose dynamics is governed by a Hamiltonian, $H$. We then introduce a local or global projection operator, $\Pi$, whose image defines a target subspace (onto which we project states), e.g. ground state(s) subspace of the Hamiltonian, $\mathcal{H}_{\mathrm{GS}}$. One may then construct a class of \textit{local} operators $\{\mathcal{A}_i\}$, which is denoted a {\it subspace conserved quantity} (SCQ).

\begin{definition}[Subspace Conserved Quantity (SCQ)]
\label{def:SCQ}
    Assuming that $[\Pi,H] = 0$, a non-unital local Hermitian operator $\mathcal{A}$ is a subspace conserved quantity in the subspace defined by $\Pi$, iff

\begin{subequations}\label{eq:SCQ}
\begin{align}
    &[\Pi,\mathcal{A}] = 0, \label{eq:SCQcond1}\\
    \Pi&[H,\mathcal{A}]\Pi = 0. \label{eq:SCQcond2}
\end{align}
\end{subequations}

\end{definition}

The trivial case of  $\Pi = \mathbbm{1}$ corresponds to the common definition of a conserved quantity $\mathcal{A}$.  The term \textit{subspace conserved quantity} refers to the fact that if an initial state is an eigenstate of $\Pi$ with eigenvalue 1 (in the image of $\Pi$), the expectation value $\langle \mathcal{A} \rangle$ is conserved under the unitary evolution with $H$, i.e. $e^{-iHt}$. We refer to $\mathcal{A}$ as trivial SCQ if  $\Pi \mathcal{A}\Pi\propto \Pi$. This means that  among the image of $\Pi$, i.e. the set of orthogonal basis $|\psi_i\rangle$ such that $\Pi|\psi_i\rangle = |\psi_i\rangle$, $\mathcal{A}$ has the same expectation value , i.e. $\langle \psi_{i}| \mathcal{A}|\psi_j\rangle =a\delta_{ij}$. For a trivial $\mathcal{A}$, one may simply replace $\mathcal{A}$ with $\Pi_\mathcal{A}$. The latter represents the projection operator onto the eigenspace of $\mathcal{A}$ containing the subspace defined by $\Pi$, i.e. $\Pi\times \Pi_\mathcal{A}\times \Pi=\Pi$.   Otherwise an SCQ $\mathcal{A}$ is denoted non-trivial. 

According to Prop.~\ref{prop:nontrivial SCQ} in Appendix~\ref{sec:Necessary conditions for Non-Frustration-Free Jittery Steerable (NFFJS) class}, the presence of a non-trivial SCQ always implies the presence of several trivial SCQs. We demonstrate this idea for an operator comprising two distinct eigenvalues (each of these may be degenerate). Generalization to more eigenvalues is straightforward.  For this purpose, consider a non-trivial SCQ operator decomposed as $\mathcal{A}=\sum_i a |a^i\rangle\langle a^i|+\sum_i b |b^i\rangle\langle b^i|=a\Pi_a + b\Pi_b$. Here $a, b$ are eigenvalues of $\mathcal{A}$, and $\{|a^i\rangle\}_i,\{|b^i\rangle\}_i$ and $\Pi_a,\Pi_b$ are the corresponding eigenvectors and projection operators. (The fact that O has more than one eigenvalues in the projection subspace defined by $\Pi$ implies automatically that we are dealing with a non-trivial SCQ). Then it is shown that $\Pi_a, \Pi_b$ are SCQs with respect to the subspace defined by $\Pi$. Furthermore, $\Pi_a$ is trivial with respect to either the subspace defined by $\Pi$ or the subspace defined by $\Pi \Pi_a \Pi$. In the latter case, the subspace projected by $\Pi \Pi_a \Pi$ is a subspace of the subspace defined by $\Pi$. Details are shown in the proof of Prop.~\ref{prop:nontrivial SCQ}.  Below we examine the relationship between $\{\mathcal{A}_i\}$, the set of SCQs, and the different steerability classes. 

To fully characterize the necessary and sufficient conditions for NFFSS Hamiltonians (cf. Fig.~\ref{fig:classfication} and Theorem~\ref{thm:FFS} below), we introduce the notion of an auxiliary local \textit{Parent Hamiltonian}, $H_{\mathrm{PH}} = \sum_i H_{\operatorname{PH},i}$, as presented in Sec.~\ref{sec:Definition of stabilizability}. For a given Hamiltonian $H$ and its GS manifold $\mathcal{H}_{\mathrm{GS}}$, we consider the parent Hamiltonian corresponding to $\mathcal{H}_{\mathrm{GS}}$.
While the original Hamiltonian $H$ does not need to be FF, achieving steerability requires the existence of an FF parent Hamiltonian. This requirement is articulated in Thm.~\ref{thm:FFS} below (see also Ref.~\cite{ticozzi2012stabilizing}). Therefore, for the convenience of our discussion, we introduce the notion of a Parent-Hamiltonian-Frustration-Free (PHFF).

\begin{definition}[Parent Hamiltonian Frustration-Free(PHFF)]
\label{def:pHFF}
    A local Hamiltonian is called Parent-Hamiltonian-Frustration-Free(PHFF), if the following conditions are satisfied:
    \begin{itemize}
        \item  There exists a local Parent Hamiltonian $H_{\mathrm{PH}}$, such that GS manifold of $H_{\mathrm{PH}}$ exactly \textit{equals} to the GS manifold of $H$.
        \item $H_{\mathrm{PH}}$ is Frustration-Free.
    \end{itemize}
\end{definition}
Based on this notion, let us introduce the necessary and sufficient conditions for FFS Hamiltonian in the following Theorem.
\begin{theorem}
\label{thm:FFS}
$H=\sum_i H_i$ is steadily steerable, i.e. $H$ is an FFS or a NFFSS Hamiltonian, if and only if there exists a local Hamiltonian $H_{\mathrm{PH}}$, such that $H$ is PHFF. Note $H_{\mathrm{PH}}$ can be constructed from the set of trivial SCQs in GS subspace.
\end{theorem}
In the following, we refer to a Hamiltonian $H$ as Frustration-Free if and only if it is PHFF. If a Hamiltonian $H$ is PHFF, then its GS manifold can also be steered in two ways, as discussed in Sec.~\ref{sec:Steerability of FF Hamiltonian}.

\subsection{Identifying Non-Frustration-Free Steadily Steerable Hamiltonians}\label{sec:Identifying Non-Frustration-Free Steadily Steerable Hamiltonians}

{\it How can one determine whether such a  FF parent Hamiltonian exists?} Imagine that we are given a pure target state or a target space, i.e. the GS manifold of $H$, and we need to determine whether the problem is or is not steerable. The construction of a corresponding parent Hamiltonian $H_{\mathrm{PH}}$ is intimately related to the notion of SCQ. Concretely, note that when the system is not frustrated, there exist  trivial SCQs with respect to the GS subspace. These can be used to construct  a parent Hamiltonian. Fo such a construction we take the following steps: To do this we take the following steps: 
(i)  We first compute a set of \textit{local} reduced density matrices (RDMs) of target state(s). Specifically, for each candidate steering operator, of a finite range (e.g., a steering operator that acts non-trivially on spins 1,2,3 in a long spin chain and otherwise, acting on other spins is represented by a unit operator) we compute the RDM of the corresponding support (the RDM of the space spanned by spins 1,2,3).  
(ii) We diagonalize the local RDM,  and note the local eigenvectors with non-vanishing eigenvalues, denoted as $|\phi_j^i\rangle$; we then  construct  the subspace spanned by these (non-vanishing eigenvalue) eigenvectors.  Here $i$ runs over the RDMs, while $j$ runs over the eigenvectors of a given RDM. 
(iii) We can then construct a set of local trivial SCQs denoted as $\{\Pi_i\}$. For every local RDM, $\Pi_i=\sum_j|\phi_j^i\rangle\langle\phi_j^i|$ is a projection operator. One readily verifies that $\Pi_i \rho \Pi_i=\rho$ for the RDM $i$.  
(iv) {\it Constructing the parent Hamiltonian}. We are now in a position to construct the parent Hamiltonian, $H_{\mathrm{PH}}=\sum_i (\mathbbm{1}-\Pi_i)$, and subsequently compute the GSs of $H_{\mathrm{PH}}$. Note that for each $\Pi_i$ and each GS $|\psi_{\mathrm{GS}}\rangle$, the above construction ensures $\Pi_i|\psi_{\mathrm{GS}}\rangle=|\psi_{\mathrm{GS}}\rangle$, implying $H_{\mathrm{PH}}|\psi_{\mathrm{GS}}\rangle=0|\psi_{\mathrm{GS}}\rangle$. One then concludes that $H_{\mathrm{PH}}$ is a parent Hamiltonian. 
(v) {\it The GS of $H_{\mathrm{PH}}$ is identical to the  GS of H}. Since target states are by construction GSs of $H_{\mathrm{PH}}$,  it follows that if there are no additional excited states of $H$ that are GSs of $H_{\mathrm{PH}}$, then Thm.~\ref{thm:FFS} is satisfied. This class of Hamiltonians is called either  FFS or NFFSS.  Otherwise $H$ is frustrated,  and is classified as either NFFJS or NFFNS.

We can also construct the parent Hamiltonian -- this time with a symmetry-motivated approach. Accounting for the set of symmetries  of the Hamiltonian, there exists a set of conserved quantities $\{\mathcal{A}_i\}$. Given a set of non-trivial  SCQs,  earlier in this Section we show that this implies a trivial set of  SCQs,  $\{\Pi_{\mathcal{A}_i}\}$,  for the given non-trivial SCQs $\{\mathcal{A}_i\}$. We can then construct the parent Hamiltonian as $H_{\mathrm{PH}}=\sum_i (\mathbbm{1}-\Pi_{\mathcal{A}_i})$. The latter satisfies $\Pi_{\mathcal{A}_i}|\psi_{\mathrm{GS}}\rangle=|\psi_{\mathrm{GS}}\rangle$ for any GS $|\psi_{\mathrm{GS}}\rangle$, which further implies $H_{\mathrm{PH}}|\psi_{\mathrm{GS}}\rangle=0|\psi_{\mathrm{GS}}\rangle$. Since $H_{\mathrm{PH}}\ge 0$, it follows that $|\psi_{\mathrm{GS}}\rangle$ is GS of $H_{\mathrm{PH}}$. Consequently, $H_{\mathrm{PH}}$ is FF.  We next note that the GS of $H_{\mathrm{PH}}$ may be degenerate, and the dimension of this GS manifold may be larger than that of the GS manifold of the original $H$. To render H steerable we must be able to find additional SCQs that will further reduce the dimensionality of the GS manifold of $H_{\mathrm{PH}}$. Once the latter is equal to the GS manifold of H, the problem is steerable. Examples are discussed below. 

\subsection{Examples of NFFSS Hamiltonians}\label{sec:Examples of NFFSS Hamiltonian}

Here we present an example of a Hamiltonian of the NFFSS class:

\textbf{\textit{Honeycomb Kitaev model}}~\cite{kitaev2006anyons}. As a final example, the Honeycomb Kitaev model (class B of Fig.~\ref{fig:classfication}) is defined on a two-dimensional spin-1/2 honeycomb lattice. Through conserved quantities, it is exactly solvable.  One can then engineer a parent Hamiltonian employing the conserved quantities of the Hamiltonian and the related SCQs. Given the set of conserved quantities $\{W_i\}$, the total Hilbert space is divided into numerous sectors according to the eigenvalues of $\{W_i\}$. In each of the sector, the original Hamiltonian can be rewritten as a non-interacting majorana-fermionic Hamiltonian, hence defines a full set of SCQs.  This  model represents an example of a Hamiltonian which is  {\it not} FF, yet one can construct a frustration-free parent Hamiltonian  which satisfies the conditions for steerability (cf. Appendix~\ref{sec:Parent Hamiltonian and FF Steering}).

\section{Necessary conditions for Non-Frustration-Free Jittery Steerable (NFFJS) Hamiltonians}\label{sec:Necessary conditions for NFFJS Hamiltonian}
In the following we present a general  study of NFFJS Hamiltonians, supplemented by a set of specific examples. Formulating the necessary conditions for the NFFJS class allows us to delineate the boundary between classes $B$ and  $\overline{B+C}$ on the one hand, and $C$  on the other hand (for the latter these conditions are not satisfied), cf. Fig.~\ref{fig:classfication}.  The term "jittery"  refers to the fact that for the NFFJS class it is not possible to keep each and every state within the GS manifold invariant under the steering protocol (unlike the case for FFS). This implies that at least a subset of states within this manifold are transformed to others under the steering superoperators. This characterization is also employed in the derivation of the necessary conditions for NFFJS, cf. Appendix~\ref{sec:Necessary conditions for Non-Frustration-Free Jittery Steerable (NFFJS) class}.
The necessary conditions are then:

\begin{widetext}
\begin{subequations}
\label{cond}
\begin{align}
    &\text{CONDITION I: The GS manifold $\mathcal{H}_{\mathrm{GS}}$ must be \textit{degenerate}.} \label{cond:1}\\
    &\text{CONDITION II: There exists at least one \textit{trivial SCQ } (cf.  Def.~\ref{def:SCQ}) within a subspace of the }\nonumber\\
    &\text{ground state manifold, } \mathcal{H}_{\mathrm{GS}}. \label{cond:2}\\
    &\text{CONDITION III: } \mathcal{H}_{\mathrm{GS}} \text{ is \textit{bipartite indistinguishable}.} \label{cond:3}
\end{align}
\end{subequations}
\end{widetext}

 The term in the last condition is readily defined. The GS subspace $\mathcal{H}_{\mathrm{GS}}\subseteq\mathcal{H}$  is called \textit{bipartite indistinguishable}, if there exists at least one bipartite decomposition of the total Hilbert space $\mathcal{H}=\mathcal{H}_S\otimes\mathcal{H}_{\bar S}$, and at leas two mixed ground states $\rho\neq\tau\in\mathcal{L}(\mathcal{H}_{\mathrm{GS}})$ ($\mathcal{L}(\mathcal{H}_{\mathrm{GS}})$ are linear operators over the ground state space), such that $\operatorname{Tr}_S\rho= \operatorname{Tr}_S\tau$, cf. Appendix~\ref{sec:Necessary conditions for Non-Frustration-Free Jittery Steerable (NFFJS) class}. It is termed ”bipartite distinguishable“ since we can distinguish the two globally different states by looking at their reduced density matrix for the given bipartition, i.e. $\operatorname{Tr}_S \rho$ and   $\operatorname{Tr}_S \tau$. 

Interestingly, the mere existence of a non-trivial SCQ $\mathcal{A}$ in the subspace $\mathcal{H}_{\mathrm{GS}}$, also implies that the three above Conditions~\ref{cond} are satisfied, cf. Prop.~\ref{prop:nontrivial SCQ} in Appendix~\ref{sec:Necessary conditions for Non-Frustration-Free Jittery Steerable (NFFJS) class}.
It follows that if the GS manifold of an NFFJS Hamiltonian possesses a non-trivial SCQ, the above necessary conditions are satisfied. We also note that the second Condition~\ref{cond:2} and the third Condition~\ref{cond:3} are mutually independent, cf. Prop.~\ref{prop:independent} in Appendix~\ref{Discussions about steerability}. However, to obtain necessary and sufficient conditions for the steerability of NFFJS Hamiltonians, one needs to exploit global information (such as the global wavefunction, instead of only local RDMs), which is beyond our scope here, cf. Appendix~\ref{Discussions about steerability}. 

\section{Non-Frustration-Free Jittery Steerable (NFFJS) Hamiltonians: Examples}\label{sec:NFFJS Examples}
In the following we will introduce a family of NFFJS systems that, while being frustrated, are nevertheless steerable, cf.  class $B$ in Fig.~\ref{fig:classfication}. To keep it simple we begin with a family of local commuting Pauli Hamiltonians~\cite{haah2013commuting}, which plays an important role as a template for quantum error correction code~\cite{stabilizerPhysRevLett.95.230504,graph}, for phases with topological orders~\cite{kitaev2003fault}, and for realizations of geometric frustration~\cite{blote1993antiferromagnetic}. Note that commutativity does not necessarily imply FF, since there could be geometric frustration, an example would be an Ising anti-ferromagnet on a 2-dimensional triangular  lattice.

We begin with a 1-dimensional system of $N$ qubits, with Hamiltonians of the form $H=\sum_i H_i$. Here each $H_i$ can be written as a tensor product of a ($N$-independent) number of Pauli matrices $\pm\sigma^{x,y,z}$ and the identity operator $\mathbbm{1}$, such that $[H_i,H_j]=0$. The deep reason for the steerability of this family of Hamiltonians is the following: indeed, $H$ is frustrated, defining a multi-dimensional manifold of GSs. It turns out that each pair of these states, see $|m\pm\rangle$ below, is in the GS manifold of a different parent Hamiltonian $H_{\operatorname{PH},m}$, the latter being FF. Let us illustrate this idea with an example. Consider the one-dimensional anti-ferromagnetic Ising model $H=\sum_i H_i=\sum_i J\sigma^z_i\sigma^z_{i+1}, J=1$ with periodic boundary condition and with $N$ odd spins. The GSs of it can be denoted as $|m\pm\rangle=|\cdots\pm_{m-2}\mp_{m-1} \pm_{m}\pm_{m+1}\mp_{m+2}\pm_{m+3}\cdots\rangle$, where $\sigma^z_i|\pm_i\rangle=\pm|\pm_i\rangle$. 

The target GS manifold here is multidimensional ($2N$ in the present example). We first show  that steering towards this manifold based on FFS steering protocol (generalized (weak) measurement) is impossible.  Since $H$ itself is not FF, one may try to search for a FF parent Hamiltonian $H_{\mathrm{PH}}$ instead.  However, it is not possible to construct any two-qubit SCQ (other than the identity operator), hence there exits no non-trivial two-qubit parent Hamiltonian. It turns out (cf. Appendix~\ref{sec:Parent Hamiltonian and FF Steering}), that in order to construct a parent Hamiltonian whose GS is identical to that of the original Hamiltonian, one needs  $n$-qubit terms, with $n$ scaling up with the system's size.
Hence this GS manifold is intractable with the  FFS steering protocol.

A way out from this seemingly dead end is to allow for a scenario where the steered states are not stabilized, even asymptotically. For this we invoke a sequence  of projective measurements, i.e. discrete-time superoperators. Such steering dynamics cannot be cast as a single Lindbladian in the continuum-time limit. Importantly, our proposed protocol will not only steer the system towards the desired manifold, but will also dynamically "jump"  the state within that manifold. 

Coming back to our example,  note that  $|m\pm\rangle$ are only two GSs of the parent FF Hamiltonian  $H^m_{\mathrm{PH}}=\sum_{i\neq m}J\sigma^z_i\sigma^z_{i+1}-J\sigma^z_m\sigma^z_{m+1}$. Following this line of thought, the GS manifold of $H$ can be separated into a sum of GS of FF parent Hamiltonian, $H_{\mathrm{PH}}^m$. The latter, then, are not frustrated. The concrete steering protocol consists of a set of $N$ different superoperators acting on the density matrix $\rho$:
\begin{equation}
\label{eq:supop}
\mathcal{P}^{(i)}(\rho)=\Pi^{(i)}\rho\Pi^{(i)}+\Pi^{(i)}V^{(i)} \rho (V^{(i)})^\dagger \Pi^{(i)},
\end{equation}
where $\Pi^{(i)}=(\mathbbm{1}- H_i)/2$ and $V^{(i)}=\sigma^x_i$, the latter flips the spin $i$. Note that these superoperators cool the system locally ~\cite{Langbehn2023}. This is easy to see in the Heisenberg picture, where $\mathcal{P}^{(i)\dagger}(H_{i-1})=\mathcal{P}^{(i)\dagger}(H_{i})=-\mathbbm{1}$ with all initial states brought to the local GS of $H_{i-1}$ and $H_{i}$, while other $\{H_j\}$ remain invariant. Therefore, a common invariant subspace of the set of superoperators $\{\mathcal{P}^{(i)}\}_{i\neq m,m+1}$ is $\operatorname{span}\{|m+\rangle,|m-\rangle\}$, the GSs of $H_{\mathrm{PH}}^m$. It follows that the set of $\mathcal{P}^{(i)}$ dynamically steers towards the GS manifold of $H$. Yet, any given state within this subspace does not remain invariant under the application of all $\{\mathcal{P}^{(i)}\}$\footnote{This can be verified by noting the effect of $\mathcal{P}_-^{(i)}$ in Heisenberg picture. }. Note that since the $\mathcal{P}^{(i)}$ only acts on spin $i$ and spin $i+1$, $\{\mathcal{P}^{(i)}\}$ are local superoperators.  

The above example can be generalized to an arbitrary commuting Pauli Hamiltonian made up of local commuting Pauli operators. The resulting steering superoperators  will all have a form similar to Eq.~\eqref{eq:supop}. These local Pauli operators are made up of only $k$ different qubits (a "$k$-local" Hamiltonian) and render the entire Hamiltonian translationally invariant.  In particular, it is possible to show that the "flip operators" (the analogues of $V_i$ of Eq.~\eqref{eq:supop}), which facilitate switching from one GS to another, can  be decomposed into \textit{local} flip operators, cf. Fig.~\ref{fig:Stabilizing}b and Appendix~\ref{sec:Locality of superoperator}. In this sense, we only need to apply steering protocol to a constant number of qubits. The steerability of commuting Pauli Hamiltonian is established in Thm.~\ref{thm:CPH} in Appendix~\ref{app:Proof for Stabilizability of Commuting Pauli Hamiltonian}.

The superoperators of the form similar to Eq.~\eqref{eq:supop} for commuting Pauli Hamiltonians can be realized through "system-detector" interaction terms and an auxiliary qubit. The former are represented by Clifford unitary operators, and the latter is reinitialized in the state $|0\rangle$ at each measurement step. Following each measurement step, we trace over the detector's readouts~\cite{Measurement-inducedsteeringPhysRevResearch.2.033347}. The unitary evolution associated with each measurement step may involve a single (or multiple) spin flip(s). 
This is seemingly a classical operation: following the Gottesman–Knill Theorem ~\cite{gottesman1998heisenberg}, such "cooling" of the system's state (i.e. getting closer to the GS manifold) can be efficiently simulated with classical computer if we are only concerned with correlation functions measured on the computational basis $|0\rangle$ and $|1\rangle$. We stress, though, that the problem is inherently quantum. To see this, one notes that in order to reach a ground state, we have chosen the superoperators similar to Eq.~\eqref{eq:supop} only for simplicity; there can be other superoperators that manifestly do not involve only ”classical” spin flipping, cf. Appendix~\ref{app:Alternative superoperators}. Furthermore, once a ground state is reached, one may perform measurements on non-computational bases (i.e., employing superpositions of $|0\rangle$ and $|1\rangle$)  to implement measurement-based quantum  computation~\cite{MBQCPhysRevA.68.022312}, thus going beyond classical simulation.
Interestingly, one may wonder why we do not directly prepare every single GS with Clifford unitary quantum circuits. The reason is that employing superoperators in the presence of environmental noise is more stable than using unitary evolution, since superoperators can eliminate undesired excitations, cf. Appendix~\ref{app:Stability against heating}.

Intriguingly, our analysis may be pushed  beyond the class of commuting Pauli Hamiltonians.  One may then  consider a special class of non-FF Hamiltonians, and  focus on a Hamiltonian, $H$, that possesses a multitude of GSs. Following this line of thought, one  further assumes that each of these GSs is also a GS of a different FF parent Hamiltonian, $\{H^n_{\mathrm{PH}}, n=1,2,\cdots\}$. This implies that it is possible to steer towards one of these GSs. If further GSs of different parent Hamiltonian $H^n_{\mathrm{PH}}$ can be connected by local operations, and our steering protocol facilitates ”visiting”  a multitude (or all)  GSs of $H$, then $H$ is NFFJS. 

We next consider  a class where the above steering protocol as well as  FFS steering are not facilitated.

\section{Non-Frustration-Free Non-Steerable (NFFNS) Hamiltonians: surrogate states and a "glass floor"}\label{sec:NFFNS Hamiltonian}

Following our characterization of various classes of Hamiltonians, whose GS manifolds are accessible through local steering operations, cf. Sec.~\ref{sec:Identifying an FFS Hamiltonian}, ~\ref{sec:Non-Frustration-Free Steadily Steerable Hamiltonians} and ~\ref{sec:Necessary conditions for NFFJS Hamiltonian}, here we address yet another class, that of Non-Frustration-Free Non-Steerable (NFFNS) Hamiltonians. The GSs of these Hamiltonians are not accessible employing any local steering operations. The insistence on \textit{local steering} operators throughout the present study follows the observation that non-local steering, possibly involving products of numerous local operators, renders controllability of this dynamics practically impossible. We conjecture that the difficulty observed in quantum simulation experiments ~\cite{mazurenko2017coldFermiHubbard,xu2023frustration,li2024observation,shao2024antiferromagnetic,Xu2025} to cool the relevant systems to low-energy states is a manifestation of this non-steerability~\cite{Langbehn2023,PhysRevResearch.6.013244}. The latter is due to the high degree of non-local entanglement  of the GS that cannot be created with local steering.

\subsection{Approximate steering and surrogate states}\label{sec: Approximate steering and surrogate states}

Our emphasis here shifts to yet another paradigm: defining protocols that drive the system to a state close to the target of the GS, but not insisting on reaching exactly the target of the GS. The hope is that even with this imperfect cooling (but attaining a state that has a finite overlap with the GS), certain non-trivial properties of the GS can be observed.

To characterize the degree of non-steerability, one can use three possible criteria to quantify the distance between the GS and the “limit of steerability” (i.e., how close we can get to the GS) : (i) evaluate the distance between target state $|\psi_{\mathrm{target}}\rangle$ (or, more generally,  $\rho_{\mathrm{target}}$) and the optimal reachable states, which is denoted as a \textit{surrogate state}, $\rho_{\mathrm{surr}}$. This distance may be evaluated in several optional ways: the \textit{fidelity} of quantum states,  
\begin{equation}
\label{eq:fid def}
    \mathcal{F}(\rho_{\mathrm{target}},\rho_{\mathrm{surr}})=\left(\operatorname{Tr}\sqrt{\sqrt{\rho_{\mathrm{surr}}}\rho_{\mathrm{target}}\sqrt{\rho_{\mathrm{surr}}}}\right)^2,
\end{equation}
or Frobenius norm, $D_F(\rho_{\mathrm{surr}},\rho_{\mathrm{target}})=\sqrt{\operatorname{Tr}(\rho_{\mathrm{surr}}-\rho_{\mathrm{target}})^2}$ and trace norm $D_t(\rho_{\mathrm{surr}},\rho_{\mathrm{target}})=\operatorname{Tr}[(\rho_{\mathrm{surr}}-\rho_{\mathrm{target}})^2]/2$.  
(ii) compute the system’s energy expectation value, $\langle E\rangle \equiv \operatorname{Tr} (\rho H)$,  with the asymptotic state. For  $(\langle  E\rangle -E_{\mathrm{GS}}) / \operatorname{gap}[H] \ll 1$ ($E_{\mathrm{GS}}$ is the energy of the ground state and $\operatorname{gap}[H]$ is the the energy gap of the first excitation above the ground state of $H$), overlap with the target GS is  large. (iii) Employing the basis of the Hamiltonian eigenstates,  approximate the asymptotic state  by an ansatz of Gibbs ensemble, optimizing the effective temperature, cf. Sec.~\ref{sec:Derivation of the Temperature lower bound}. The dimensionless  ratio of this temperature and $\operatorname{gap}[H]$ represents the distance from the true GS.  

The plan for Sec.~\ref{sec:NFFNS Hamiltonian} is to first address the criterion  (i) of the distance between the asymptotic steerable state, $\rho_{\mathrm{surr}}$,  and the target state, $\rho_{\mathrm{target}}$.  The lower bound of this distance serves as a figure of merit to the efficiency of cooling-by-steering. Following this analysis, and attempting to employ definitions based on directly measurable quantities, we next employ criterion (ii) and (iii).

Consider as an example a target ground state which is an $N$-qubit GHZ state, underlined by highly non-local entanglement, $|\psi_{\mathrm{target}}\rangle=(|00\cdots 0\rangle+|11\cdots 1\rangle)/\sqrt{2}$, and assume that our steering is restricted to single-qubit operations. With this constraint, all accessible steerable states are product states;  the best fidelity of such a steerable state (w.r.t.  the GHZ state)  is given by $|\langle \psi_{\mathrm{target}}|00\cdots 0\rangle|^2=1/2$. This upper bound on the fidelity implies a finite distance between the non-steerable target state and any steerable state. The state $|00\cdots 0\rangle$   (as well as $|11\cdots 1\rangle$) is, in this trivial example, a surrogate state. It represents the best approximation to the true GS achievable through single-qubit blind steering. 

We now describe the surrogate steering above in more precise terms (cf. Def.~\ref{def:stabilizability}).
\begin{definition}
    \label{def:approx_stabilizability}
    A linear subspace $\mathcal{H}_{\mathrm{target}}\subseteq \mathcal{H}$ is called a  \textit{$p$-approximate steerable subspace}, if there exists a positive number $p<1$ and an infinite sequence of local non-unitary superoperators $\{\mathcal{P}_1,\mathcal{P}_2\cdots\}$, such that for arbitrary initial state $|\psi\rangle$, the probability of later-time state being in target subspace $\mathcal{H}_{\mathrm{target}}$, which can be computed as follows,
$$
\operatorname{Tr}[ \mathcal{P}_N(\mathcal{P}_{N-1}(\cdots\mathcal{P}_1(|\psi\rangle\langle\psi|))) \Pi_{\mathrm{target}}] \ge p,
$$
for sufficiently large $N$. Here $\Pi_{\mathrm{target}}$ is the projection operator for $\mathcal{H}_{\mathrm{target}}$. We also consider this subspace to be minimal. In other words, in the asymptotic limit any initial state has at least a population of $p$ in the target space.
\end{definition}
For a further discussion we refer to  Appendix~\ref{Discussions about non-steerable case}.
In the following, our general goal is to derive the distance between a non-steerable target state, $\rho_{\mathrm{target}}$, and the corresponding surrogate state, given by $\mathcal{F}_{\mathrm{max}}=\max_{\rho} \mathcal{F}(\rho_{\mathrm{target}},\rho)$, cf. Eq.~\eqref{eq:fid def}. This establishes $\mathcal{H}_{\mathrm{GS}}$ as, at best, a $\mathcal{F}_{\mathrm{max}}$-approximate steerable subspace. The steerable state $\rho$ that optimizes the fidelity with the non-steerable target state $\rho_{\mathrm{target}}$ is the surrogate state $\rho_{\mathrm{surr}}$. Specifically,  we have fully characterized the Hamiltonian classes FFS and NFFSS, cf. Sec.~\ref{sec:Identifying an FFS Hamiltonian} and ~\ref{sec:Non-Frustration-Free Steadily Steerable Hamiltonians}, and specified the necessary conditions for NFFJS Hamiltonians, cf. Sec.~\ref{sec:Necessary conditions for NFFJS Hamiltonian}. It then turns out that the present class of Hamiltonians, violating the necessary steerability conditions (cf. Conditions~\ref{cond}), is non-steerable, i.e. class $C$ (Fig.~\ref{fig:classfication}). However, since we do not have the capacity to classify states in class $\overline{C+D}$ (Fig.~\ref{fig:classfication}) as either NFFJS or NFFNS, we are unable to achieve full characterization of potential surrogate states. We are thus unable to directly compute $\mathcal{F}_{\mathrm{max}}$. 
Our solution to this difficulty is to define a {\it presumed surrogate states}, $\tilde{\rho}_{\mathrm{surr}}$, out of a set of states, that may or may not be steerable, and compute its distance from the target states. The presumed surrogate states is defined to be the state closest to $\rho_{\mathrm{target}}$ in this set. Specifically, we will derive the distance $d_2$ between a non-steerable target state $\rho_{\mathrm{target}}$, cf. class $D$  (Fig.~\ref{fig:classfication}), and a presumed surrogate state, $\tilde{\rho}_{\mathrm{surr}}$, that is closest to $\rho_{\mathrm{target}}$ in the complementary set of class $D$, cf. Fig.~\ref{fig:distances}.  Therefore, this distance $d_2$, represented by the fidelity with presumed surrogate state, namely
$$
\tilde F_{\mathrm{max}} = \max_{\rho \not \in D}\mathcal{F}(\rho,\rho_{\mathrm{target}}),
$$
serves as a lower bound to the distance $d_1$ between surrogate states and target states, represented by $\mathcal{F}_{\mathrm{max}}$, cf. Fig.~\ref{fig:distances}. In other words, for the distances we have $d_2=d(\tilde{\rho}_{\mathrm{surr}},\rho_{\mathrm{target}})\le d_1=d(\rho_{\mathrm{surr}},\rho_{\mathrm{target}})$, and for the fidelities we have 
\begin{equation}
\mathcal{F}_{\mathrm{max}}\le \tilde{\mathcal{F}}_{\mathrm{max}}.
\end{equation}  
Note that this presumed surrogate state might belong to class $A,B,C$, or it might belong to the class $\overline{C+D}$ in Fig.~\ref{fig:classfication}.

The rationale behind the introduction of $\tilde{\mathcal{F}}_{\mathrm{max}}$ is two-fold: (i) We have not managed to find necessary {\it and}  sufficient conditions for the entire class NFFJS (which includes the class $C$ in Fig.~\ref{fig:classfication}), hence this class is not fully characterized. It follows that the optimized distance between a non-steerable state in $D$ and the closest state which is unquestionably steerable, might involve unknown NFFJS state (Fig.~\ref{fig:distances}), hence that distance is not clearly determined. (ii) Even if classes $A,B$ and $C$ are fully characterized, optimizing over these steerable states for the distance to the target state is a computationally hard problem. This is because the mere determination of Frustration-Free property of a given Hamiltonian is already computationally hard~\cite{bravyi2006efficient,gosset2016quantum}.

\begin{figure}
    \centering
    \begin{tikzpicture}[x=0.75pt,y=0.75pt,yscale=-1,xscale=1]

\draw [color={rgb, 255:red, 208; green, 2; blue, 27 }  ,draw opacity=1 ][line width=3]    (100,100) -- (135,100) ;
\draw [color={rgb, 255:red, 74; green, 144; blue, 226 }  ,draw opacity=1 ][line width=0.75]    (100,98.5) -- (135,98.5)(100,101.5) -- (135,101.5) ;
\draw [color={rgb, 255:red, 74; green, 144; blue, 226 }  ,draw opacity=1 ][line width=1.5]    (135,65) -- (135,40) ;
\draw [color={rgb, 255:red, 208; green, 2; blue, 27 }  ,draw opacity=1 ][line width=2.25]    (100,65) -- (135,65) ;
\draw   (20,48) .. controls (20,32.54) and (32.54,20) .. (48,20) -- (162,20) .. controls (177.46,20) and (190,32.54) .. (190,48) -- (190,132) .. controls (190,147.46) and (177.46,160) .. (162,160) -- (48,160) .. controls (32.54,160) and (20,147.46) .. (20,132) -- cycle ;
\draw   (40,90) .. controls (40,84.48) and (44.48,80) .. (50,80) -- (90,80) .. controls (95.52,80) and (100,84.48) .. (100,90) -- (100,120) .. controls (100,125.52) and (95.52,130) .. (90,130) -- (50,130) .. controls (44.48,130) and (40,125.52) .. (40,120) -- cycle ;
\draw   (100,54) .. controls (100,46.27) and (106.27,40) .. (114,40) -- (156,40) .. controls (163.73,40) and (170,46.27) .. (170,54) -- (170,117) .. controls (170,124.73) and (163.73,131) .. (156,131) -- (114,131) .. controls (106.27,131) and (100,124.73) .. (100,117) -- cycle ;
\draw    (100,0) .. controls (101.67,1.67) and (101.67,3.33) .. (100,5) .. controls (98.33,6.67) and (98.33,8.33) .. (100,10) .. controls (101.67,11.67) and (101.67,13.33) .. (100,15) .. controls (98.33,16.67) and (98.33,18.33) .. (100,20) .. controls (101.67,21.67) and (101.67,23.33) .. (100,25) .. controls (98.33,26.67) and (98.33,28.33) .. (100,30) .. controls (101.67,31.67) and (101.67,33.33) .. (100,35) .. controls (98.33,36.67) and (98.33,38.33) .. (100,40) .. controls (101.67,41.67) and (101.67,43.33) .. (100,45) .. controls (98.33,46.67) and (98.33,48.33) .. (100,50) .. controls (101.67,51.67) and (101.67,53.33) .. (100,55) .. controls (98.33,56.67) and (98.33,58.33) .. (100,60) .. controls (101.67,61.67) and (101.67,63.33) .. (100,65) .. controls (98.33,66.67) and (98.33,68.33) .. (100,70) .. controls (101.67,71.67) and (101.67,73.33) .. (100,75) .. controls (98.33,76.67) and (98.33,78.33) .. (100,80) .. controls (101.67,81.67) and (101.67,83.33) .. (100,85) .. controls (98.33,86.67) and (98.33,88.33) .. (100,90) .. controls (101.67,91.67) and (101.67,93.33) .. (100,95) .. controls (98.33,96.67) and (98.33,98.33) .. (100,100) .. controls (101.67,101.67) and (101.67,103.33) .. (100,105) .. controls (98.33,106.67) and (98.33,108.33) .. (100,110) .. controls (101.67,111.67) and (101.67,113.33) .. (100,115) .. controls (98.33,116.67) and (98.33,118.33) .. (100,120) .. controls (101.67,121.67) and (101.67,123.33) .. (100,125) .. controls (98.33,126.67) and (98.33,128.33) .. (100,130) .. controls (101.67,131.67) and (101.67,133.33) .. (100,135) .. controls (98.33,136.67) and (98.33,138.33) .. (100,140) .. controls (101.67,141.67) and (101.67,143.33) .. (100,145) .. controls (98.33,146.67) and (98.33,148.33) .. (100,150) .. controls (101.67,151.67) and (101.67,153.33) .. (100,155) .. controls (98.33,156.67) and (98.33,158.33) .. (100,160) .. controls (101.67,161.67) and (101.67,163.33) .. (100,165) .. controls (98.33,166.67) and (98.33,168.33) .. (100,170) .. controls (101.67,171.67) and (101.67,173.33) .. (100,175) .. controls (98.33,176.67) and (98.33,178.33) .. (100,180) -- (100,180) ;
\draw  [color={rgb, 255:red, 0; green, 0; blue, 0 }  ,draw opacity=1 ][fill={rgb, 255:red, 0; green, 0; blue, 0 }  ,fill opacity=1 ] (130,40) .. controls (130,37.24) and (132.24,35) .. (135,35) .. controls (137.76,35) and (140,37.24) .. (140,40) .. controls (140,42.76) and (137.76,45) .. (135,45) .. controls (132.24,45) and (130,42.76) .. (130,40) -- cycle ;
\draw  [color={rgb, 255:red, 0; green, 0; blue, 0 }  ,draw opacity=1 ][fill={rgb, 255:red, 0; green, 0; blue, 0 }  ,fill opacity=1 ] (130,65) .. controls (130,62.24) and (132.24,60) .. (135,60) .. controls (137.76,60) and (140,62.24) .. (140,65) .. controls (140,67.76) and (137.76,70) .. (135,70) .. controls (132.24,70) and (130,67.76) .. (130,65) -- cycle ;
\draw  [color={rgb, 255:red, 0; green, 0; blue, 0 }  ,draw opacity=1 ][fill={rgb, 255:red, 0; green, 0; blue, 0 }  ,fill opacity=1 ] (95,65) .. controls (95,62.24) and (97.24,60) .. (100,60) .. controls (102.76,60) and (105,62.24) .. (105,65) .. controls (105,67.76) and (102.76,70) .. (100,70) .. controls (97.24,70) and (95,67.76) .. (95,65) -- cycle ;
\draw  [color={rgb, 255:red, 0; green, 0; blue, 0 }  ,draw opacity=1 ][fill={rgb, 255:red, 0; green, 0; blue, 0 }  ,fill opacity=1 ] (130,100) .. controls (130,97.24) and (132.24,95) .. (135,95) .. controls (137.76,95) and (140,97.24) .. (140,100) .. controls (140,102.76) and (137.76,105) .. (135,105) .. controls (132.24,105) and (130,102.76) .. (130,100) -- cycle ;
\draw  [color={rgb, 255:red, 0; green, 0; blue, 0 }  ,draw opacity=1 ][fill={rgb, 255:red, 0; green, 0; blue, 0 }  ,fill opacity=1 ] (95,100) .. controls (95,97.24) and (97.24,95) .. (100,95) .. controls (102.76,95) and (105,97.24) .. (105,100) .. controls (105,102.76) and (102.76,105) .. (100,105) .. controls (97.24,105) and (95,102.76) .. (95,100) -- cycle ;

\draw (139,22.4) node [anchor=north west][inner sep=0.75pt]  [font=\large]  {$\tilde{\rho }_{\text{surr}}$};
\draw (142,68.4) node [anchor=north west][inner sep=0.75pt]  [font=\large]  {$\rho _{\text{tar}}$};
\draw (59,52.4) node [anchor=north west][inner sep=0.75pt]  [font=\large]  {$\rho _{\text{surr}}$};
\draw (145,92.4) node [anchor=north west][inner sep=0.75pt]  [font=\large]  {$\sigma _{\text{tar}}$};
\draw (59,92.4) node [anchor=north west][inner sep=0.75pt]  [font=\large]  {$\sigma _{\text{surr}}$};
\draw (43,112.4) node [anchor=north west][inner sep=0.75pt]  [font=\large]  {$A,B,C$};
\draw (131,112.4) node [anchor=north west][inner sep=0.75pt]  [font=\large]  {$D$};
\draw (111,138.4) node [anchor=north west][inner sep=0.75pt]  [font=\large]  {$\overline{C+D}$};
\draw (109,73.4) node [anchor=north west][inner sep=0.75pt]  [font=\large]  {$d_{1}$};
\draw (142,46.4) node [anchor=north west][inner sep=0.75pt]  [font=\large]  {$d_{2}$};
\draw (21,2) node [anchor=north west][inner sep=0.75pt]  [font=\large] [align=left] {Steerable};
\draw (111,2) node [anchor=north west][inner sep=0.75pt]  [font=\large] [align=left] {Non-steerable};

\end{tikzpicture}
    \caption{Schematics of the distances from target states to surrogate states. $\rho_{tar}$ and $\sigma_{tar}$ are two non-steerable target states defined in the class $D$, cf. Fig.~\ref{fig:classfication}. For $\rho_{tar}$, the presumed surrogate state, $\tilde{\rho}_{\mathrm{surr}}$, lies on the boundary of $D$, and is the closest \textit{possibly} steerable state to $\rho_{tar}$, while $\rho_{\mathrm{surr}}$ is a surrogate state, that is steerable and closest to $\rho_{tar}$. For target states like $\sigma_{tar}$, these two states coincide and are represented by $\sigma_{\mathrm{surr}}$. Red lines indicate the distance $d_1$ between target states and steerable surrogate states. Blue lines indicate the distance $d_2$ between target states and the boundary of class $D$, which is given by the bound $1-p(\Pi_{\mathrm{GS}})$, cf. Eq.~\eqref{eq:upper bound} and Eq.~\eqref{eq:upper bound fidelity}. The latter distance is a lower bound of the former, i.e. $d_1\ge d_2$. The zigzag indicates the boundary between steerable and non-steerable states.}
    \label{fig:distances}
\end{figure}

\subsection{A bound on the distance between the surrogate state and the target state employing local measure}\label{sec:A bound on the distance between the surrogate state and the target state employing local measure}

Dealing now with the set of presumed surrogate states, we are now in a position to compute a lower bound to the distance from the target state. We first derive the (optimal) presumed surrogate state based on the non-degenerate non-steerable target state as in Eq.~\eqref{eq:presumed state}, then we derive the upper bound of the fidelity between the target state and surrogate state, given in Eq.~\eqref{eq:upper bound}. This bound is then generalized to the degenerate case in Eq.~\eqref{eq:upper bound fidelity}.

As stated in the previous section, there are three measures of such a distance.  We stress that these measures employ global features (considering states of the entire system; overlaps  between full system density matrices, etc.). This is computationally challenging. To facilitate computationally feasible protocols, we next define a distance based on local quantities. We consider, for simplicity, a non-degenerate case, i.e., the non-steerable target state is a single GS manifold $|\psi_{\mathrm{GS}}\rangle$. Non-Steerability is due to the fact that such a GS possesses a large degree of non-local entanglement, which cannot be constructed solely by local operations. Simply speaking, in a steerable case, there exists a sufficiently large set of local projection operators $\{\Pi_j\}$ satisfy that 
\begin{equation}    
\label{eq:projection}\Pi_j|\psi_{\mathrm{GS}}\rangle=|\psi_{\mathrm{GS}}\rangle,
\end{equation}
with sufficiently many $\{\Pi_j\}$ such that $|\psi_{\mathrm{GS}}\rangle$ is the only possible state satisfying Eq.~\eqref{eq:projection}. This guarantees that our steering protocol converges only towards the given GS (or the GS manifold in the degenerate case). Here we address a scenario where not even a single  local trivial SCQ  $\Pi_j$ acting on the GS exists which satisfies Eq.~\eqref{eq:projection}. In this case, the necessary Condition~\ref{cond:2} does not hold.
Within this framework, there is a clear boundary between the sector of non-steerable states ($D$ in Fig.~\ref{fig:distances}), and the presumed surrogate states (the rest of the Hilbert space in Fig.~\ref{fig:distances}). The lower bound on the distance, $d_2$, is indeed smaller than the actual distance, $d_1$,  between a target state in $D$ and a true steerable surrogate state.

Now suppose we are given this non-steerable target state $|\psi_{\mathrm{GS}}\rangle$, and we are trying to design a local steering operator. Let us consider this local steering operator acts non-trivially on a local region $S$, e.g. composing of qubit $1$ and $2$, and the complementary space is denoted as $\bar S$. Given this bipartition, we can implement Schmidt decomposition of $|\psi_{\mathrm{GS}}\rangle$ as follows: 
\begin{equation}
\label{eq:schmidt decomposition}
|\psi_{\mathrm{GS}}\rangle=\sum_i c_i |\alpha_i\rangle_S\otimes |\beta_i\rangle_{\bar S},
\end{equation}
where $\{|\alpha_i\rangle_S\}$ is a set of  bases defined on $\mathcal{H}_S$, the Hilbert space on $S$,  $\{|\beta_i\rangle_{\bar S}\}$ is a set of  bases defined on the complementary Hilbert space $\mathcal{H}_{\bar S}$, and $c_i\neq 0$. Suppose $|c_0|^2$ is the smallest among all coefficients.
The condition that there does not exist any trivial SCQ for the GS (cf. Condition~\ref{cond:2}) implies that $\{|\alpha_i\rangle\}$ is a complete basis for $\mathcal{H}_S$, meaning there does not exist a non-trivial local projector $\Pi$ such that $\Pi|\psi_{\mathrm{GS}}\rangle=|\psi_{\mathrm{GS}}\rangle$. 

Let us first consider the construction of a pure presumed surrogate state $|\tilde{\psi}_{\mathrm{surr}}\rangle$. On the contrary, $|\tilde{\psi}_{\mathrm{surr}}\rangle$ does not violate Condition~\ref{cond:2}, thus it must support at least one trivial SCQ $\tilde{\Pi}^{\mathrm{surr}}$, implying that $\tilde{\Pi}^{\mathrm{surr}}|\tilde{\psi}_{\mathrm{surr}}\rangle=|\tilde{\psi}_{\mathrm{surr}}\rangle$. Without loss of generality, we consider the support of this $\tilde{\Pi}^{\mathrm{surr}}$ is also the region $S$.  Then given the two conditions: (i) $\tilde{\Pi}^{\mathrm{surr}}|\tilde{\psi}_{\mathrm{surr}}\rangle=|\tilde{\psi}_{\mathrm{surr}}\rangle$; (ii) $|\tilde{\psi}_{\mathrm{surr}}\rangle$ is as close to $|\psi_{\mathrm{GS}}\rangle$ as possible, the seemingly best choice is to project out the most irrelevant Schmidt component, namely $c_0|\alpha_0\rangle_S\otimes |\beta_i\rangle_{\bar{S}}$. The presumed surrogate state can then be defined as 
\begin{equation}
\label{eq:presumed state}
    |\tilde{\psi}_{\mathrm{surr}}\rangle=\sum_{i\neq 0} \frac{c_i}{\sqrt{1-|c_0|^2} }|\alpha_i\rangle_S\otimes |\beta_i\rangle_{\bar S},
\end{equation} 
with the trivial SCQ $\tilde{\Pi}^{\mathrm{surr}}=\mathbbm{1}-|\alpha_0\rangle\langle \alpha_0|$. We can then quantify the distance between target state and presumed surrogate state as follows:
\begin{align}
\label{eq:fidelity non-deg GS}
        \tilde{\mathcal{F}}(|\psi_{\mathrm{GS}}\rangle,|\tilde{\psi}_{\mathrm{surr}}\rangle)&=|\langle \psi_{\mathrm{GS}}|\tilde{\psi}_{\mathrm{surr}}\rangle|^2\nonumber\\
        &=\Bigg|\sum_{i\neq 0}\frac{|c_i|^2}{\sqrt{1-|c_0|^2}}\Bigg|^2\nonumber\\
        &=1-|c_0|^2.
\end{align}
This fidelity serves as a finite distance between $|\psi_{\mathrm{GS}}\rangle$ and $|\tilde{\psi}_{\mathrm{surr}}\rangle$, and it is calculated through the Schmidt spectrum of the target state $|\psi_{\mathrm{GS}}\rangle$. Equivalently, $|c_0|^2$ is the smallest eigenvalue of the reduced density matrix of $|\psi_{\mathrm{GS}}\rangle$ by tracing over $\mathcal{H}_{\bar S}$, i.e. $\rho^S_{\mathrm{GS}}=\operatorname{Tr}_{\mathcal{H}_{\bar S}} |\psi_{\mathrm{GS}}\rangle\langle \psi_{\mathrm{GS}}|=\sum_i |c_i|^2 |\alpha_i\rangle_S{}_S\langle \alpha_i|$. Also note that $\tilde{\Pi}^{\mathrm{surr}}$ can be equivalently obtained by maximizing $\langle \psi_{\mathrm{GS}}|\tilde{\Pi}^{\mathrm{surr}}|\psi_{\mathrm{GS}}\rangle$ for all non-unit projection operator $\tilde{\Pi}^{\mathrm{surr}}$. Note that the value of the fidelity upper bound of $\mathcal{F}_{\mathrm{max}}$, i.e. $1-|c_0|^2$ (Eq.~\eqref{eq:fidelity non-deg GS}), depends entirely on the way we bipartition the entire system into $S$ and $\bar{S}$ (cf. Eq.~\eqref{eq:schmidt decomposition}). Since $S$ is assumed to be the subsystem we may control and manipulate, it should represent the interaction range of the steering operator (the detector-system coupling). Applying this bipartition (with a given range of $S$, say 2 neighboring sites) to different locations in the system,  $\mathcal{F}_{\mathrm{max}}$ is upper bounded by the smallest value of $1-|c_0|^2$.

More generally, let us consider computing the fidelity between a mixed presumed surrogate state and the tarrget state. In principle, the one  closest  to $|\psi_{\mathrm{GS}}\rangle$ may be a mixed state, and  is denoted as  $\tilde{\rho}_{\mathrm{surr}}$. Since the manifold of presumed surrogate states is required to satisfy Condition~\ref{cond:2}, there exists a set of local trivial SCQs $\{\tilde{\Pi}^{\mathrm{surr}}_i\}$, such that the support of each $\tilde{\Pi}^{\mathrm{surr}}_i$ is defined on a given local region $S_i$. By definition of SCQ (cf. Def.~\ref{def:SCQ}), we have $\tilde{\Pi}^{\mathrm{surr}}_i \tilde{\rho}_{\mathrm{surr}}\tilde{\Pi}^{\mathrm{surr}}_i=\tilde{\rho}_{\mathrm{surr}}$. Hence the distance between the presumed surrogate target state $\tilde{\rho}_{\mathrm{surr}}$  and the GS $|\psi_{\mathrm{GS}}\rangle$ is bounded by the \textit{fidelity}
\begin{align}
\label{eq:upper bound}
&\mathcal{F}(|\psi_{\mathrm{GS}}\rangle,\tilde{\rho}_{\mathrm{surr}})\nonumber\\
  =&\langle\psi_{\mathrm{GS}}|\tilde{\rho}_{\mathrm{surr}}|\psi_{\mathrm{GS}}\rangle \nonumber\\
  =&\operatorname{Tr} |\psi_{\mathrm{GS}}\rangle\langle \psi_{\mathrm{GS}}| \tilde{\rho}_{\mathrm{surr}}\nonumber\\
  =&\operatorname{Tr} (\tilde{\Pi}^{\mathrm{surr}}_i|\psi_{\mathrm{GS}}\rangle\langle \psi_{\mathrm{GS}}|\tilde{\Pi}^{\mathrm{surr}}_i) \tilde{\rho}_{\mathrm{surr}}\nonumber\\
  \leq &1-p(|\psi_{\mathrm{GS}}\rangle),
\end{align}
where $p(|\psi_{\mathrm{GS}}\rangle)$ is the minimum of the set of smallest eigenvalues of  local reduced density matrices of $|\psi_{\mathrm{GS}}\rangle$ computed over all local region $S_i$.  The last inequality holds since $\tilde{\Pi}^{\mathrm{surr}}_i|\psi_{\mathrm{GS}}\rangle \langle\psi_{\mathrm{GS}}|\tilde{\Pi}^{\mathrm{surr}}_i$ is a semi-positive definite operator with trace $\langle \psi_{\mathrm{GS}}|\tilde{\Pi}^{\mathrm{surr}}_i|\psi_{\mathrm{GS}}\rangle\leq 1-p(|\psi_{\mathrm{GS}}\rangle)$, as discussed above.  This indicates how close $|\psi_{\mathrm{GS}}\rangle$ is to $\tilde{\rho}_{\mathrm{surr}}$, cf. Fig.~\ref{fig:distances}. Note that distances here involve local RDMs. Let us denote the complementary set of support of $\tilde{\Pi}^{\mathrm{surr}}_i$ as $S_i^c$; we thus calculate the distance between local reduced presumed surrogate state $\operatorname{Tr}_{\bar S_i}\tilde{\rho}_{\mathrm{surr}}$ and the local reduced target state $\operatorname{Tr}_{\bar S_i}|\psi_{\mathrm{GS}}\rangle\langle \psi_{\mathrm{GS}}|$, quantified by the expectation value of local operator $\tilde{\Pi}^{\mathrm{surr}}_i$. Our measure for the distance is thus \textit{local}.

Next, let us generalize this bound to a degenerate non-steerable GS manifold. In that case (and, again, assuming the necessary Condition ~\ref{cond:2} is violated), a lower bound for the distance from the degenerate GS manifold can be derived as well, denoted as $p(\Pi_{\mathrm{GS}})$. Similarly, for a mixed presumed surrogate state $\tilde{\rho}_{\mathrm{surr}}$, there exists a set of local trivial SCQs $\{\tilde{\Pi}^{\mathrm{surr}}_i\}$, such that $\tilde{\Pi}^{\mathrm{surr}}_i \tilde{\rho}_{\mathrm{surr}}\tilde{\Pi}^{\mathrm{surr}}_i=\tilde{\rho}_{\mathrm{surr}}$.
The probability for a presumed surrogate state $\tilde{\rho}_{\mathrm{surr}}$ to stay in the GS manifold is upper bounded by
\begin{align}
\label{eq: upper bound degenerate}
    \operatorname{Tr}\tilde{\rho}_{\mathrm{surr}} \Pi_{\mathrm{GS}} &= \operatorname{Tr}\tilde{\Pi}^{\mathrm{surr}}_i\tilde{\rho}_{\mathrm{surr}}\tilde{\Pi}^{\mathrm{surr}}_i \Pi_{\mathrm{GS}}\nonumber\\
    &=\operatorname{Tr}\tilde{\rho}_{\mathrm{surr}}\tilde{\Pi}^{\mathrm{surr}}_i \Pi_{\mathrm{GS}}\tilde{\Pi}^{\mathrm{surr}}_i\nonumber\\
    & \leq 1-p(\Pi_{\mathrm{GS}}),
\end{align}
where $1-p(\Pi_{\mathrm{GS}})=\sup_{\tilde{\Pi}^{\mathrm{surr}}_i}\lambda(\tilde{\Pi}^{\mathrm{surr}}_i \Pi_{\mathrm{GS}}\tilde{\Pi}^{\mathrm{surr}}_i)$ is maximized over all possible projection operators $\{\tilde{\Pi}^{\mathrm{surr}}_i\}$ defined on all possible local regions $S_i$, and $\lambda(M)$ is the largest eigenvalue of the operator $M$. The last inequality in Eq.~\eqref{eq: upper bound degenerate} holds since $\tilde{\Pi}^{\mathrm{surr}}_i \Pi_{\mathrm{GS}}\tilde{\Pi}^{\mathrm{surr}}_i$ is a Hermitian operator with spectrum no larger then $1-p(\Pi_{\mathrm{GS}})$. Since there is no trivial SCQ in any subspace of GS manifold, there exists no GS $|\psi_{\mathrm{GS}}\rangle$ satisfying $\tilde{\Pi}^{\mathrm{surr}}_i|\psi_{\mathrm{GS}}\rangle\langle \psi_{\mathrm{GS}}|\tilde{\Pi}^{\mathrm{surr}}_i=|\psi_{\mathrm{GS}}\rangle\langle\psi_{\mathrm{GS}}|$. In other words, there does not exist an eigenstate of $\tilde{\Pi}^{\mathrm{surr}}_i \Pi_{\mathrm{GS}}\tilde{\Pi}^{\mathrm{surr}}_i$ with eigenvalue $1$, hence the above inequality is strictly less than $1$. If we further take $\eta\in \mathcal{H}_{\mathrm{GS}}$ as am arbitrary GS of $H$, then the fidelity of presumed surrogate state $\tilde{\rho}_{\mathrm{surr}}$ with respect to GS $\eta$ is
\begin{align}
    \mathcal{F}(\eta,\tilde{\rho}_{\mathrm{surr}})&=\mathcal{F}(\eta,\tilde{\Pi}^{\mathrm{surr}}\tilde{\rho}_{\mathrm{surr}}\tilde{\Pi}^{\mathrm{surr}})\nonumber\\
    &=\mathcal{F}(\tilde{\Pi}^{\mathrm{surr}}\eta\tilde{\Pi}^{\mathrm{surr}},\tilde{\rho_{\mathrm{surr}}})\nonumber\\
    &\leq \operatorname{Tr}\tilde{\Pi}^{\mathrm{surr}}\eta\tilde{\Pi}^{\mathrm{surr}}\nonumber\\
    & \leq 1-p(\Pi_{\mathrm{GS}}),
    \label{eq:upper bound fidelity}
\end{align}
where we have used the fact that $\mathcal{F}(\rho,\sigma)\leq 1$ for $\rho,\sigma$ with trace 1, and $\tilde{\Pi}^{\mathrm{surr}}\eta\tilde{\Pi}^{\mathrm{surr}}$ might not be normalized, i.e. $\operatorname{Tr}\tilde{\Pi}^{\mathrm{surr}}\eta\tilde{\Pi}^{\mathrm{surr}}<1$. Note that the upper bound of Eq.~\eqref{eq:upper bound fidelity} (the r.h.s of the inequality) usually appears as inequality. This is due to the fact that the $\tilde{\rho}_{\mathrm{surr}}$ is associated with the class $\overline{C+D}$, hence it may or may not be steerable, cf. Sec.~\ref{sec: Approximate steering and surrogate states} and Fig.~\ref{fig:distances}. We note that, in general, if we consider a specific state in the degenerate GS manifold as target state, the upper bound on the fidelity will be lower than the one stated in Eq.~\eqref{eq:upper bound fidelity}. Conversely, if we consider combining multiple states as GS manifold, the upper bound on the fidelity may be higher. To demonstrate the latter statement, consider GHZ states $|GHZ_\pm\rangle=(|000\cdots 0\rangle\pm|111\cdots 1\rangle)/\sqrt{2}$. If only $|GHZ_+\rangle$ is considered, according to the discussion in the last section the fidelity is upper bounded by $1/2$. However, if we consider both $|GHZ_\pm\rangle$ to span the GS manifold, then the superposition of these two states, $|000\cdots 0\rangle$, is steerable, implying the upper bound on fidelity to be $1$.

Noteworthy, our derivation above invokes only the necessary condition for steerable states. This condition, relying on local distances (cf. discussion following Eq.~\eqref{eq:upper bound}), implies a finite distance between the "surrogate state" $\tilde{\rho}_{\mathrm{surr}}$ and the target state. We recall that the presumed surrogate state may be non-steerable, hence the calculated distance to the target state may serve only as a lower bound. Assume now that both states are pure. The non-degenerate target state is decomposed as $|\psi_{\mathrm{GS}}\rangle=\sum_i c_i |\alpha_i\rangle_S\otimes |\beta_i\rangle_{\bar S}$ and the presumed surrogate state is constructed as $|\tilde{\psi}_{\mathrm{surr}}\rangle=\sum_{i\neq 0} c_i/\sqrt{1-|c_0|^2} |\alpha_i\rangle_S\otimes |\beta_i\rangle_{\bar S}$, as discussed above. The global fidelity defined in Eq.~\eqref{eq:fidelity non-deg GS}, which may serve as a measure of the distance between the presumed surrogate state and the target state, also has its local counterpart. For the latter one employs local density matrix as follows:
\begin{align}
    &\mathcal{F}(\operatorname{Tr}_{{\bar S}} |\psi_{\mathrm{GS}}\rangle\langle \psi_{\mathrm{GS}}|,\operatorname{Tr}_{{\bar S}} |\tilde{\psi}_{\mathrm{surr}}\rangle\langle \tilde{\psi}_{\mathrm{surr}}|)\nonumber\\
    =&\Bigg(\operatorname{Tr}\sqrt{\left(\sum_{i\neq 0}|c_i|/\sqrt{1-|c_0|^2}|\alpha_i\rangle\langle \alpha_i|\right)}\nonumber\\
    &\times \sqrt{\left(\sum_{i\neq}|c_i|^2|\alpha_i\rangle\langle \alpha_i|\right)}\nonumber\\
    &\times \sqrt{\left(\sum_{i\neq 0}|c_i|/\sqrt{1-|c_0|^2}|\alpha_i\rangle\langle \alpha_i|\right)}\Bigg)^2\nonumber\\
    =&\left(\operatorname{Tr}\sqrt{\sum_{i\neq 0}\frac{|c_i|^4}{1-|c_0|^2}|\alpha_i\rangle\langle \alpha_i|}\right)^2\nonumber\\
    =&1-|c_0|^2,
\end{align}
which is exactly the same as the upper bound on the global fidelity, which is derived in Eq.~\eqref{eq:upper bound}.

Naively one could employ local operators to steer each system's segment towards the local RDM of the target state. Such a procedure raises two difficulties: first, steering operations in neighboring segments may be non-commuting, hence interfering with each other. Second, even if we achieve the correct local RDMs by direct projecting (which, in general, represent mixed states), we lose the coherence which underlie global entanglement. Importantly, for the case of a nondegenerate GS, the latter will always be excited by an arbitrary local non-trivial operation: such a GS {\it cannot} be stabilized by an arbitrary local operation. This implies  a finite distance between the apparent surrogate state  and $|\psi_{\mathrm{GS}}\rangle\langle \psi_{\mathrm{GS}}|$ (cf. Appendix~\ref{sec:Non-steerability for states with full-rank reduced states}).

So far, we have employed the notion of "global" vs. "local" to characterize our approach to assessing the distance from the target state. We note though that these terms, "local" and "global" might also appear in another context: the range of the steering operators. The latter usually comprises a product of a detector operator and a combination of system operators that span $N$ degrees of freedom. We now focus on this notion of "local" and "global" for steering system operators. To bridge between local and global steering, consider non-local operations that span an ever-increasing number of system's degrees of freedom, namely $3$-body steering operations, $4$-body steering operations, etc. As the degrees of freedom of control increases, the efficacy of steering should also increase. It is manifested by the upper bound $p(|\psi_{\mathrm{GS}}\rangle)$: when the size of the subsystem attached to the detector, $S$ (cf. Eq.~\eqref{eq:schmidt decomposition}) increases, generally the Schmidt coefficients decrease and hence $p(|\psi_{\mathrm{GS}}\rangle)$ decreases.
Particularly interesting is the limiting case, where a steering detector is coupled to $\lfloor N/2\rfloor+1$ degrees of freedom of the system. In this case, addressing the Schmidt decomposition of $|\psi_{\mathrm{GS}}\rangle$, and noting that the dimension of region $S$ is larger than its complement, $\bar S$, i.e. $dim \mathcal{H}_S>dim \mathcal{H}_{\bar S}$, the Schmidt decomposition is automatically not full-rank, implying that the smallest Schmidt coefficient is $0$, implying a trivial bound, i.e. the fidelity is not larger than 1. Since we are dealing here with a necessary condition for steerability, this trivial bound implies that the target state {\it may} be reached, but this is not guaranteed. 

In summary, the non-existence of a trivial SCQ for the ground state represents a sufficient condition for an NFFNS Hamiltonian, cf. Fig.~\ref{fig:classfication}.  This enables us to derive a lower bound on the distance between non-steerable GS and steerable states.  We stress that for the class of NFFNS Hamiltonians discussed here, local operations cannot cool these Hamiltonians to the GS. From a different perspective, in the following we will derive a lower bound for the system's temperature from the global measure (the "glass floor"), assigning a direct physical meaning to our limited ability to implement optimal steering.

\subsection{Lower bound on the energy}\label{sec:Lower bound on the energy}

In the derivation above, we have seen that there exists an upper bound on the population in the non-steerable (degenerate) GS manifold, $1- p(\Pi_{\mathrm{GS}})$, cf. Eq.~\eqref{eq:upper bound fidelity}. Equivalently, it implies a lower bound on the achievable energy of the system. In the best case, one would have a population of $1-p(\Pi_{\mathrm{GS}})$ in the GS manifold with energy $E_{\mathrm{GS}}$, and a population of $p(\Pi_{\mathrm{GS}})$ in the first excited state with energy $E_{\mathrm{GS}}+\operatorname{gap}[H]$ (here $\operatorname{gap}[H]$ is the energy gap of the spectrum of $H$, between the GS and the first excited state). In this situation, the lower bound on energy expectation value is calculated as follows:
\begin{align}
\label{eq: energy}
    \langle H\rangle \ge& (1-p(\Pi_{\mathrm{GS}}))E_{\mathrm{GS}}+p(\Pi_{\mathrm{GS}})(E_{\mathrm{GS}}+\operatorname{gap}[H])\nonumber\\
    =&E_{\mathrm{GS}}+p(\Pi_{\mathrm{GS}})\operatorname{gap}[H].
\end{align}
In the case of a degenerate GS manifold, we note that focusing on one specific state, our energy bound, in general, will be higher than the one stated in Eq.~\eqref{eq: energy}.  Equivalently, we can normalize this energy expectation value with the energy gap and obtain
\begin{equation}
    \frac{\langle H\rangle -E_{\mathrm{GS}}}{\operatorname{gap}[H]}\ge p(\Pi_{\mathrm{GS}}).
\end{equation}

\subsection{Derivation of the temperature lower bound}\label{sec:Derivation of the Temperature lower bound}

Above we have shown that for non-steerable states violating Condition~\ref{cond:2}, there exists an upper bound for the overlap between GS manifold and steerable states. This implies a bound on the weight of the steerable state in the GS manifold. We next use this weight to define (a bound on) the effective temperature of the presumed surrogate state, i.e. the "glass floor". 

Trying to assign a direct physical interpretation to the bound on the distance to the GS, we define an effective temperature of the presumed surrogate state of our steering protocol, $\tilde{\rho}_{\mathrm{surr}}$. Our working approximation is to replace the presumed surrogate state by a thermal state, with thermal Boltzmann weights of the low-lying states. By Eq.~\eqref{eq: upper bound degenerate} the weight of the GS in this thermal state does not exceed $1-p(\Pi_{\mathrm{GS}})$. This implies a lower bound on the temperature of the asymptotic running state, which is a steerable surrogate state, namely
\begin{align}
\label{eq:temp}
    1-p(\Pi_{\mathrm{GS}})&\geq \operatorname{Tr}(\tilde{\rho}_{\mathrm{surr}} \Pi_{\mathrm{GS}})\ge \operatorname{Tr}(\rho_{\mathrm{surr}} \Pi_{\mathrm{GS}})\nonumber\\
    &= \frac{\operatorname{deg}(H) e^{-\beta E_{\mathrm{GS}}}}{Z(\beta)},
\end{align}
where $\operatorname{deg}(H)$ is GS degeneracy of $H$, $\beta=1/T_{\mathrm{eff,min}}$, referring to the minimal effective temperature that can be reached through our steering protocol; $E_{\mathrm{GS}}$ is the GS energy of  $H$ and $Z(\beta)=\operatorname{Tr}e^{-\beta H}$ is the partition function. 
For systems with an additional symmetry (implying that both the system’s Hamiltonian and the steering dynamics satisfy this symmetry), the fit to Boltzmann weights may be replaced by modified Gibbs distributions. For example, imposing particle number conservation, the Boltzmann distribution should be replaced by the canonical partition function. Note that the so-defined effective temperature can be smaller than the gap to the system's lowest excited state, cf. the examples in Sec.~\ref{sec:Implementations of NFFNS model}. We also note that, in general, if we consider a specific state in the degenerate GS manifold, the lower bound on the temperature will be higher than the one stated in Eq.~\eqref{eq:temp}.

\section{Examples of the NFFNS class: the "glass floor"}\label{sec:Implementations of NFFNS model}

Our discussion of Non-Frustration-Free Non-Steerable (NFFNS) models has brought forward a couple of conceptually significant quantities: the emergence of a presumed surrogate  target state and  an upper bound on its overlap with the true ground state, and then an effective minimal temperature that signifies the distance from being at the the ground state (the "glass floor"). 
Furthermore, models underlain with randomness (SYK model) may be characterized by {\it distributions}  of these quantities. To demonstrate our ideas  we focus here on three paradigmatic models, which fall in the class of NFFNS Hamiltonians. Some technical details of the results presented here are discussed in the Appendix~\ref{sec:Implementation of Examples}. 

\subsection{Anti-ferromagnetic Heisenberg model}\label{sec:Anti-ferromagnetic Heisenberg model}

As a first example, we explore the \textit{\textbf{One-dimensional Anti-Ferromagnetic Heisenberg model}}. The Hamiltonian is defined as follows:
\begin{equation}
H_{\mathrm{AFH}} = \sum_i S_x^i S_x^{i+1} + S_y^i S_y^{i+1} + S_z^i S_z^{i+1},
\end{equation}
where $S_{x,y,z}^i$ represent the spin operators for the $i$-th spin, and we consider periodic boundary conditions. The GS manifold is 4-fold degenerate when the spin $s$ is a half-integer, and the number of spins is odd; otherwise, it is non-degenerate.

The non-steerability of this model can be understood as follows. Note that this model possesses a global $SO(3)$ symmetry, i.e., there exists a family of $U = \bigotimes_{i} U_i$, and $UHU^\dagger = H$, where each $U_i$ is the unitary $SO(3)$ rotation for the $i$-th spin. Consequently, for an SCQ (cf. Sec.~\ref{sec:Identifying an FFS Hamiltonian}, Def.~\ref{def:SCQ}) $\mathcal{A}$, $U \mathcal{A} U^\dagger$ is also an SCQ, since $\Pi_{\mathrm{GS}}[H, U\mathcal{A} U^\dagger] \Pi_{\mathrm{GS}} = U\Pi_{\mathrm{GS}}[H, \mathcal{A} ] \Pi_{\mathrm{GS}}U^\dagger = 0$ and $[\Pi, U\mathcal{A}U^\dagger] = U[\Pi, \mathcal{A}]U^\dagger = 0$, cf. Def.~\ref{def:SCQ}. Thus, if there exists an SCQ $\mathcal{A}$, it satisfies either (I) $U \mathcal{A} U^\dagger = \mathcal{A}$ for all $SO(3)$ unitary $U$, or (ii) $U \mathcal{A} U^\dagger \neq \mathcal{A}$ is another SCQ. The former case implies that $\mathcal{A}$, which is a local operator, commutes with the local 'total' angular momentum, e.g., $\boldsymbol{S}_1+\boldsymbol{S}_2$. We note that such commutation relation would imply that projected to the GS subspace, local total angular momenta commutes with the Hamiltonian. This would imply that the  GS become a product state, which is easily shown not to be the case. For the scenario (ii), one also expects the absence of local SCQ; otherwise, the presence of trivial SCQ like $\Pi\rho\Pi=0$, and hence the transformed ones $(U\Pi U^\dagger) \rho (U\Pi U^\dagger)=0$, implies the GS does not exist, i.e., $\rho=0$. For instance, if $|0\rangle_i {}_i\langle 0|$ is an SCQ for $i$th spin-1/2, then $S_x^i |0\rangle_i {}_i\langle 0|S_x^i=|1\rangle_i {}_i\langle 1|$ is also an SCQ. In this case, $|0\rangle_i {}_i\langle 0|\rho |0\rangle_i {}_i\langle 0|=|1\rangle_i {}_i\langle 1|\rho |1\rangle_i {}_i\langle 1|=0$ implies $\rho=0$. The absence of a SCQ implies the violation of Condition~\ref{cond:2}, hence it is non-steerable. The further discussion is presented in Appendix~\ref{sec:Non-steerability for states with full-rank reduced states}.

Based on these considerations, we proceed to calculate the upper bound, $1-p(\Pi_{\mathrm{GS}})$, on the fidelity between the family of states in the  non-steerable GS manifold, and the  family of the respective surrogate states (cf. Eq.~\eqref{eq:upper bound fidelity}).  Subsequently we are able to calculate the lower bound on the energy expectation value of surrogate states (cf. Sec.~\ref{sec:Lower bound on the energy}), and the effective dimensionless temperature $T_{\mathrm{eff,min}}/\operatorname{gap}[H_{\mathrm{AFH}}]$ (cf. Sec.~\ref{sec:Derivation of the Temperature lower bound}, and here $\operatorname{gap}[H]$ is the energy gap of the spectrum of $H$, between the GS and the first excited state). Here we for the computation of $p(\Pi_{\mathrm{GS}})$ we use a two-spin reduced density matrix. In other words, we consider the steering operators are coupled to two neighboring spins of the Heisenberg chain. These quantities are depicted in Fig.~\ref{fig:AFH}, where $N$ is the number of spins and $s=1/2,1,3/2,2$ is the spin. It is evident that (i) The non-vanishing value of $T_{\mathrm{eff,min}}$, which reflects the degree of frustration, tends to converge to a constant value as $N$ grows, and (ii) There is a clear odd-even effect.

\begin{figure}[htbp]
    \centering
      \begin{subfigure}[b]{0.4\textwidth}
         \centering
         \includegraphics[width=\textwidth]{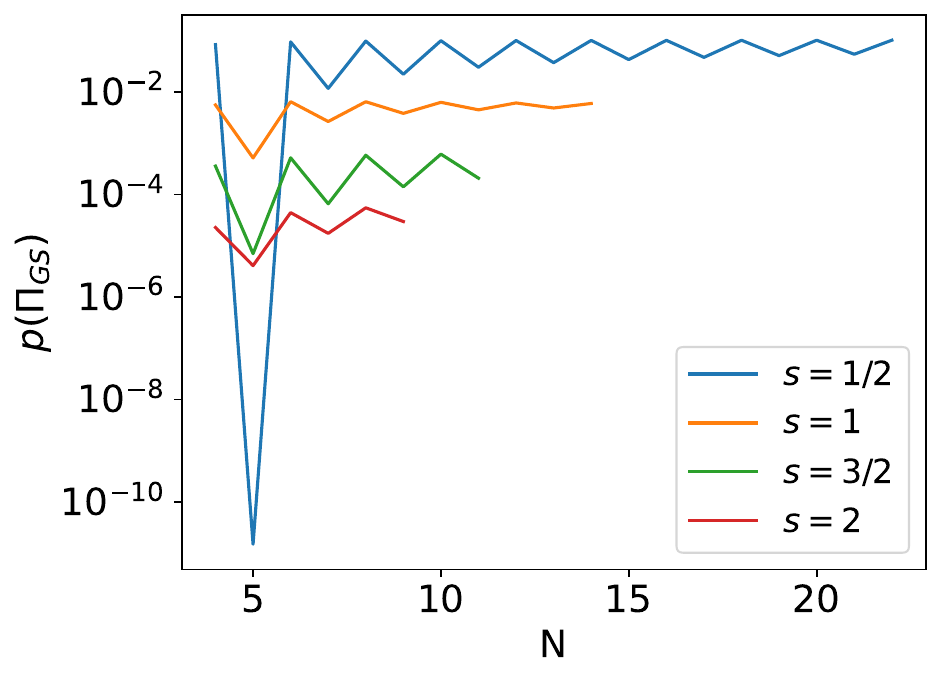}
         \caption{}
         \label{fig:AFH_fied}
     \end{subfigure}
     \hfill
      \begin{subfigure}[b]{0.4\textwidth}
         \centering
         \includegraphics[width=\textwidth]{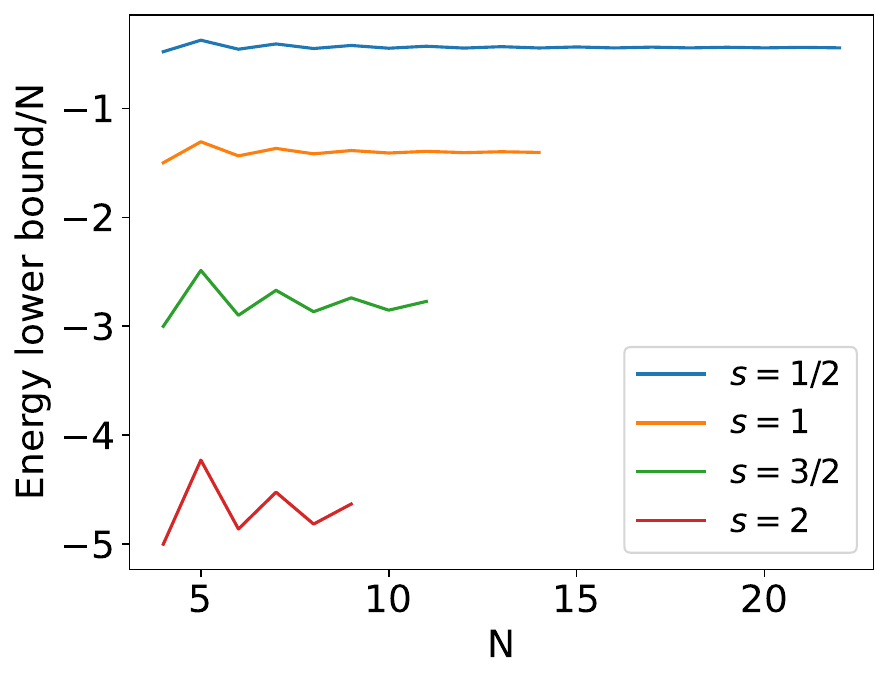}
         \caption{}
         \label{fig:AFH_ene}
     \end{subfigure}
     \hfill
      \begin{subfigure}[b]{0.4\textwidth}
         \centering
         \includegraphics[width=\textwidth]{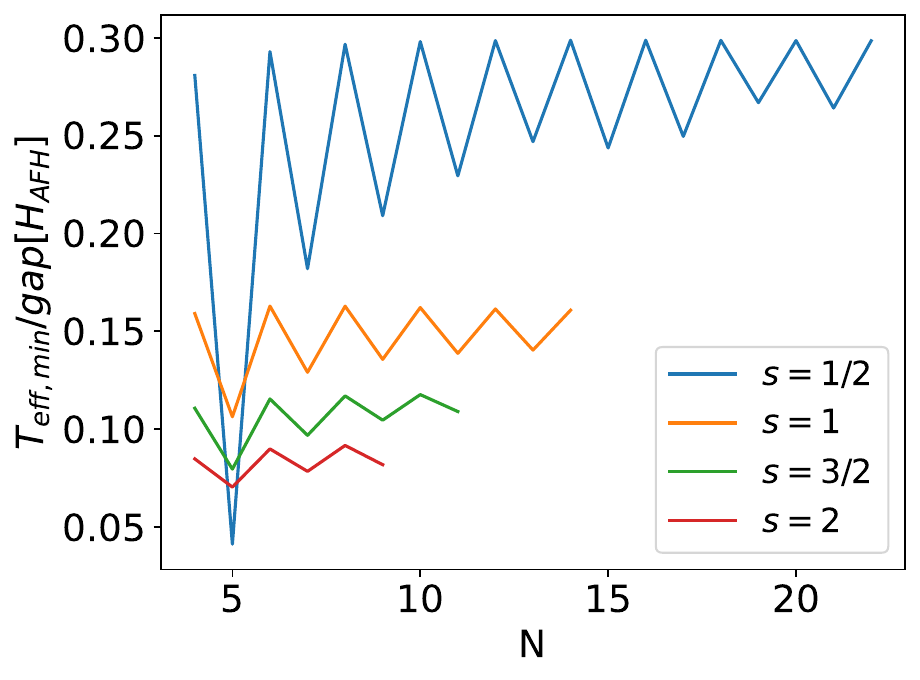}
         \caption{}
         \label{fig:AFH_temp}
     \end{subfigure}
    \caption{\textbf{Anti-ferromagnetic Heisenberg} model with $N$-long chains (periodic boundary conditions) of spin $s=1/2,1,3/2,2$. (a) The function $p(\Pi_{\mathrm{GS}})$, which is 1 minus the fidelity upper bound (cf. Eq.~\eqref{eq:upper bound fidelity}), calculated with the GS (manifold) of the above models. (b) Lower bound on the energy expectation value per spin of the asymptotic state. (c) Effective temperature lower bound $T_{\mathrm{eff,min}}/\operatorname{gap}[H_{\mathrm{AFH}}]$ calculated for the GS (manifold) of the models above.}
    \label{fig:AFH}
\end{figure}

\subsection{SYK model}\label{sec: SYK model}

\subsubsection{Dirac Fermion SYK model}\label{sec:Dirac Fermion SYK model}

 The Dirac Fermion version of the \textit{\textbf{SYK model}}~\cite{Kiteav15,PhysRevB.94.035135}  is given by the Hamiltonian 
\begin{equation}
H_{\mathrm{SYK}}^{\mathrm{D}}=\frac{1}{(2N)^{3/2}}\sum_{i,j,k,l=1}^N J_{ij;kl} c^\dagger_i c^\dagger_j c_k c_l - \mu \sum_i c^\dagger_i c_i,
\end{equation}
where $c_i$($c_i^\dagger$) run over $N$ Dirac fermion modes, $\mu$ is the  chemical potential, and $\{J_{ij;kl}\}$ are complex Gaussian random couplings. This model supports no local SCQ (see also Appendix~\ref{Discussions about non-steerable case}). This is to be expected since the SYK model is a standard platform to model quantum chaos; one then expects no preservation of local information, hence no local SCQ. In more detail, if there exists an SCQ, then under unitary evolution and with GS as the initial state, the expectation value of the SCQ is a constant over time. This would be in contradiction to the quantum chaos nature of the SYK ground state. The absence of SCQ immediately implies Non-Steerability (cf.  Condition~\ref{cond:2}, Sec.~\ref{sec:Necessary conditions for NFFJS Hamiltonian}). 
We next compute the minimal effective temperature, $T_{\mathrm{eff,min}}$ by calculating $p(|\psi_{\mathrm{GS}}\rangle)$ (cf. Eq.~\eqref{eq:upper bound}) for all $m$-body RDMs ($1\le m\le \lfloor N/2\rfloor$) obtained from the GS $|\psi_{\mathrm{GS}}\rangle$, where $N$ is the number of Fermion modes. We expect the GS to be non-degenerate, although our discussion holds for degenerate GS as well.  We note that the more local the steering operators are (hence the lower rank the RDM is), the smaller the bound $1-p(|\psi_{\mathrm{GS}}\rangle)$ (cf. Eq.~\eqref{eq:upper bound}) on how close we can get to the GS is. 
Roughly speaking, the bound on the minimal effective temperature reachable, $T_{\mathrm{eff,min}}$, depends on the rank of the steering operator(s)  (i.e. on $m$, where the $m+1$ body steering operator involves $m$ system's degrees-of-freedom and one ancilla). The specific form of the steering operators will determine only how close to that bound we can actually get.  

Fig.~\ref{fig:SYK}  depicts simulations of our SYK systems of varying sizes. For each value of $N$ we have considered $10^2$--$10^4$ realizations of $\{J_{ij,kl}\}$; for each realization we computed all possible RDMs of GS involving $ m=2 $ or $ \lfloor N/2\rfloor $ degrees of freedom, which would be applied with steering operators. In other words, we investigate both notions of locality as discussed in Sec.~\ref{sec:Definition of stabilizability}. We note that for smaller rank steering operators $T_{\mathrm{eff,min}}$ increases, and it likewise increases with $N$ when $m=2$. We note that for any given $N$ the distribution of $T_{\mathrm{eff,min}}$ as well as the distribution of the lowest excitation gap, $\operatorname{gap}[H_{\mathrm{SYK}}^{\mathrm{D}}]$, are non-Gaussian. Interestingly, for $m=\lfloor N/2 \rfloor $, the dimensionless ratio $T_{\mathrm{eff,min}}/\operatorname{gap}[H_{\mathrm{SYK}}^{\mathrm{D}}]$ is non-zero but Gaussian. This implies that the exact GS is not accessible even when each detector is coupled to more than half the system's degrees-of-freedom.  Even more intriguing: the variance of these dimensionless Gaussian distributions is  $\sim 1/N^2$, and the average scales as  $\sim 1/N$, meaning that our ability to cool this NFFNS model down towards the GS potentially improves with $N$. In other words, the bounds on overlap of the presumed surrogate state with the target state improve with $N$, i.e.   
\begin{equation}
\label{eq:SYKscaling}
 1-\langle\psi_{\mathrm{GS}} | \tilde{\rho}_{\mathrm{surr}}|\psi_{\mathrm{GS}}\rangle  \ge e^{-O(N)}.
\end{equation}
This scaling of Eq.~\eqref{eq:SYKscaling} is numerically obtained by performing finite size scaling with the average over disorder realizations. This is a reasonable result since our steering power, represented by $m$, increases with $N$. On the other hand, for $m=2$, we note that $T_{\mathrm{eff,min}}/\operatorname{gap}[H_{\mathrm{SYK}}^{\mathrm{D}}]$ increases with $N$, as follows
$$
1-\langle\psi_{\mathrm{GS}} | \tilde{\rho}_{\mathrm{surr}}|\psi_{\mathrm{GS}}\rangle  \ge \frac{1}{4}-O(1/N).
$$
It follows that as $N$ increases, this model becomes more chaotic. In other words, as we go to the thermodynamic limit, the two-body RDM is closer and closer to the maximally mixed state. In that limit it is not possible to distinguish the global GS from a generic random quantum state employing the Haar measure ~\cite{zyczkowski2001induced}. Therefore locally it is not possible to design any state-specific steering operation. 

\begin{figure}
    \centering
    \includegraphics[width=0.4\textwidth]{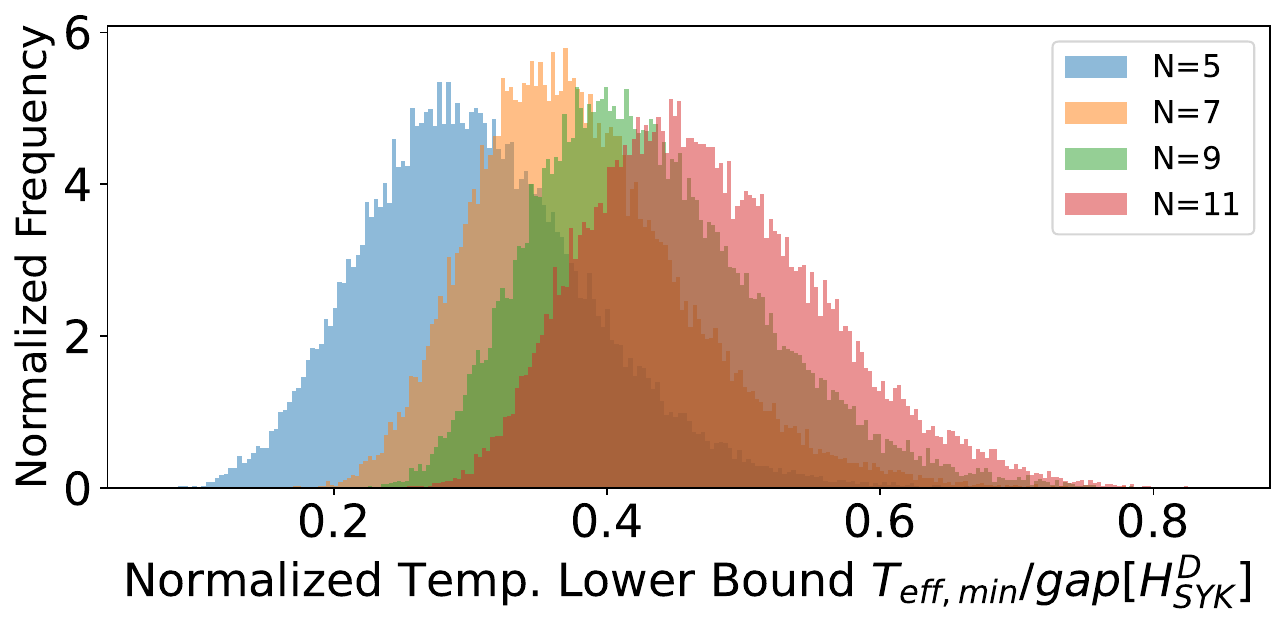}

    \caption{{\it Temperature lower bound for short-range system-detector interaction, $m=2$.} Shown are normalized histograms of the lower bound of the normalized effective temperature $T_{\mathrm{eff,min}}/\operatorname{gap}[H^{\mathrm{D}}_{\mathrm{SYK}}]$ computed for SYK model for Dirac fermions, for different disorder configurations (Eq.~\eqref{eq:temp}), and for different numbers of fermionic modes $N$. Here the occupation (number of Dirac fermions) is  $\lfloor N/2 \rfloor$.  Short-range steering operators consist of the detector's degree of freedom and  $m=2$ fermionic degrees of freedom, cf. Sec.~\ref{sec:Dirac Fermion SYK model}. }
    \label{fig:SYK2}
\end{figure}

\begin{figure}
    \centering
      \begin{subfigure}[b]{0.4\textwidth}
         \centering
         \includegraphics[width=\textwidth]{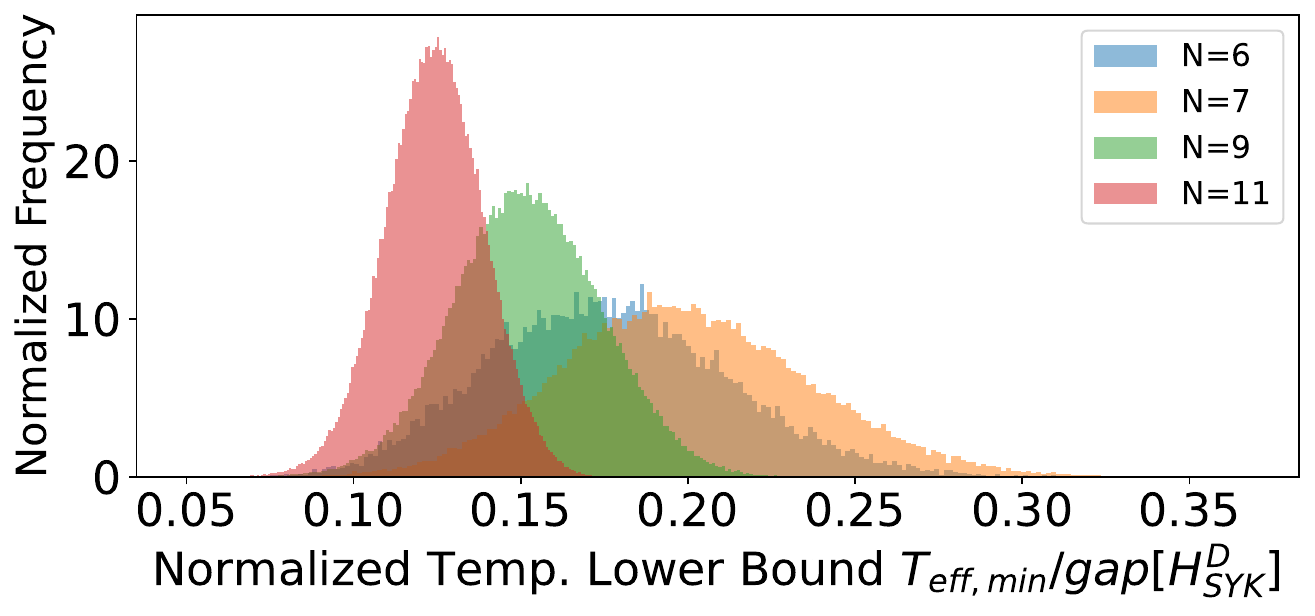}
         \caption{}
         \label{fig:SYK-Tgap}
     \end{subfigure}
     \hfill
      \begin{subfigure}[b]{0.4\textwidth}
         \centering
         \includegraphics[width=\textwidth]{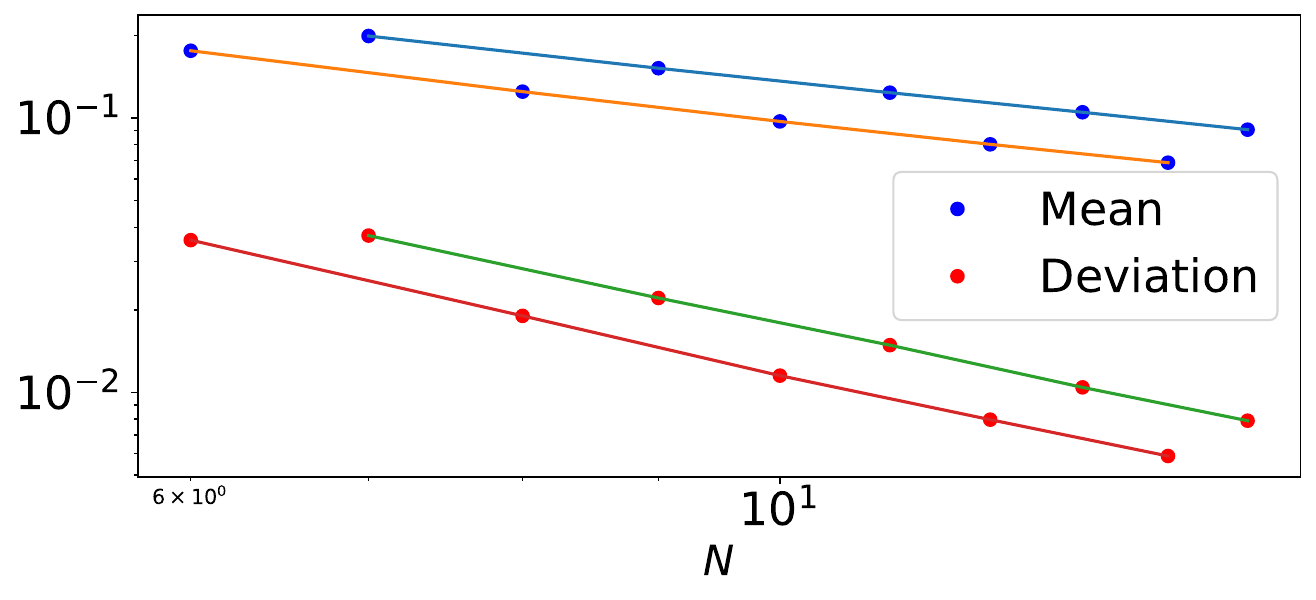}
         \caption{}
         \label{fig:SYK-mean}
     \end{subfigure}
    \caption{\textit{Temperature lower bound for long-range system-detector interaction, $m=\lfloor N/2\rfloor$.} (a) Normalized histograms of  $T_{\mathrm{eff,min}}/\operatorname{gap}[H_{\mathrm{SYK}}^{\mathrm{D}}]$ for Dirac-Fermionic SYK model with different numbers of fermion modes $N$. Note that it behaves like Gaussian distributions. (b) Mean and standard deviation for the Gaussian-like distributions in (a). Note the mean and standard deviation of distributions with odd $N$ is larger than that with even $N$. Here the occupation (number of Dirac fermions) is  $\lfloor N/2 \rfloor$.  Long-range steering operators consist of the detector's degree of freedom and  $\lfloor N/2\rfloor$ fermionic degrees of freedom, cf. Sec.~\ref{sec:Dirac Fermion SYK model}.}
    \label{fig:SYK}
\end{figure}

\subsubsection{Majorana Fermion SYK model }

The \textit{\textbf{Majorana Fermionic SYK Model}} ~\cite{SYKPhysRevLett.70.3339,Kiteav15} is represented by the Hamiltonian 
\begin{equation}
H_{\mathrm{SYK}}^{\mathrm{M}}=\frac{1}{(N)^{3/2}}\sum_{i,j,k,l=1}^{2N} J_{ijkl} \gamma_i \gamma_j \gamma_k \gamma_l,
\end{equation}
where $\{J_{ijkl}\}$ are anti-symmetric real random Gaussian couplings. 
We next re-write this Hamiltonian with Dirac fermion operators, employing the conversion $\gamma_{2j}=c_j+c^\dagger_j$ and $\gamma_{2j-1}=-i(c_j-c^\dagger_j)$, where $j=1,2,\cdots, N$ and $\{c_j\}$ are $N$ Dirac fermion annihilation operators. The resulting Hamiltonian is particle non-conserving but parity-conserving.

We then follow the line of our Dirac fermion SYK model, computing $m$-body RDMs and then deriving the $T_{\mathrm{eff,min}}/\operatorname{gap}[H_{\mathrm{SYK}}^{\mathrm{M}}]$ for the superoperators acting on $m$ Dirac Fermion modes (which represents $2m$ Majorana Fermion modes). Specifically we consider $N=5,6,\cdots, 12$. The resulting distributions, depicted in Fig.~\ref{fig:MFSYK-Tgap} and ~\ref{fig:MFSYK-Tgap2}, are qualitatively similar to the ones derived for the Dirac fermion SYK case. 

\begin{figure}[htbp]
    \centering
    \includegraphics[width=0.4\textwidth]{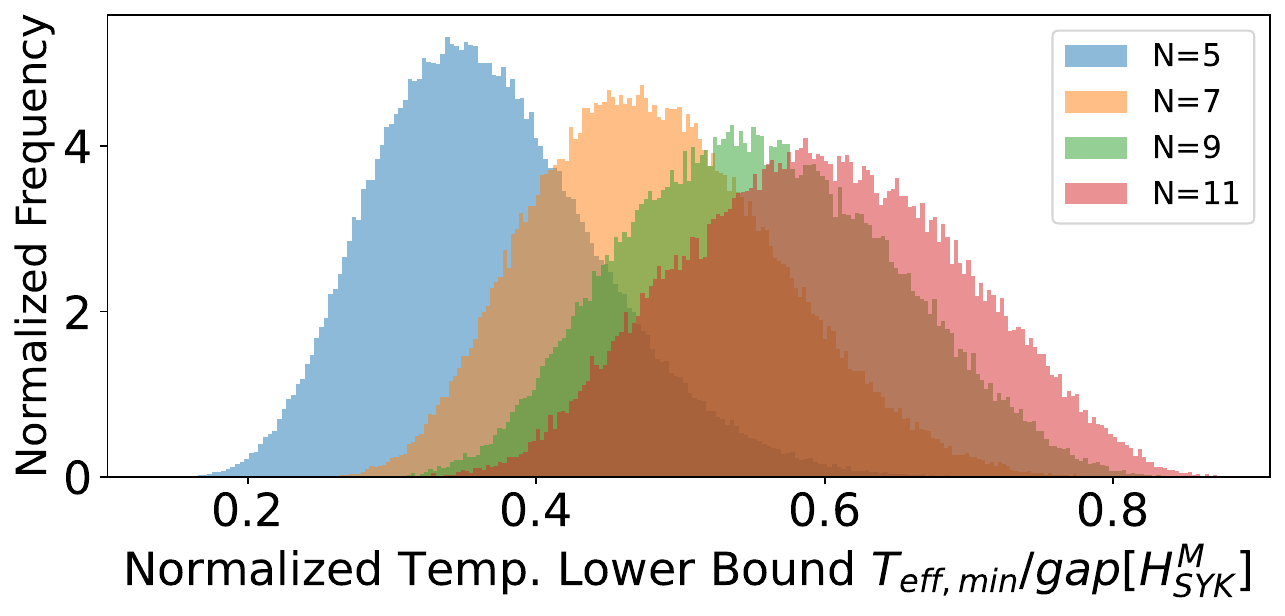}
    \caption{{\it Temperature lower bound for short-range system-detector interaction, $m=2$.} Shown are normalized histograms of the lower bound of the normalized effective temperature $T_{\mathrm{eff,min}}/\operatorname{gap}[H^{\mathrm{M}}_{\mathrm{SYK}}]$ computed for the SYK model for Majorana fermions, for different disorder configurations (Eq.~\eqref{eq:temp}), and for different numbers of fermionic modes $N$.  Short-range steering operators consist of the detector's degree of freedom and  $m=2$ fermionic degrees of freedom, cf. Sec.~\ref{sec:Dirac Fermion SYK model}.}
    \label{fig:MFSYK-Tgap2}
\end{figure}

\begin{figure}[htbp]
    \centering
      \begin{subfigure}[b]{0.4\textwidth}
         \centering
         \includegraphics[width=\textwidth]{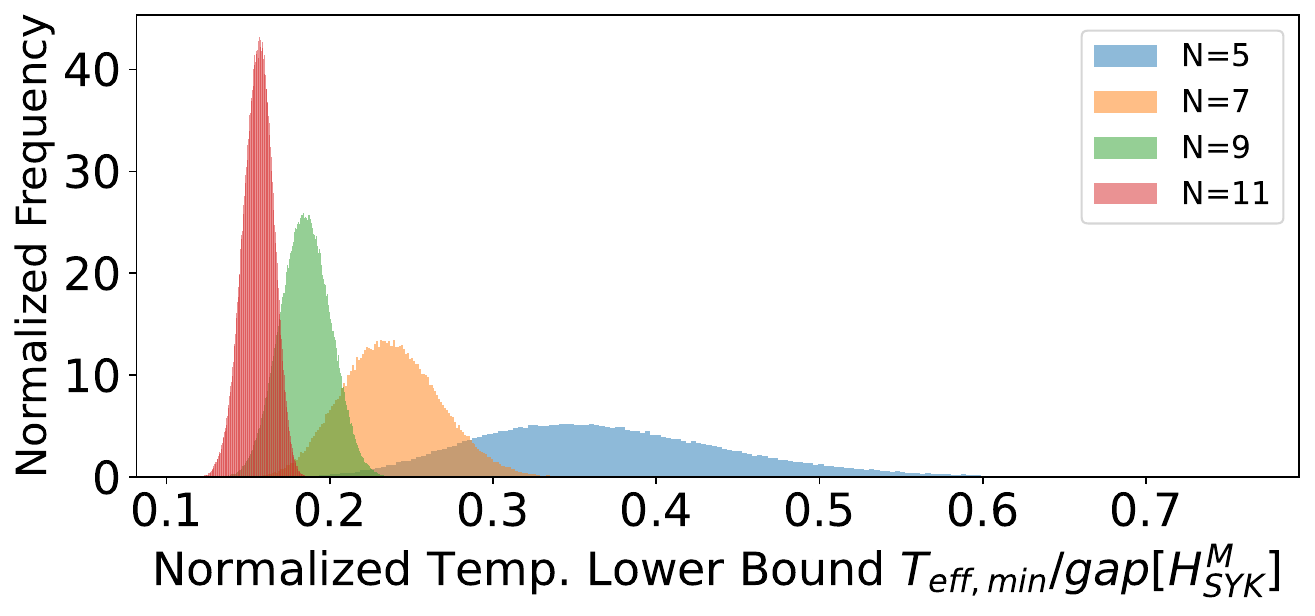}
         \caption{}
         \label{fig:MFSYK-Tgap}
     \end{subfigure}
     \hfill
      \begin{subfigure}[b]{0.4\textwidth}
         \centering
         \includegraphics[width=\textwidth]{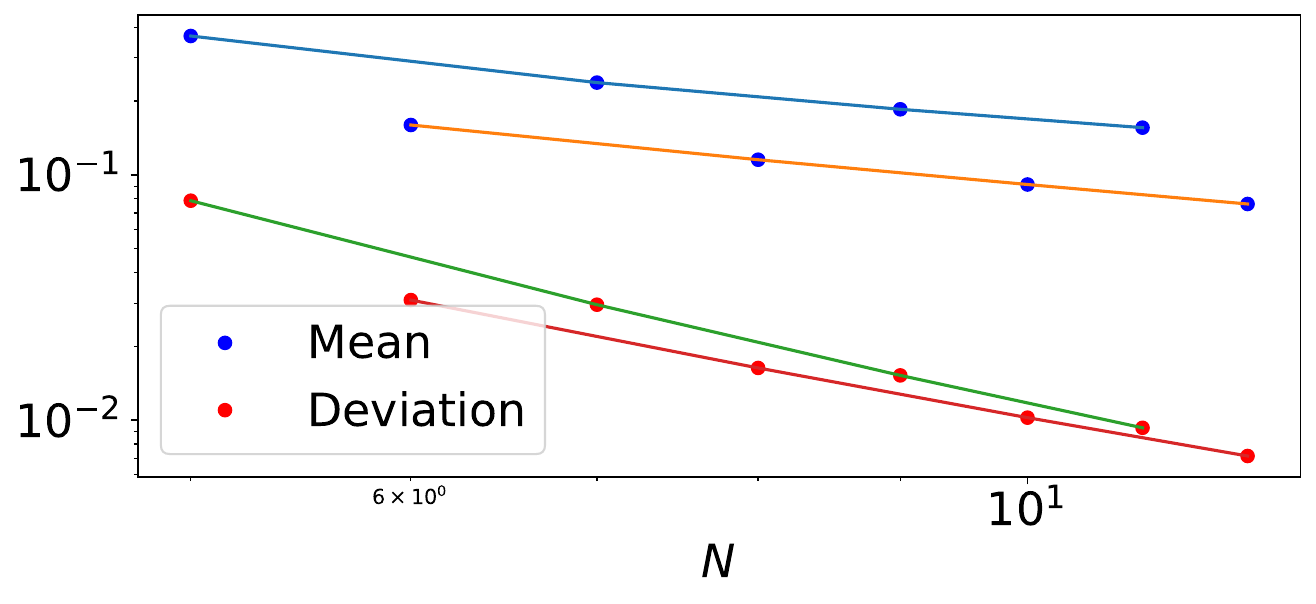}
         \caption{}
         \label{fig:MFSYK-mean}
     \end{subfigure}
    \caption{\textit{Temperature lower bound for long-range system-detector interaction, $m=\lfloor N/2\rfloor$.} (a)Normalized histograms of  $T_{\mathrm{eff,min}}/\operatorname{gap}[H_{\mathrm{SYK}}^{\mathrm{M}}]$ for Majorana-Fermionic SYK model with different numbers of fermion modes $N$. Note that it behaves like Gaussian distributions. (b) Mean and standard deviation for the Gaussian-like distributions in (a). Note the mean and standard deviation of distributions with odd $N$ is larger than that with even $N$.  Long-range steering operators consist of the detector's degree of freedom and  $\lfloor N/2\rfloor$ fermionic degrees of freedom, cf. Sec.~\ref{sec:Dirac Fermion SYK model}.}
    \label{fig:MFSYK}
\end{figure}

\subsection{Fermi-Hubbard model}\label{sec:Fermi-Hubbard model}

The \textit{\textbf{two-dimensional Fermi-Hubbard model}} ~\cite{FHRevModPhys.70.1039} is known to have concrete physical realizations~\cite{tarruell2018quantum}. Its Hamiltonian is given by
\begin{equation}
H_{\mathrm{FH}}=-t\sum_{\langle ij\rangle}(c^\dagger_{i\uparrow} c_{j\uparrow}+c^\dagger_{i\downarrow} c_{j\downarrow})+U\sum_i n_{i\uparrow}n_{i\downarrow},
\end{equation}
where $t$ is the hopping amplitude and $U$ is the on-site interaction. Here too we calculate $T_{\mathrm{eff,min}}$, with respect to RDMs computed according to dashed boxes in Fig.~\ref{fig:FHL}. To perform this computation we first need to find the ground state, employing exact diagonalization. Specifically, we compute GSs of a Fermi-Hubbard model on a  $3\times 3$ lattice with periodic boundary conditions, and different numbers of electrons, $N$. We have addressed finite-size manifestations of two macroscopic phases. (i) Super-conducting phase. Here we consider $N=8$.   We then compute $T_{\mathrm{eff,min}}$ based on the exact GS wavefunctions; $T_{\mathrm{eff,min}}$ turns out to be of order $0.2 \times \operatorname{gap}[H_{\mathrm{FH}}]$, cf. Fig.~\ref{fig:FHR}, which marks the lower bound achievable through our passive steering, but not necessarily a temperature that can be reached in reality. This energy turns out to be smaller than estimates of the critical temperature, $T_c$, of $d$-wave superconductor. The latter is approximated by $1/1.76\times \operatorname{gap}[H_{\mathrm{FH}}]\approx 0.57\times \operatorname{gap}[H_{\mathrm{FH}}]$~\cite{sigrist2005introduction}. We also compute $T_{\mathrm{eff,min}}$ in momentum space, cf. Fig.~\ref{fig:FHM}.  Note that there is no cap on our protocols to reach a temperature compatible with the excitation gap, as well as the thermodynamics critical temperature. Note that quantum simulation of $d$-wave superconductors has been a central focus in cold atom experiments. The difficulty to explore the superconducting phase may reflect the fact that for this unsteerable model a tighter bound of $T_{\mathrm{eff,min}}$ is set at a temperature higher than the critical temperature $T_c$.  (ii) The anti-ferromagnetic phase of $H_{\mathrm{FH}}$ with $N=9$ electrons, which acts as an example of a degenerate case, cf. Fig.~\ref{fig:FHAFM}. In this case we have a four-fold degenerate GS manifold. Since the necessary Condition~\ref{cond:3} is not satisfied, the model is NFFSS. Details of these results are presented in Appendix~\ref{sec:Implementation of Examples}. 

\begin{figure}[htbp]
    \centering
    \begin{subfigure}[b]{0.3\textwidth}
         \centering
         \include{FHlattice}
         \caption{}
         \label{fig:FHL}
     \end{subfigure}
     \hfill
     \begin{subfigure}[b]{0.31\textwidth}
         \centering
         \includegraphics[width=\textwidth]{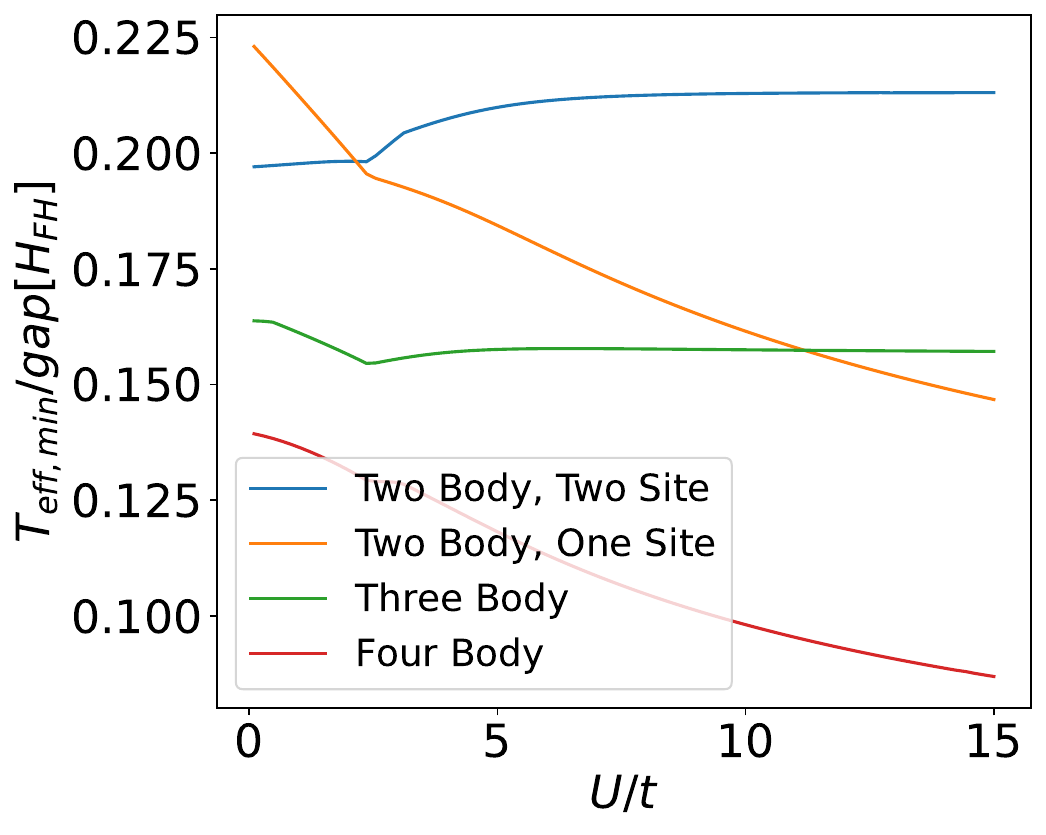}
         \caption{}
         \label{fig:FHR}
     \end{subfigure}
     \hfill
     \begin{subfigure}[b]{0.3\textwidth}
         \centering
         \includegraphics[width=\textwidth]{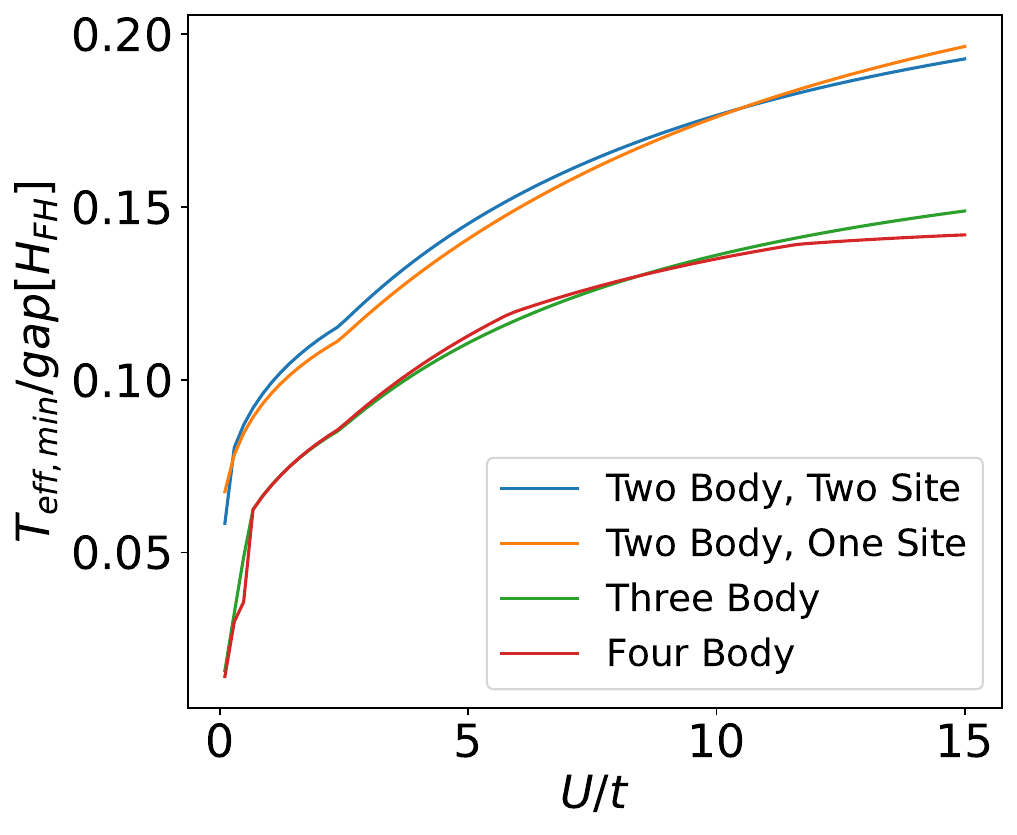}
         \caption{}
         \label{fig:FHM}
     \end{subfigure}
     \hfill
     \begin{subfigure}[b]{0.3\textwidth}
         \centering
         \includegraphics[width=\textwidth]{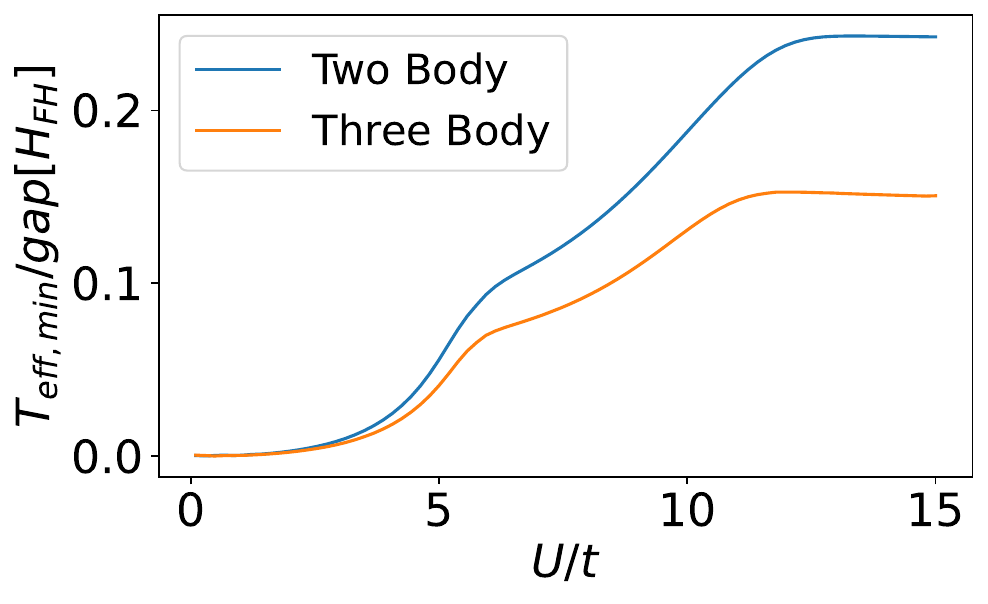}
         \caption{}
         \label{fig:FHAFM}
     \end{subfigure}
    \caption{(a)Lattice of the two dimensional Fermi-Hubbard model, showing single electron fermionic modes. Each mode may  be filled by electrons of spin up (blue) or down (red). The boxed (marked by dashed green lines) represents the support of the reduced density matrices employed for the calculation of the parameter $p(\Pi_{\mathrm{GS}})$, cf. Eq.~\eqref{eq: upper bound degenerate}.  (b) $T_{\mathrm{eff,min}}/\operatorname{gap}[H_{\mathrm{FH}}]$ for $N=8$ electrons, which (nearly) reflects the doping required  from the the superconducting phase in the thermodynamic limit. (c) $T_{\mathrm{eff,min}}/\operatorname{gap}[H_{\mathrm{FH}}]$ for $N=8$ calculated in momentum space. (d) $T_{\mathrm{eff,min}}/\operatorname{gap}[H_{\mathrm{FH}}]$ for the half-filling phase, as function of of $U/t$. In the thermodynamic limit, this corresponds to the anti-ferromagnetic phase.}
    \label{fig:FH}
\end{figure}

\section{Summary and Perspective}
\subsection{Main results and their significance}\label{sec:Summary of main contributions}

Our analysis here addresses a family of steering protocols, referred to as  passive (blind) steering of quantum states~\cite{Measurement-inducedsteeringPhysRevResearch.2.033347},  aimed at acting on an arbitrarily chosen initial state and driving ("steering") it towards a linear target space (e.g., a target state) which is pre-selected at will.  This represents a class of protocols consisting of generalized (weak) measurements, with the following provisions: (i) all measurements are local (that is, involve a finite number of the system’s degrees of freedom); (ii) each measurement involves an initialization of the detector in a prescribed state, and then coupling and decoupling of the system and the detector, and finally tracing out over the detector’s readouts; (iii) the nature of the individual measurements and their sequence are passive: they are independent of the detectors’ readouts and of the system’s running state. 

The analysis presented here comprises two layers:
\begin{itemize}
    \item We present a classification of steerable (and non-steerable) classes of states, introducing a steerable class which has been hitherto unnoticed: a Non-Frustration-Free Jittery Steerable (NFFJS) class. We derive the  necessary conditions for states belonging with this class should satisfy,  and provide a family of examples consisting of local commuting Pauli Hamiltonians. Our examples represent classically frustrated Hamiltonians. 
    \item The Non-Frustration-Free Non-Steerable (NFFNS) class of Hamiltonians does not allow steering towards the Hamiltonian’s ground state.  Instead, we ask how close to the ground state one can get through passive steering. We quantify the distance to the ground state invoking three criteria: (i) the (global) fidelity  (involving the scalar product of the asymptotically steered state (surrogate state) and the true target state), cf. Eq.~\eqref{eq:upper bound fidelity}; (ii) the "effective temperature" of the asymptotically steered state, cf. Eq.~\eqref{eq:temp}; (iii) the energy expectation value of the asymptotically steered state, cf. Eq.~\eqref{eq: energy}. For criterion (i) we provide an upper bound on the fidelity, while for (ii) we present a lower bound on the effective temperature, i.e. the "glass floor", and for (iii) we derive a lower bound on the energy. We implement our general approach on several paradigmatic examples: both the Dirac fermion and the Majorana fermion Sachdev–Ye–Kitaev (SYK) model, the Fermi-Hubbard model, and the anti-ferromagnetic Heisenberg model.
\end{itemize}

We note that criterion (ii) and (iii) are essentially equivalent (cf. Eq.~\eqref{eq: energy}). We then focus on criteria (i) and (ii) only. For the four non-steerable models considered, the effective temperature lower bound reveals the following insights:
\begin{itemize}
    \item  $T_{\mathrm{eff,min}}/\operatorname{gap}[H]$  serves as a normalized quantifier for the degree of frustration. The fact that with the lower bound on the effective temperature  this parameter may reach values much smaller than unity, does not imply that this bound is necessarily reachable. In case it does it would imply an efficient cooling protocol, carrying us to sub-gap temperatures. 
\item The degree of frustration diminishes with
increased controlling power, where steering operators involve larger subsystems (a larger number of system's degrees-of-freedom). 
\item With constant controlling power, the degree of frustration converges in the thermodynamic limit.
\end{itemize}

In summary, our discussion extends and refines the classification of Hamiltonians defining steerability of ground states. A key theme here is frustration. Revising existing knowledge, we underline the subtlety of this notion, demonstrating that for certain classes of frustration steering can be realized. Likewise, we stress that steering may not necessarily  result in a steady steerable state, but rather in a steady subspace featuring jittery dynamics. We have provided a full characterization of the Frustration-Free Steerable (FFS) and Non-Frustration-Free Steadily Steerable (NFFSS) class, along with necessary conditions for the Non-Frustration-Free Jittery Steerable (NFFJS) class, and equivalently sufficient condition for the Non-Frustration-Free Non-Steerable (NFFNS) class.  Each of these classes was supported by certain examples. 

\subsection{Discussion and Perspective}

Evidently, our analysis presents new intriguing questions, opening directions for future study. Some of these challenges are listed here:

{\it Incorporating system’s Hamiltonian.} Our present analysis of steerability and non-steerability scenarios does not include a system’s Hamiltonian; the dynamics is solely introduced through local steering operators. A hypothetical local Hamiltonian is invoked (cf. Def.~\ref{def:stabilizability}) in order to define a target state (manifold), which happens to be the GS of this local Hamiltonian. The question at focus here is whether one may employ a quantum computer (no need to introduce a system's Hamiltonian, as opposed to a quantum simulator) in order to prepare the ground state of a given Hamiltonian through passive steering.  As we have discussed, the steerability of the target GS (manifold) is intimately related to the frustration of the hypothetical Hamiltonian defining the target manifold. Given the fact that there is no system’s Hamiltonian, we implement an alternative definition of Frustration-Free property of the target state(s), i.e. the parent Hamiltonian, cf. Def.~\ref{def:pHFF}.

Incorporating a system's Hamiltonian into the dynamics, might modify, and in fact improve, our "glass floor" lower bound, see below.

\textit{Finding tighter bounds}. It is noteworthy that, in the case of non-steerable ground states,  our lower bound on the minimal normalized temperature that cannot be surpassed, $T_{\mathrm{eff,min}}/\operatorname{gap}[H]$, is a bit loose. Establishing a tighter bound, closer to the actual achievable temperature, may be facilitated through the study of specific models. Such an endeavor could reveal that certain quantum phases which we try to engineer through our steering protocols (which amounts to local cooling) are unattainable,  if the normalized temperature lower bound turns out to be higher than the phase transition temperature. Our effort might also address achieving  a tighter bound  on the  fidelity of a numerical low-rank Ansatz (such as Tensor Network States~\cite{orus2014practical}), with respect to the genuine target state.

\textit{Ergodicity}. The issue of ergodicity arises when the steering protocol does not converge to a specific target state (or a specific state within the target space). One example is the class of NFFJS states, which are shown here to be steerable (we approach the target manifold), but then proceed to hop from one state to another within this manifold. Averaging over time (or alternatively over different steering trajectories), the question is whether one covers uniformly the target space (hence “ergodicity”). Another open issue here concerns non-steerable states. Given one asymptotically approaches the glass floor, there is clearly a lack of ergodicity (the weight of high energy states vanishes). Can one quantitatively characterize the degree of non-ergodicity, and does the latter relate to the glass floor bound?

\textit{Relation of steerability to integrability}. To facilitate the classification of Hamiltonian steerability, we introduced the concept of subspace conserved quantity (SCQ). Given that a subspace conserved quantity can be interpreted as a conserved quantity within a subspace, the steerability of a Hamiltonian is closely tied to its integrability. Further investigation might unveil the deep connection between steerability and integrability.

{\it Alternative approaches to GS  engineering}. We stress that the protocols discussed  here employ \textit{local}, \textit{passive} (Markovian) and \textit{physical} operations (superoperators). Our classification of steerable and non-steerable classes of states may be drastically revised if the restriction to \textit{local passive}  steering is relaxed.  There are several paradigmatic  extensions to the protocols discussed here:   

(i) An extension to non-local steering superoperators implies that originally non-steerable states may become steerable. This non-steerable scenario is the case when the detector is coupled to a set of system operators, each covering a finite volume of the system. Importantly, non-local steering operators may emerge through the following protocol. Let us combine weak local steering with strong non-commuting system Hamiltonian~\cite{PhysRevResearch.6.033147,Josias2025prep}. The dressing of local steering with the system Hamiltonian (e.g., when the unitary time evolution between two consecutive measurements in long enough) leverages it to a non-local superoperator, or, alternatively, non-local Kraus operators describing the system's evolution at each discrete measurement step. This may eventually facilitate  the GS engineering of originally "non-steerable" states. As a comparison, our present steering protocol, focusing on local superoperators, can be thought of as the limit of finite steering with a vanishingly weak system Hamiltonian. In this limit, the Hamiltonian evolution can be Trotterized into a product of local unitaries, and then the whole dynamics can be described with local superoperators.

(ii) One can implement a series of active decision-making protocols~\cite{PRXQuantum.4.020347}, where we follow the evolution of the initial state and the readout history. Based on these (that allow the computation of the running state), we can  determine the next time step steering operator that optimizes the cost function. This facilitates steering towards states that, in our classification, are non-steerable (e.g., the W or the GHZ states)~\cite{PhysRevResearch.6.013244}.  This family, by construction, represents non-Markovian protocols. Another example of the latter are variational imaginary time evolution protocols, where the imaginary time evolution at each time step is simulated by an optimized unitary, where the optimization is based on computing the running state ~\cite{mcardle2019variational}. Finally,  we note that one can simplify the steering dynamics discussed here by resorting to {\it dilute steering} ~\cite{Langbehn2023}, where one a limited subset of the system's degrees of freedom are operated upon by the measuring detectors.

(iii) One can resort to non-trace-preserving protocols, which might be physically not feasible or inefficient. A paradigmatic example here is dynamics based on post-selection: only a certain subset of readout histories are to be included (when only "no jump" trajectories are considered, one resort to non-Hermitian Hamiltonian dynamics~\cite{plenio1998quantum,PhysRevA.101.062112,minev2019catch}. Furthermore, in the absence of any inherent system's Hamiltonian, such post-selection amounts to anti-Hermitian Hamiltonian dynamics, mimicking imaginary time evolution). Once certain trajectories are discarded, the probability (i.e., the trace) is not preserved. Clearly, the efficiency of such protocols (unless some correction steps are taken ~\cite{minev2019catch}) is exponentially low, since, with increasing measurement time, we lose an exponentially large number of possible trajectories. 

{\it Determining steerability: how hard is it?}.  As the steering protocol is equivalent to the stable preparation of the ground state of a quantum Hamiltonian, an obvious question is how hard it is to determine the steerability of the system at hand.  In essence, the steerability problem is connected to the \textit{Quantum $k$-SAT Problem} ~\cite{kempe2006complexity}, which seeks to determine whether a given Hamiltonian is Frustration-Free. If it is Frustration-Free, then it is steerable. However, the 3-local version of this problem is $QMA_1$-complete ~\cite{gosset2016quantum}, implying that determining whether a given Hamiltonian is steerable might be inherently challenging, even on a quantum computer.

{\it Quantum advantage in executing steering protocols? } Conversely, one might contemplate the potential existence of a quantum speedup in finding ground states. We note that  simply determining the ground state energy of a general commuting Pauli Hamiltonian is already an NP-complete problem, as it is equivalent to the Maximum-Likelihood Decoding Problem ~\cite{NPC1055873} (cf. Prop.~\ref{prop:npc} in Appendix~\ref{app:Spectrum of Commuting Pauli Hamiltonian}). Consequently, at present one may not expect quantum speedup (as compared with classical methods) in solving for the ground state manifold, noting that  currently there is no clue of quantum speedup for NP-complete problems. Nevertheless, a quantum advantage could emerge in the quest for non-trivial correlation functions of ground states, given the direct measurability of such quantities.

\section*{Acknowledgement}

We acknowledge very useful discussions with Peize Ding, Igor Gornyi, Yong-Jian Han, Christiane Koch, Giovanna Morigi and Hao Zhang.  In particular we enjoyed active interaction with Josias Langbehn, Yuzhuo Tian and Zekai Wang. Y.W. was supported by the BMBF project Spinning (grant no. 13N16215), project PhoQuant (grant no. 13N16110). Y.G. was supported by the Deutsche Forschungsgemeinschaft (DFG) grant SH 81/8-1 and by a National Science Foundation (NSF)–Binational Science Foundation (BSF) grant 2023666.

\appendix

\section{Introductory comments on open Quantum system dynamics}\label{Open Quantum System Dynamics}

In the following, we provide a concise review of the notations and key results pertaining to the dynamics of open quantum systems. We also elucidate the distinctions between steering protocols driven by continuous evolution and those governed by discrete evolution.

\subsection{Preliminaries}

In this section, we provide a brief review of the dynamics of open quantum systems. The Lindblad equation, an instance of the quantum master equation, describes the non-unitary evolution of a quantum system, extending the standard unitary dynamics driven by a Hamiltonian~\cite{preskill1998lecture}. Without loss of generality, it can be written as
\begin{align}
\dot{\rho}=\mathcal{L}(\rho)=&-\frac{i}{\hbar}[H,\rho]\nonumber\\
&+\sum_i \gamma_i \left( L_i\rho L_i^\dagger-\frac{1}{2}\{L_i^\dagger L_i,\rho\} \right),
\end{align}
where $\rho$ denotes the density matrix of the system, $\mathcal{L}$ is the Lindbladian superoperator, $H$ is the system's Hamiltonian, $\gamma_i$ are non-negative real numbers representing jumping rates, and $L_i \in \mathcal{L}(\mathcal{H})$ are the Lindblad or jump operators. Further details can be found in Ref.~\cite{manzano2020short}.

We now analyze the asymptotic behavior of the quantum system in the infinite-time limit, following Ref.~\cite{PRX15}. Note that the symbol $\mathcal{L}$ may refer to a linear space or the Lindbladian superoperator, depending on the context. The solution can be formally expressed as
\begin{equation*}
\rho(t)=e^{\mathcal{L}t}(\rho(0)),
\end{equation*}
where we consider a time-independent $\mathcal{L}$ for simplicity. The asymptotic behavior is encoded in the infinite-time superoperator $\mathcal{P} = \lim_{t \rightarrow \infty} e^{\mathcal{L} t}$, which we refer to as the strong superoperator. Denoting the invariant subspace, or fixed points of $\mathcal{P}$, as $\operatorname{As}(\mathcal{H})$, the steady state $\rho(\infty) = \mathcal{P}(\rho(0))$ must lie within this subspace.

Furthermore, it has been shown that $\operatorname{As}(\mathcal{H})$ constitutes a "distorted" C*-algebra, which can be decomposed up to an isomorphism as
\begin{equation*}
\mathcal{H}\simeq\bigoplus_k A_k\otimes \mathbbm{1}_{d_k},
\end{equation*}
where $\{A_k\}$ are linear spaces, and $\{\mathbbm{1}_{d_k}\}$ are $d_k$-dimensional identity matrices~\cite{InvariantPhysRevA.82.062306}. Consequently, the states in $\operatorname{As}(\mathcal{H})$ must have the form
\begin{equation}
\rho= \sum_k p_k \rho_k, \quad \sum_k p_k=1,\quad \rho_k=a_k\otimes \mu_k,
\label{eq:decomposition}
\end{equation}
where $a_k$ is an arbitrary density matrix on $\mathcal{L}(\mathcal{H}_{A_k})$, and $\mu_k \in \mathcal{L}(\mathcal{H}_{d_k})$ is a fixed state. In other words, $\operatorname{As}(\mathcal{H})$ exhibits a block-diagonal structure, and the possible steady states are mixtures of states from different blocks (see Fig.~\ref{fig:fixed}).

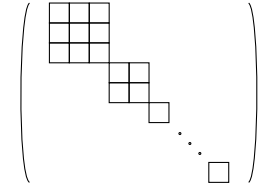
\begin{figure}
    \centering
    \begin{tikzpicture}[x=0.75pt,y=0.75pt,yscale=-1,xscale=1]

\draw    (80,50) .. controls (74.5,54.44) and (74,136.44) .. (80,140) ;
\draw   (90,50) -- (100,50) -- (100,60) -- (90,60) -- cycle ;
\draw   (100,50) -- (110,50) -- (110,60) -- (100,60) -- cycle ;
\draw   (110,50) -- (120,50) -- (120,60) -- (110,60) -- cycle ;
\draw   (90,60) -- (100,60) -- (100,70) -- (90,70) -- cycle ;
\draw   (100,60) -- (110,60) -- (110,70) -- (100,70) -- cycle ;
\draw   (110,60) -- (120,60) -- (120,70) -- (110,70) -- cycle ;
\draw   (90,70) -- (100,70) -- (100,80) -- (90,80) -- cycle ;
\draw   (100,70) -- (110,70) -- (110,80) -- (100,80) -- cycle ;
\draw   (110,70) -- (120,70) -- (120,80) -- (110,80) -- cycle ;
\draw   (120,80) -- (130,80) -- (130,90) -- (120,90) -- cycle ;
\draw   (130,80) -- (140,80) -- (140,90) -- (130,90) -- cycle ;
\draw   (120,90) -- (130,90) -- (130,100) -- (120,100) -- cycle ;
\draw   (130,90) -- (140,90) -- (140,100) -- (130,100) -- cycle ;
\draw   (140,100) -- (150,100) -- (150,110) -- (140,110) -- cycle ;
\draw   (155,115.5) .. controls (155,115.22) and (155.22,115) .. (155.5,115) .. controls (155.78,115) and (156,115.22) .. (156,115.5) .. controls (156,115.78) and (155.78,116) .. (155.5,116) .. controls (155.22,116) and (155,115.78) .. (155,115.5) -- cycle ;
\draw   (160,120.5) .. controls (160,120.22) and (160.22,120) .. (160.5,120) .. controls (160.78,120) and (161,120.22) .. (161,120.5) .. controls (161,120.78) and (160.78,121) .. (160.5,121) .. controls (160.22,121) and (160,120.78) .. (160,120.5) -- cycle ;
\draw   (165,125.5) .. controls (165,125.22) and (165.22,125) .. (165.5,125) .. controls (165.78,125) and (166,125.22) .. (166,125.5) .. controls (166,125.78) and (165.78,126) .. (165.5,126) .. controls (165.22,126) and (165,125.78) .. (165,125.5) -- cycle ;
\draw   (170,130) -- (180,130) -- (180,140) -- (170,140) -- cycle ;
\draw    (190,50) .. controls (195.5,54.44) and (195.5,135.94) .. (190,140) ;

\end{tikzpicture}
    \caption{Schematics of the fixed points and block-diagonal structure.}
    \label{fig:fixed}
\end{figure}

For the convenience of representing superoperators, we introduce the so-called double-ket notation. This establishes an equivalence between $\mathcal{L}(\mathcal{H})$ and $\mathcal{H} \otimes \mathcal{H}^*$; specifically, a matrix $M = M_{ij} |i\rangle \langle j|$ is represented as $|M\rangle\rangle = M_{ij} |i\rangle \otimes |j\rangle$. The inner product of double-kets and double-bras corresponds to the Frobenius inner product, namely, $\langle\langle N|M\rangle\rangle = \operatorname{Tr} N^\dagger M$. In particular, the expectation value of an observable $O$ with respect to a quantum state $\rho$ can now be written as $\operatorname{Tr} O \rho = \operatorname{Tr} O^\dagger \rho = \langle\langle O|\rho\rangle\rangle$. The superoperator can thus be viewed as a linear map over $\mathcal{H} \otimes \mathcal{H}^*$, simply expressed as $\mathcal{P}: |\rho\rangle\rangle \rightarrow \mathcal{P} |\rho\rangle\rangle$.

Additionally, we introduce superoperators operating in the Heisenberg picture. Given the established equivalence between $\mathcal{L}(\mathcal{H})$ and $\mathcal{H} \otimes \mathcal{H}^*$, the evolution of the expectation value becomes $\langle\langle O|\rho\rangle\rangle \rightarrow \langle\langle O| \mathcal{P} |\rho\rangle\rangle$. Therefore, in the Heisenberg picture, the quantum state $|\rho\rangle\rangle$ remains unchanged, while observables evolve as $|O\rangle\rangle \rightarrow \mathcal{P}^\dagger |O\rangle\rangle$, where $\mathcal{P}^\dagger$ is the Hermitian conjugate of $\mathcal{P}$ as a linear map over $\mathcal{H} \otimes \mathcal{H}^*$. Specifically, with the Kraus representation $\mathcal{P}(\cdot) = \sum_i K_i \cdot K_i^\dagger$, its Hermitian conjugate is $\mathcal{P}^\dagger(\cdot) = \sum_i K_i^\dagger \cdot K_i$, where $K_i$ are Kraus operators satisfying $\sum_i K_i^\dagger K_i = \mathbbm{1}$ to ensure the map is trace-preserving. Note that in the Heisenberg picture, superoperators are applied in reverse time order.

\subsection{Steering superoperators}\label{app:Steering superoperators}

In this section,  we further clarify certain notions concerning  steering operators

For FFS and NFFSS Hamiltonians, we can directly construct a set of superoperators $\{\mathcal{P}_i\}_i$ such that the GS manifold coincides with the common invariant subspace of all $\{\mathcal{P}_i\}_i$. This steering protocol can be simulated by the Lindblad equation $\dot{\rho} = \sum_i \mathcal{L}_i(\rho)$, where each $\mathcal{L}_i$ contains only local Lindblad operators and satisfies $\lim_{t \rightarrow \infty} e^{\mathcal{L}_i t} = \mathcal{P}_i$. However, this simulation fails for NFFJS Hamiltonians where frustration is present. The simultaneous action of local "cooling" operations fails due to geometric frustration.

We now discuss the assumptions regarding strong superoperators. Suppose $\mathcal{P}_0$ is a non-strong superoperator, and its strong version is defined as $\mathcal{P}_s = \lim_{a \rightarrow \infty} \mathcal{P}_0^a$, where the limit can be approximated by a single, sufficiently large $a$ when the limit does not exist. If $\mathcal{P}_0$ stabilizes the target space $\mathcal{H}_1$, then $\mathcal{P}_0(\rho) \in \mathcal{H}_1$ for all $\rho \in \mathcal{H}_1$. Consequently, for every integer $a$, we have $\mathcal{P}_0^a(\rho) \in \mathcal{H}_1$ for all $\rho \in \mathcal{H}_1$, implying that sequences of superoperators can be replaced by strong ones. Additionally, when unitary dynamics exist within $\operatorname{As}(\mathcal{H})$, since $\mathcal{P}_0(\rho) \in \mathcal{H}_1$ for all $\rho \in \mathcal{H}_1$, such unitary dynamics are confined within $\mathcal{H}_1$ and decoupled from other subspaces. Therefore, unitary dynamics do not affect discussions about steering. For simplicity and to avoid ambiguity in defining the infinite-time limit, we assume there are no unitary dynamics in $\operatorname{As}(\mathcal{H})$ and that the superoperators are strong.

\section{Conditions for steerability}\label{Discussions about steerability}

In the following, we first derive the necessary and sufficient conditions for the Frustration-Free Steerable (FFS) and Non-Frustration-Free Steadily Steerable (NFFSS) classes, and propose the Kitaev model as an example of an NFFSS class~\cite{kitaev2006anyons}. Subsequently, we derive the three necessary conditions in Conditions~(\ref{cond:1}-\ref{cond:3}) for the Non-Frustration-Free Jittery Steerable (NFFJS) class. The independence of these conditions indicates that all of them must be considered.

\subsection{Parent Hamiltonian and Frustration-Free Steering}\label{sec:Parent Hamiltonian and FF Steering}

\addtocounter{theorem}{-1}
\begin{theorem}
\label{thm:FFS_ap}
$H=\sum_i H_i$ is an FFS or NFFSS Hamiltonian if and only if there exists a local FF Hamiltonian $H_{\mathrm{PH}}$, such that $H_{\mathrm{PH}}$ have the same GS manifold as $H$. $H_{\mathrm{PH}}$ can be constructed from a set of trivial SCQs in GS subspace.
\end{theorem}

In the following, we provide the proof of Thm.~\ref{thm:FFS}. As noted in Def.~\ref{def:stabilizability}, only two possibilities are possible. Here, we focus on the scenario in which all GSs remain invariant under the action of all later-time superoperators. Suppose $\{\mathcal{P}_i\}$ is a set of superoperators constructed such that each GS of $H$ is invariant. In this case, the Hamiltonian $H$ is classified as either FFS or NFFSS.
\begin{proof}
Suppose that $|\psi\rangle$ is an invariant state of the superoperator $\mathcal{P}_i$, which acts non-trivially on $\mathcal{H}_{S}$, and let $\Pi_i$ denote the projection operator onto the invariant subspace of $\mathcal{P}_i$. We can perform a Schmidt decomposition of $|\psi\rangle$ as follows:
\begin{equation}
|\psi\rangle = \sum_j c_j^i |\phi_{S_i,j}\rangle \otimes |\varphi_{\bar S_i,j}\rangle, \quad c_j^i \neq 0,
\end{equation}
where $S_i$ is the support of $\mathcal{P}_i$, and $\bar S_i$ is the complement set of $S_i$. Since $\mathcal{P}_i(|\psi\rangle \langle \psi|) = |\psi\rangle \langle \psi|$, as given the definition of an FFS Hamiltonian, it follows that
\begin{equation}
    \mathcal{P}_i\big(|\phi_{S_i,j}\rangle \langle \phi_{S_i,k}|\big) = |\phi_{S_i,j}\rangle \langle \phi_{S_i,k}|.
\end{equation}
This implies that $\Pi_i |\phi_{S_i,j}\rangle = |\phi_{S_i,j}\rangle$. Considering all $\Pi_i$, we conclude that $|\psi\rangle$ is a ground state of the frustration-free parent Hamiltonian $H_{\mathrm{PH}} = \sum_i \mathbbm{1} - \Pi_i$.

To complete the proof, we need to show that no additional invariant states exist other than the GSs of $H_{\mathrm{PH}}$. According to the decomposition given in Eq.~\eqref{eq:decomposition}, it suffices to demonstrate that any pure state outside the ground subspace is not invariant. This implies that the ground subspace forms a maximal invariant sector in the decomposition of invariant subspaces of the global superoperators, constructed as the product of local superoperators $\{\mathcal{P}_i\}$. Equivalently, since these superoperators can be applied in arbitrary order, the asymptotic invariant state must remain invariant under each individual $\mathcal{P}_i$. Furthermore, we also need to show that no other mixed states are jointly invariant, ensuring that there is exactly one global jointly invariant sector.

Suppose $|\alpha\rangle$ is a pure state not belonging to the ground-state subspace. Since $\langle\alpha|H_{\mathrm{PH}}|\alpha\rangle > 0$, there exists at least one index $i$ such that $\langle\alpha|\Pi_i \otimes \mathbbm{1}|\alpha\rangle < 1$, indicating that $|\alpha\rangle$ is not an eigenstate of the projector $\Pi_i$. Consequently, its Schmidt decomposition can be expressed as
\begin{equation}
|\alpha\rangle= \sum_j c_j^{i} |\phi'_{S_i,j}\rangle\otimes|\varphi'_{\bar S_i,j}\rangle,\quad c_j^{i}\neq 0,
\end{equation}
and there exists some $j$ such that $\Pi_i|\phi'_{S_i,j}\rangle \neq |\phi'_{S_i,j}\rangle$. If $|\alpha\rangle\langle\alpha|$ is invariant under the operation $\mathcal{P}_i $, it would contradict our earlier assumption that the set $\{|\phi'_{S_i,j}\rangle\}_j$ is the maximal invariant subspace of $\mathcal{P}_i$, since $|\phi'_{S_i,j}\rangle$ should also be included in this subspace. Therefore, $|\alpha\rangle\langle\alpha|$ must not be an invariant state. This proves that the invariant sector spanned by GSs of parent Hamiltonian $H_{\mathrm{PH}}$ is maximal.

Furthermore, using a similar argument, we can perform a Schmidt decomposition on a mixed state $\rho$, or more precisely, a Schmidt decomposition on its double-ket representation:
\begin{equation}
|\rho\rangle\rangle = \sum_j c_j^i |M_{S_i,j}\rangle\rangle \otimes |N_{\bar S_i,j}\rangle\rangle.
\end{equation}
Equivalently, $\rho = \sum_j c_j^i M_{S_i,j} \otimes N_i$, where $M_{S_i,j} \in \mathcal{L}(\mathcal{H}_{S_i})$ and $\operatorname{Tr} M_{S_i,j}^\dagger M_{S_i,k} = \langle\langle M_{S_i,j} | M_{S_i,k} \rangle\rangle = \delta_{jk}$, with a similar condition for $N_{\bar S_i,j}$. Employing orthogonality and a similar line of reasoning, we conclude that $\rho$ is jointly invariant only if $\{M_{S_i,j}\}$ lie within $\operatorname{span}\{|\phi_{S_i,j}\rangle \langle\phi_{S_i,k}| \}_{j,k}$. Thus we prove that there are no other jointly invariant states.

\end{proof}

As a further example, we illustrate how the Kitaev model can be steered by applying Thm.~\ref{thm:FFS}. The steering process can be constructed exactly as in the analytical solution of the model. The Kitaev model is an exactly solvable spin-$1/2$ system on a hexagonal lattice, with the Hamiltonian given by~\cite{kitaev2006anyons} 
\begin{align*}
H = & -J_x \sum_{\text{$x$-links}} \sigma^x_j \sigma^x_k - J_y \sum_{\text{$y$-links}} \sigma^y_j \sigma^y_k \\
    & - J_z \sum_{\text{$z$-links}} \sigma^z_j \sigma^z_k,
\end{align*}
where the three directions of edges on the hexagonal lattice correspond to the three distinct types of links. Note that each plaquette of the lattice possesses conserved quantities $W_p = \sigma^x \otimes \sigma^y \otimes \sigma^z \otimes \sigma^x \otimes \sigma^y \otimes \sigma^z$, all commuting with each other and having eigenvalues $\pm 1$. Restricting to the common eigenspaces of the operators $W_p$ decomposes the total Hilbert space into a direct sum $\mathcal{H} = \bigoplus_{w_1,\cdots,w_m} \mathcal{H}_{w_1,\cdots,w_m}$, where each subspace $\mathcal{H}_{w_1,\cdots,w_m}$ corresponds to a fixed set of eigenvalues $W_1=w_1,\cdots,W_m=w_m$.

By introducing a representation for each spin-$1/2$ particle in terms of four Majorana operators $b^{x,y,z}$ and $c$, with Pauli operators defined as $\tilde{\sigma}^{x,y,z} = i b^{x,y,z} c$, the original Hamiltonian $H$ becomes a non-interacting fermionic Hamiltonian $\tilde{H}_{w_1,\cdots,w_m}$ in each subspace $\mathcal{H}_{w_1,\cdots,w_m}$. Consequently, the GS manifold becomes steerable, with the non-interacting fermionic Hamiltonian acting as a Frustration-Free parent Hamiltonian. Note that this representation introduces an additional $\mathbb{Z}_2$ gauge degree of freedom, which can be fixed by imposing the condition $D_j = b^{x}_j b^{y}_j b^{z}_j c = 1$. This constraint commutes with each Hamiltonian $\tilde{H}_{w_1,\cdots,w_m}$. Although the steering discussed here occurs in a Hilbert space including gauge degrees of freedom, the results can be directly mapped to the physical GS manifold by a local relabeling of states, preserving steerability.

Next, we discuss the inapplicability of Thm.~\ref{thm:FFS} to the one-dimensional Ising anti-ferromagnet introduced in Sec.~\ref{sec:NFFJS Examples}. The Hamiltonian of this model is given by $H = \sum_i H_i = \sum_i J \sigma^z_i \sigma^z_{i+1}$ with coupling strength $J = 1$, periodic boundary conditions, and an odd number $N$ of spins. The GSs can be denoted by $|m\pm\rangle = |\cdots \pm_{m-2} \mp_{m-1} \pm_{m} \pm_{m+1} \mp_{m+2} \pm_{m+3} \cdots\rangle$, where $\sigma^z_i |\pm_i\rangle = \pm |\pm_i\rangle$.

Although $H$ consists solely of $2$-local (two-qubit) Pauli operators, constructing a corresponding $2$-local parent Hamiltonian is not possible. This is because considering all two-qubit RDMs (not just nearest-neighbor pairs) derived from the GS manifold yields the complete two-spin Hilbert space spanned by ${|+_i +_j\rangle, |+_i -_j\rangle, |-_i +_j\rangle, |-_i -_j\rangle}$. Thus, no $2$-local SCQ exists that characterizes the GS manifold. Examining the nearest-neighbor three-spin RDMs, the GS manifold now spans the set n $\{|+_{i-1} -_{i} -_{i+1}\rangle, |+_{i-1} -_i +_{i+1}\rangle, |+_{i-1} +_i -_{i+1}\rangle, |-_{i-1} +_{i} +_{i+1}\rangle, |-_{i-1} +_{i} -_{i+1}\rangle, |-_{i-1} -_i +_{i+1}\rangle\}$. Thus, we can construct an SCQ as $\Pi_i = \mathbbm{1} - |+_{i-1} +_{i} +_{i+1}\rangle \langle +_{i-1} +_{i} +_{i+1}| - |-_{i-1} -_{i} -_{i+1}\rangle \langle -_{i-1} -_{i} -_{i+1}|$. However, this set of SCQs is not sufficient, since the GS manifold of the corresponding parent Hamiltonian $H_{\mathrm{PH}} = \sum_i \mathbbm{1} - \Pi_i$ also contains other states, such as $|+_1 +_2 -_3 -_4 +_5 \cdots\rangle$. Numerical calculations performed up to $N = 23$ indicate that one requires at least $\lceil (N + 4)/3 \rceil$-local nearest-neighbor SCQs to ensure that $H_{\mathrm{PH}}$ and $H$ share exactly the same GS manifold, $\{|m\pm\rangle\}_{m=1,2,\cdots,N}$. Hence, any local parent Hamiltonian $H_{\mathrm{PH}}$ either possesses a strictly larger GS manifold compared to $H$, or if it matches exactly, it must necessarily be non-local.

\subsection{Necessary conditions for the Non-Frustration-Free Jittery Steerable (NFFJS) class}\label{sec:Necessary conditions for Non-Frustration-Free Jittery Steerable (NFFJS) class}

In this section, we derive the necessary conditions for the NFFJS class. Since we are addressing an NFFJS scenario, we necessarily need to consider a multi-dimensional steerable target space—the GS manifold of a local Hamiltonian. We focus on steering protocols comprising local operators, or more generally, operators that act on a finite number of the system's degrees of freedom. This implies that, in constructing a steering protocol, we must design "local superoperators" involving finite segments of the system, using local information of the target manifold. The following necessary condition for NFFJS is based on this observation. We note, though, that further non-local features (e.g., long-range entanglement) may underlie the states that span our target space; it follows that a full characterization of the NFFJS Hamiltonian may require global information, which is beyond the scope of the present analysis (see discussions below).

The necessary conditions for the NFFJS class is presented in the following proposition.
\begin{prop}
    Given a steerable target subspace $\mathcal{H}_{\mathrm{GS}}\subset \mathcal{H}$, as described by Def.~\ref{def:stabilizability}, consider the situation that condition in Thm.~\ref{thm:FFS_ap} is not satisfied.  Then the following three conditions are satisfied:
    \begin{itemize}
        \item $\mathcal{H}_{\mathrm{GS}}$ is degenerate.
        \item There exists at least one trivial  SCQ inside one subspace of $\mathcal{H}_{\mathrm{GS}}$.
        \item $\mathcal{H}_{\mathrm{GS}}$ is bipartite indistinguishable.
    \end{itemize}
    Moreover, the second necessary condition is independent of the first and the third condition.
\end{prop}
\begin{proof}
Each local superoperator $\mathcal{P}_i$ must ensure that the system remains within the target manifold at later times; otherwise, Def.~\ref{def:stabilizability} would be violated. Explicitly, for an arbitrary ground state $\rho$, the resulting state under the action of a local superoperator, $\mathcal{P}_i(\rho) = \tau$, must also belong to the GS manifold. This requirement directly implies the first necessary condition: the GS manifold must be \textit{degenerate}.

Considering strong superoperators, there must exist at least two distinct ground states $\rho \neq \tau\in \mathcal{H}_{\mathrm{GS}}$ and at least one local superoperator $\mathcal{P}$ satisfying $\mathcal{P}(\rho) = \tau$ and $\mathcal{P}(\tau) = \mathcal{P}^2(\rho) = \mathcal{P}(\rho) = \tau$. Let us denote the support of the superoperator $\mathcal{P}$ as $S$. The condition $\mathcal{P}(\tau) = \tau$ implies that there exists a \textit{trivial SCQ} within a subspace of $\mathcal{H}_S$ containing $\tau$. This is the second necessary condition.

The third necessary condition emerges from the fact that a local superoperator $\mathcal{P}$, acting non-trivially on $\mathcal{H}_S$ and realizing a transformation between two distinct GSs, e.g., $\mathcal{P}(\rho)=\tau$, can be represented in terms of Kraus operators as $P(\rho)= \sum_i K_i \rho K_i^{\dagger}$. From this representation, it immediately follows that the RDMs of $\rho$ and $\tau$ on the complementary region $\bar S$ must be identical, i.e. $\operatorname{Tr}_S \tau= \operatorname{Tr}_S \sum_i K_i \rho K_i^{\dagger}=\operatorname{Tr}_S \sum_i  K_i^{\dagger} K_i \rho= \operatorname{Tr}_S \rho$. Thus, the RDMs of $\rho$ and $\tau$ coincide on the region $\bar S$, which typically constitutes a large region. This property contrasts with the definition of a bipartite distinguishable subspace provided in Def.~\ref{def:qgd}. The independence of these necessary conditions will be demonstrated later (see Prop.~\ref{prop:independent}).
\end{proof}

As stated above, here we introduce the notion of \textit{bipartite distinguishable} subspace. 
\begin{definition}
\label{def:qgd}    Given a Hilbert space $\mathcal{H}=\mathcal{H}_S\otimes\mathcal{H}_{\bar S}$, where $S$ is a local region, a subspace $\mathcal{H}_1\subseteq\mathcal{H}$ is called bipartite distinguishable, if for any two quantum states $\rho\neq\tau\in\mathcal{L}(\mathcal{H}_1)$, we have  $\operatorname{Tr}_S\rho\neq \operatorname{Tr}_S\tau$. 
\end{definition}

We now discuss how to determine whether a subspace is bipartite distinguishable and the implications of this property. For a target GS subspace $\mathcal{H}_{\mathrm{GS}} \subseteq \mathcal{H} = \mathcal{H}_S \otimes \mathcal{H}_{\bar S}$, suppose $\mathcal{H}_{\mathrm{GS}}$ is spanned by $\{|\psi_i\rangle\}$. The subspace $\mathcal{H}_{\mathrm{GS}}$ is bipartite distinguishable if and only if the set of operators $\{\operatorname{Tr}_S |\psi_i\rangle \langle \psi_j|\}_{i,j}$ are linearly independent operators on $\mathcal{H}_{\bar S}$. If there exist $\rho \neq \tau$ such that $\operatorname{Tr}_S \rho = \operatorname{Tr}_S \tau$, then $\operatorname{Tr}_S (\rho - \tau) = 0$, implying that $\operatorname{Tr}_S |\psi_i\rangle \langle \psi_j|$ are not linearly independent. Conversely, when $\operatorname{Tr}_S |\psi_i\rangle \langle \psi_j|$ are not linearly independent—say, there exists a linear combination $\sum_{ij} c_{ij} \operatorname{Tr}_S |\psi_i\rangle \langle \psi_j| = 0$ with some $c_{ij} \neq 0$—we can then choose $\rho = (1/d_{\mathrm{GS}}) \sum_i |\psi_i\rangle \langle \psi_i|$ and $\tau = \rho + \Delta \sum_{ij} c_{ij} \operatorname{Tr}_S |\psi_i\rangle \langle \psi_j|$, where $d_{\mathrm{GS}} = \operatorname{dim} \mathcal{H}_{\mathrm{GS}}$. With $1 \gg \Delta > 0$, $\tau$ is positive and satisfies $\operatorname{Tr}_S \rho = \operatorname{Tr}_S \tau$. In this spirit, we have the following proposition in a subspace $\mathcal{H}_{\mathrm{GS}} \subseteq \mathcal{H}$.

Aside from the proposition regarding steerability, note that the necessary conditions discussed previously are also closely connected to the existence of non-trivial SCQs, as we elaborate next.
\begin{prop}
\label{prop:nontrivial SCQ}
    If there exists non-trivial local SCQ $\mathcal{A}$ supported on local region $S$ in target subspace $\mathcal{H}_{\mathrm{target}}$, then the following three conditions are satisfied:
    \begin{itemize}
        \item $\mathcal{H}_{\mathrm{target}}$ is degenerate.
        \item There exists at least one trivial  SCQ inside one subspace of $\mathcal{H}_{\mathrm{target}}$.
        \item $\mathcal{H}_{\mathrm{GS}}$ is bipartite indistinguishable with respect to the bipartition $\mathcal{H}=\mathcal{H}_S\otimes \mathcal{H}_{\bar S}$, where $\bar S$ is the complement of  region $S$.
    \end{itemize}
The converse is not true, which means they are necessary but not sufficient conditions. Moreover, the second necessary condition is independent of the first and the third condition.
\end{prop}
\begin{proof}
Firstly, since $\mathcal{A}$ is a non-trivial local SCQ in $\mathcal{H}_{\mathrm{target}}$, which means that in this subspace $\mathcal{A}$ has eigenstates with different eigenvalues, $\mathcal{H}_{\mathrm{target}}$ must be at least two-dimensional and thus degenerate.

Secondly, let us derive the trivial SCQs for a given non-trivial SCQ. Without loss of generality, consider a non-trivial SCQ operator $\mathcal{A}$ with two distinct eigenvalues $a$ and $b$, defined with respect to the subspace $\mathcal{H}_{\mathrm{target}}$. Then it can be decomposed as  $\mathcal{A}=\sum_i a |a^i\rangle\langle a^i|+\sum_i b |b^i\rangle\langle b^i|=a\Pi_a + b\Pi_b$. Here, $a$ and  $b$ are eigenvalues of $\mathcal{A}$, and $\{|a^i\rangle\}_i,\{|b^i\rangle\}_i$ and $\Pi_a,\Pi_b$ are the corresponding eigenvectors and projection operators. Let us also denote $\Pi_{\mathrm{target}}$ as the projection operator onto $\mathcal{H}_{taarget}$. Since we have imposed that $\mathcal{A}$ is a non-trivial SCQ with respect to $\mathcal{H}_{\mathrm{target}}$, we can derive the following properties, by examining Def.~\ref{def:SCQ}: (i) $[\Pi_{\mathrm{target}},\mathcal{A}]=0$; (ii)$\Pi_{\mathrm{target}} [H,\mathcal{A}]\Pi_{\mathrm{target}}=0$; (iii) $\mathcal{A}$ is non-trivial implies $\Pi_{\mathrm{target}} \mathcal{A} \Pi_{\mathrm{target}}\not\propto \Pi_{\mathrm{target}}$, or in other words both $\{|a^i\rangle\}$ and $\{|b^i\rangle\}$ have overlap with $\mathcal{H}_{\mathrm{target}}$. In the following we will show that $\Pi_a$ is a trivial SCQ in (a subspace of) $\mathcal{H}_{\mathrm{target}}$. It is to be done by verifying Eq.~\eqref{eq:SCQ}.

We first prove Eq.~\eqref{eq:SCQcond1} is satisfied for $\Pi_a$. According to the property (i), $0=\langle a^i|[\mathcal{A},\Pi_{\mathrm{target}}]|b^j\rangle=(a-b)\langle a^i|\Pi_{\mathrm{GS}}|b^j\rangle$ implies that $\langle a^i|\Pi_{\mathrm{target}}|b^j\rangle=0$, hence $\langle a^i|[\Pi_a,\Pi_{\mathrm{GS}}]|b^j\rangle=\langle a^i|\Pi_{\mathrm{GS}}|b^j\rangle=0$. Similarly, one can show that  $\langle a^i|[\Pi_a,\Pi_{\mathrm{target}}]|a^j\rangle=\langle b^i|[\Pi_a,\Pi_{\mathrm{target}}]|b^j\rangle=0$. We thus infer that $[\Pi_{\mathrm{target}},\Pi_a]=0$, thus completing the proof of Eq.~\eqref{eq:SCQcond1}. 

Next, we address Eq.~\eqref{eq:SCQcond2}. According to Def.~\ref{def:SCQ} and specifically property (i), $[\Pi_{\mathrm{target}}, \mathcal{A}]=[\Pi_{\mathrm{target}}, H]=0$. The latter commutations imply two facts:  (a) we can define  
\begin{align*}
\mathcal{A}'&=\Pi_{\mathrm{target}} O\Pi_{\mathrm{target}} \\
&= a \Pi_{\mathrm{target}} \Pi_a\Pi_{\mathrm{target}} + b \Pi_{\mathrm{target}} \Pi_b \Pi_{\mathrm{target}};
\end{align*}
 it can be shown that $\Pi_{\mathrm{target}} \Pi_a\Pi_{\mathrm{target}}$ and $ \Pi_{\mathrm{target}} \Pi_b \Pi_{\mathrm{target}}$ are projection operators, such as 
\begin{align*}
(\Pi_{\mathrm{GS}} \Pi_a\Pi_{\mathrm{GS}})^2&=\Pi_{\mathrm{GS}} \Pi_a\Pi_{\mathrm{GS}}\Pi_{\mathrm{GS}} \Pi_a\Pi_{\mathrm{GS}}\\
&=\Pi_{\mathrm{GS}}\Pi_a=\Pi_{\mathrm{GS}} \Pi_a\Pi_{\mathrm{GS}};
\end{align*}
(b) We can show that 
\begin{align*}
    0&=\Pi_{\mathrm{target}}[H,\mathcal{A}]\Pi_{\mathrm{target}}\\
    &=[H,\Pi_{\mathrm{target}} \mathcal{A} \Pi_{\mathrm{target}}]=[H,\mathcal{A}'].
\end{align*}
 That leads to the conclusion that $H$ and $\mathcal{A}'$ commute. Similar to the reasoning above about $\mathcal{A}$ and $\Pi_{\mathrm{target}}$, we can infer that $[H, \Pi_{\mathrm{target}}\Pi_a\Pi_{\mathrm{target}}]=0$, hence 
 $$
 \Pi_{\mathrm{target}}[H,\Pi_a]\Pi_{\mathrm{target}}= [H, \Pi_{\mathrm{target}}\Pi_a \Pi_{\mathrm{target}}]=0.
 $$ 

Therefore, given that $\Pi_{\mathrm{target}}[H,\Pi_a]\Pi_{\mathrm{target}}=[\Pi_{\mathrm{target}},\Pi_a]=0$, the local projection operator $\Pi_a$ is also an SCQ with respect to $\mathcal{H}_{\mathrm{target}}$. If further for the whole image of $\Pi_a$, there does not exist eigenstate of $\Pi_{\mathrm{target}}$ with eigenvalue 0, i.e.  $\Pi_{\mathrm{target}} \Pi_a \Pi_{\mathrm{target}} = \Pi_{\mathrm{target}}$, then $\Pi_a$ is a trivial SCQ. 

If this is not the case, let us denote $\Pi_{\mathrm{target}}\Pi_a \Pi_{\mathrm{target}} = \Pi'_{\mathrm{target}}$. Obviously image of $\Pi'_{\mathrm{target}}$ is a subspace of $\Pi_a$. It can then be verified that
\begin{align*}
    \text{(i)}[\Pi'_{\mathrm{target}},H]=&[\Pi_{\mathrm{target}} \Pi_a \Pi_{\mathrm{target}}, H]\\
    =&\Pi_{\mathrm{target}}[H,\Pi_a]\Pi_{\mathrm{target}}\\
    =&0;\\
    \text{(ii)}[\Pi'_{\mathrm{target}},\Pi_a]=&[\Pi_{\mathrm{target}} \Pi_a \Pi_{\mathrm{target}},\Pi_a]\\
    =&\Pi_{\mathrm{target}} \Pi_a \Pi_{\mathrm{target}} \Pi_a\\
    &- \Pi_a \Pi_{\mathrm{target}} \Pi_a \Pi_{\mathrm{target}}\\
    =&\Pi'_{\mathrm{target}}-\Pi'_{\mathrm{target}}=0;\\
    \text{(iii)}\Pi'_{\mathrm{target}} \Pi_a \Pi'_{\mathrm{target}}=&\Pi_{\mathrm{target}} \Pi_a \Pi_{\mathrm{target}}\\
    &\times\Pi_a\Pi_{\mathrm{target}} \Pi_a \Pi_{\mathrm{target}}\\
    =&\Pi'_{\mathrm{target}};\\
    \text{(iv)}\Pi'_{\mathrm{target}}[H,\Pi_a]\Pi'_{\mathrm{target}}=&[H,\Pi'_{\mathrm{target}}\Pi_a\Pi'_{\mathrm{target}}]\\
    =&[H, \Pi'_{\mathrm{target}}]=0.
\end{align*}
Therefore we conclude $\Pi_a$ is a trivial SCQ with respect to the reduced subspace defined by $\Pi'_{\mathrm{target}}$. Such discussion can be generalized to SCQs with more eigenvalues.

Thirdly, when there exists a non-trivial SCQ $A$, namely $[\Pi_{\mathrm{target}}, \mathcal{A}] = 0$ and $\Pi_{\mathrm{target}} \mathcal{A} \Pi_{\mathrm{target}} \not\propto \Pi_{\mathrm{target}}$, then there exist two pure states $|a_1\rangle, |a_2\rangle \in \mathcal{H}_{\mathrm{target}}$ which are eigenstates of $\mathcal{A}$ with different eigenvalues, i.e., $\mathcal{A} |a_1\rangle = a_1 |a_1\rangle$, $\mathcal{A} |a_2\rangle = a_2 |a_2\rangle$ with $a_1 \neq a_2$. This implies that if we perform a Schmidt decomposition for $|a_1\rangle$ and $|a_2\rangle$ with respect to the bipartition $\mathcal{H} = \mathcal{H}_S \otimes \mathcal{H}_{\bar S}$, then each branch of the Schmidt vector of $|a_1\rangle$ on $\mathcal{H}_S$ is an eigenstate of $\mathcal{A}$ with eigenvalue $a_1$, and similarly for $|a_2\rangle$. Therefore, $\operatorname{Tr}_S (|a_1\rangle \langle a_2| + |a_2\rangle \langle a_1|) = 0$. This already implies that the two globally orthogonal states $\frac{1}{\sqrt{2}}(|a_1\rangle + |a_2\rangle)$ and $\frac{1}{\sqrt{2}}(|a_1\rangle - |a_2\rangle)$ are bipartite indistinguishable.
\end{proof}
Equivalently, if a subspace $\mathcal{H}_{\mathrm{target}}$ is bipartite distinguishable with respect to the bipartition $\mathcal{H} = \mathcal{H}_S \otimes \mathcal{H}_{\bar S}$, it implies that no SCQ or only trivial SCQs can exist on $\mathcal{H}_S$. Thus, \textit{bipartite distinguishability on $\bar S$ implies the absence of local distinguishability on $S$}, although the converse is not necessarily true. Here, by the absence of local distinguishability, we specifically mean that there is no local operator capable of distinguishing the target states as eigenstates corresponding to distinct eigenvalues.

We now present the proof demonstrating the independence of the necessary conditions.
\begin{prop}
\label{prop:independent}
    In the following, the second necessary condition is independent of the first and the third condition:
    \begin{itemize}
        \item $\mathcal{H}_{\mathrm{GS}}$ is degenerate.
        \item There exists at least one trivial  SCQ inside one subspace of $\mathcal{H}_{\mathrm{GS}}$.
        \item $\mathcal{H}_1$ is bipartite indistinguishable with respect to the bipartition $\mathcal{H}=\mathcal{H}_S\otimes \mathcal{H}_{\bar S}$.
    \end{itemize}
\end{prop}
\begin{proof}
By constructing explicit counterexamples, we demonstrate the independence of these conditions. First, we examine the first and second conditions. Consider the case where $\mathcal{H}_{\mathrm{GS}}$ is a single product state, for example, $|0\rangle^\otimes N$; in this scenario, the first condition is violated while the second condition is satisfied. Alternatively, if we define $\mathcal{H}_{\mathrm{GS}}$, consisting of $N$ qubits, as $|\psi_i\rangle = \frac{1}{2}(|0\rangle |\phi_{2i}\rangle + |1\rangle |\phi_{2i+1}\rangle)$ for $i = 0, 1$, where $|\phi_{i}\rangle$ are orthonormal states of $N-1$ qubits, then the second condition is violated while the first condition is satisfied.

Furthermore, the second and third conditions are independent, as illustrated by two counterexamples. The first counterexample demonstrates a scenario where a trivial SCQ exists within a subspace of $\mathcal{H}_{\mathrm{GS}}$, yet $\mathcal{H}_{\mathrm{GS}}$ is bipartite distinguishable. Specifically, consider $\mathcal{H}_{\mathrm{GS}} = {|\psi_0\rangle, |\psi_1\rangle}$, where $|\psi_0\rangle = |0\rangle_S \otimes |\alpha_0\rangle_{\bar S}$ and $|\psi_1\rangle = |0\rangle_S \otimes |\alpha_1\rangle_{\bar S}$, with ${}_{\bar S}\langle \alpha_0 | \alpha_1 \rangle_{\bar S} = 0$. It can be verified that $|0\rangle_S {}_S\langle 0|$ is a trivial SCQ in the subspace $\operatorname{span}\{|\psi_0\rangle\}$, while $\operatorname{Tr}_S \sum_{ij} \rho_{ij} |\psi_i\rangle \langle \psi_j| = \sum_{ij} \rho_{ij} |\alpha_i\rangle_{\bar S} {}_{\bar S}\langle \alpha_j| \neq 0$ for any non-zero coefficient matrix $\rho_{ij}$.

The second counterexample illustrates a case where no trivial SCQ exists within any subspace of $\mathcal{H}_{\mathrm{GS}}$, yet $\mathcal{H}_{\mathrm{GS}}$ is bipartite indistinguishable. Consider $\mathcal{H}_{\mathrm{GS}} = \operatorname{span}\{|\psi_0\rangle, |\psi_1\rangle, |\psi_2\rangle, |\psi_3\rangle\}$, with $\mathcal{H}_S = \operatorname{span}\{|i\rangle_S\}_{i=0,1,2}$ and $\mathcal{H}_{\bar S} = \operatorname{span}\{|i\rangle_{\bar S}\}_{i=0,1,\ldots,5,\ldots}$. We define the states as follows: 
\begin{align*} 
|\psi_0\rangle &= \frac{1}{\sqrt{3}} \left( |0\rangle_S |0\rangle_{\bar S} + |1\rangle_S |1\rangle_{\bar S} + |2\rangle_S |2\rangle_{\bar S} \right), \\ |\psi_1\rangle &= \frac{1}{\sqrt{3}} \left( |0\rangle_S |1\rangle_{\bar S} + |1\rangle_S |3\rangle_{\bar S} + |2\rangle_S |0\rangle_{\bar S} \right), \\ |\psi_2\rangle &= \frac{1}{\sqrt{3}} \left( |0\rangle_S |2\rangle_{\bar S} + |1\rangle_S |5\rangle_{\bar S} + |2\rangle_S |4\rangle_{\bar S} \right), \\ |\psi_3\rangle &= \frac{1}{\sqrt{3}} \left( |0\rangle_S |3\rangle_{\bar S} + |1\rangle_S |4\rangle_{\bar S} + |2\rangle_S |5\rangle_{\bar S} \right). 
\end{align*} 
Note that 
\begin{align*}
\operatorname{Tr}_S ( &|\psi_0\rangle\langle\psi_0| + |\psi_3\rangle\langle\psi_3| \\
&- |\psi_1\rangle\langle\psi_1| - |\psi_2\rangle\langle\psi_2| ) = 0,
\end{align*}
which implies that $\mathcal{H}_{\mathrm{GS}}$ is bipartite indistinguishable. Simultaneously, it can be verified that for any ground state $|\psi\rangle = \sum_i c_i |\psi_i\rangle$, the equation $\Pi_S |\psi\rangle = 0$ has no solution, where $\Pi_S$ is any potential local SCQ defined on $S$. Therefore, the third condition does not necessarily imply the possibility of transforming one ground state to another. This is because if we suppose there exists a local superoperator $\mathcal{P}$ such that
$$
\mathcal{P}(|\psi_0\rangle\langle\psi_0|+|\psi_3\rangle\langle\psi_3|-|\psi_1\rangle\langle\psi_1|-|\psi_2\rangle\langle\psi_2|)=0,$$
this equation leads to 
\begin{align*}
\mathcal{P}(|i\rangle_S {}_S\langle i|)=&\mathcal{P}(|i\rangle_S {}_S\langle i|),\\
\mathcal{P}(|i\rangle_S {}_S\langle j|)=&0,\quad\text{for } i\neq j \in\{0,1,2\}.
\end{align*}
In other words, $\mathcal{P}$ maps all states to the same state, destroying the coherence of the ground states. Thus, it cannot transform one ground state to another but instead replaces a ground state with a product of a final state on $S$ and some other state on $\bar S$.
\end{proof}

These examples indicate that all of the necessary conditions above must be considered. Here, we also discuss the incompleteness arising from strictly local conditions. Working in the Heisenberg picture, observables on $\mathcal{H}_S$, such as $O_S \otimes \mathbbm{1}_{\bar S}$, transform as desired into $\mathcal{P}_S^\dagger(O_S) \otimes \mathbbm{1}_{\bar S}$, while observables on the complementary subsystem $\bar S$ remain invariant, i.e., $\mathcal{P}_S^\dagger(\mathbbm{1}_S) \otimes O_{\bar S} = \mathbbm{1}_S \otimes O_{\bar S}$.
Our objective is to stabilize the GS manifold of a given local Hamiltonian. Practically, we often only have access to local system information, such as local energy of the Hamiltonian. During the steering process, only local information regarding the GS manifold can be exploited to design local superoperators—allowing us to control how this local information is transformed or preserved across different ground states. Such considerations lead directly to the previous necessary conditions. However, due to the presence of entanglement in quantum systems, certain ground states cannot be distinguished by purely local observables. Thus, a complete characterization of the NFFJS Hamiltonian may require global information about the system.
Furthermore, observables of the more general form $\sum_i O_{S,i} \otimes O_{\bar S,i}$ are transformed into $\sum_i \mathcal{P}_S(O_{S,i}) \otimes O_{\bar S,i}$, and cannot easily be converted into SCQs without detailed prior knowledge of their properties. Consequently, determining necessary and sufficient conditions for the steerability of the GS manifold $\mathcal{H}_{\mathrm{GS}}$ inevitably involves non-local information, which lies beyond the scope of the present paper.

\section{Proof for steerability of Commuting Pauli Hamiltonians}\label{app:Proof for Stabilizability of Commuting Pauli Hamiltonian}

In this section, we prove the steerability of commuting Pauli Hamiltonians. The discussion is organized into the following steps:
(i) we begin by analyzing the spectrum of a commuting Pauli Hamiltonian, establishing an equivalent, simplified binary representation of its spectrum; we also show that determining the GS energy of a commuting Pauli Hamiltonian is an NP-complete problem;
(ii) next, we introduce a family of steering superoperators and describe their action in the binary representation;
(iii) finally, we propose a simple random algorithm to generate a sequence of such superoperators, demonstrating that, in the asymptotic limit, the system converges to the GS manifold of a given commuting Pauli Hamiltonian. We further discuss alternative choices of steering superoperators, implementation via Clifford unitaries, and the robustness of the protocol under environmental decoherence.

\subsection{Spectrum of Commuting Pauli Hamiltonians}\label{app:Spectrum of Commuting Pauli Hamiltonian}

In the following, to establish the steerability of commuting Pauli Hamiltonians, we first introduce a lemma regarding their spectrum.

Commuting Pauli Hamiltonians are $n$-qubit Hamiltonians of the form $H = \sum_i H^{(i)}$, where each term $H^{(i)}$ belongs to the Pauli group and satisfies $[H^{(i)}, H^{(j)}] = 0$ for all $i,j$. Since each eigenstate of $H$ is simultaneously an eigenstate of every individual operator $H^{(i)}$, the eigenstates of $H$ can be uniquely labeled by the corresponding eigenvalues of the operators ${H^{(i)}}$. This leads naturally to a binary representation of the spectrum of $H$. The lemma is stated as follows:

\begin{lemma}
The spectrum of a $n$-qubit commuting Pauli Hamiltonian $H$ one-to-one corresponds to the weights of a subspace over $\mathbb{Z}_2$ field. This subspace is purely determined by $H$.
\end{lemma}
Since we are considering commuting Pauli Hamiltonians, there exists a complete set of eigenstates of $H$ such that each eigenstate is simultaneously an eigenstate of every local term $H^{(i)}$, with eigenvalues $E^{(i)} = \pm 1$. These eigenvalues collectively define a one-dimensional linear representation of the abelian group generated by the set $\{H^{(i)}\}$. Consequently, the spectrum of $H$ can be fully characterized by analyzing this linear representation.

\begin{proof}
The proof proceeds in two steps. First, we derive a system of linear equations for the eigenvalues $E^{(i)}$, which define a one-dimensional linear representation of the group generated by the local terms $\{H^{(i)}\}$. Second, we show that every solution to these equations corresponds uniquely to a valid set of eigenenergies of $\{H^{(i)}\}$.

First, we derive the constraints on the eigenvalues $E^{(i)}$. Simply speaking, there may be commuting but dependent local eigenenergies. For instance, consider $H^{(1)} = \sigma_z \otimes \sigma_z \otimes I$, $H^{(2)} = I \otimes \sigma_z \otimes \sigma_z$, and $H^{(3)} = \sigma_z \otimes I \otimes \sigma_z$. Then $H^{(1)} H^{(2)} H^{(3)} = I$ implies $E^{(1)} E^{(2)} E^{(3)} = 1$. Such constraints can be converted into linear equations. To study $\{E^{(i)}\}$, which forms a linear representation of the group generated by $H^{(i)}$, we first examine this group.
We begin by deriving constraints on the eigenvalues $E^{(i)}$. In general, although the terms $\{H^{(i)}\}$ commute, their eigenvalues may not be independent. For example, consider the local terms $H^{(1)} = \sigma_z \otimes \sigma_z \otimes I$, $H^{(2)} = I \otimes \sigma_z \otimes \sigma_z$, and $H^{(3)} = \sigma_z \otimes I \otimes \sigma_z$. These satisfy the relation $H^{(1)} H^{(2)} H^{(3)} = I$, which imposes the constraint $E^{(1)} E^{(2)} E^{(3)} = 1$ on any simultaneous eigenstate. Such multiplicative constraints can be systematically translated into linear equations over $\mathbb{Z}_2$. To analyze the eigenvalue structure $\{E^{(i)}\}$, we thus begin by examining the group generated by the $H^{(i)}$.

Suppose there are $n$ qubits, and $N = \operatorname{poly}(n)$ local terms in $H$, as is typical for a local Hamiltonian. We introduce an Abelian subgroup of the Pauli group generated by $H^{(i)}$ and $-1$, denoted as the Hamiltonian group $\mathbb{H}$. Formally, $\mathbb{H}$ can be written as 
\begin{equation}
\label{eq:Hgroup}
    \mathbb{H}=\{g_1^{i_1} g_2^{i_2} \cdots g_N^{i_N}| i_{j} \in \{0,1\}, g_j\in\{-1,H^{(1)},\cdots\}\}.
\end{equation}
However, the generators $H^{(i)}$ are not independent. According to fundamental theorem of homomorphism~\cite{hungerford2012algebra}, an arbitrary group $G$ is fully specified by its generators $GE=\{g_1,\cdots g_p\}$ and certain constraints called relations, written in the form 
    $$
    r=g_1^{i_1} g_2^{i_2} \cdots g_N^{i_N}=1.
    $$
Therefore, if we can determine the relations of $\mathbb{H}$, we can fully specify $\mathbb{H}$. This can be accomplished using linear equations.

Each Pauli observable $H^{(i)}$ can be represented by a string of $2n$ binary numbers $v^{(i)}_{\mathrm{P}} = (a^{(i)}_1, \cdots, a^{(i)}_n; b^{(i)}_1, \cdots, b^{(i)}_n)_{\mathrm{P}}$ when it is proportional to $\sigma_x^{a^{(i)}_1} \sigma_z^{b^{(i)}_1} \otimes \cdots \otimes \sigma_x^{a^{(i)}_n} \sigma_z^{b^{(i)}_n}$ up to a phase factor. The subscript $P$ indicates that it encodes the powers of the Pauli matrices. Note that the index within parentheses denotes the label of the local term $H^{(i)}$, while the indices without parentheses refer to the qubit sequence number. For a relation $H^{(i_1)} \times \cdots \times H^{(i_k)} = \pm 1$, this is equivalent to
$$
    v^{(i_1)}_{\mathrm{P}}\oplus\cdots\oplus v^{(i_k)}_{\mathrm{P}}=0,
$$
where $\oplus$ denotes addition modulo 2. By representing the Hamiltonians as an $N \times 2n$ binary matrix $M$, the relations can be found as the left null space of this matrix:
    \begin{align*}
         0=&c_{\mathrm{H}}\times M_{\mathrm{P}}=(c_1,c_2,\cdots,c_N)_{\mathrm{H}}\\
         &\times\begin{pmatrix}
    a_1^{(1)} & \cdots & a_n^{(1)} & b_1^{(1)} &\cdots & b_n^{(1)}\\
    \vdots    &        &            &           &       & \vdots\\
    a_1^{(N)} & \cdots & a_n^{(N)} & b_1^{(N)} &\cdots & b_n^{(N)}
    \end{pmatrix}_{\mathrm{P}},
    \end{align*}   
where the subscript $H$ indicates that it represents the powers of $H^{(i)}$. All possible solutions $c_{\mathrm{H}}$ to the above linear equation encompass all relations.

We represent this linear null subspace with its independent bases $\{c_{\mathrm{H}}^{(i)}\in \mathbb{Z}_2^N, i=1,2,\cdots,d_r\}$. Note that the product of the $H^{(i)}$ could yield either $-1$ or $1$. The actual phases for each relation can be represented by a $d_r$-dimensional binary vector $p_{\mathrm{ph}}$ as
$$
    \prod_{j=1}^N(H^{(j)})^{c_{j,\mathrm{H}}^{(i)}}=(-1)^{p_{i,\mathrm{ph}}},
$$
where the subscript denotes that it encodes the phases of the relations. The relations of $\mathbb{H}$ are thus fully characterized by all possible $c_{\mathrm{H}}$ and the corresponding $p_{\mathrm{ph}}$.

Next, we derive the constraints on $E^{(i)}$ based on the constraints on $H^{(i)}$ described above. We encode the eigenvalues of a given eigenstate as an $N$-dimensional binary vector $v_{\mathrm{E}}$, where the subscript $E$ indicates that it encodes information about the eigenenergies, with $E^{(i)} = (-1)^{v_{i,\mathrm{E}}}$. The relation above then reads
    $$
    \prod_{j=1}^N(E^{(j)})^{c_{j,\mathrm{H}}^{(i)}}=(-1)^{\boldsymbol{c_{\mathrm{H}}^{(i)}}\cdot\boldsymbol{v_{\mathrm{E}}}}=(-1)^{p_{i,\mathrm{ph}}}.
    $$
If we rewrite the relations as a $d_r \times N$ binary matrix $C_{ij,\mathrm{H}}=c^{(i)}_{j,\mathrm{H}}$, then the energy vector must satisfy
    \begin{equation}
        C_{\mathrm{H}}\times v_{\mathrm{E}}^T= p_{\mathrm{ph}}^T,
        \label{EnergyEquation}
    \end{equation}
where the addition is defined over the field $\mathbb{Z}_2$. This equation provides the constraint on the eigenenergies $E^{(i)}$.

The second step is to prove that for each solution $v_{\mathrm{E}}$, there exists a corresponding eigenstate. In other words, each solution of this constraint corresponds one-to-one with a set of eigenenergies. For an arbitrary $N$-dimensional vector $v_{\mathrm{E}}$, the corresponding projection operator onto the subspace assigning energies according to $v_{\mathrm{E}}$ is
    \begin{align*}
        &\frac{1}{2^N}\prod_{i=1}^N (1+(-1)^{v_{i,\mathrm{E}}}H^{(i)})\\
        =& \frac{1}{2^N}\sum_{\boldsymbol{b}\in\{0,1\}^N} (-1)^{\boldsymbol{b}\cdot \boldsymbol{v}_{E}} \prod_{i=1}^N (H^{(i)})^{b_i}.
    \end{align*}
Since all local terms commute, the square of the operator above is itself, confirming that it is indeed a projection operator. The eigenstates of this projection operator possess the energy vector $v_{\mathrm{E}}$. Whether the corresponding subspace has non-zero dimension can be determined by its trace, which is non-zero if and only if the product $\prod_{i=1}^N (H^{(i)})^{b_i} \propto 1$. This holds true if and only if $\boldsymbol{b} \in \operatorname{span}\{c^{(i)}_{\mathrm{H}}\}$, or equivalently, $\boldsymbol{b} = \boldsymbol{x} \times C_{\mathrm{H}}$, where $\boldsymbol{x} \in \mathbb{Z}_2^{d_r}$; that is, $\boldsymbol{b}$ is a vector representing possible relations. The trace can thus be computed as follows:
    \begin{align*}
         &\frac{1}{2^N}\operatorname{Tr}\sum_{\boldsymbol{b}\in\{0,1\}^N} (-1)^{\boldsymbol{b}\cdot \boldsymbol{v}_{E}} \prod_{i=1}^N (H^{(i)})^{b_i}\\
         =&\sum_{\boldsymbol{x}\in\{0,1\}^{d_r}} (-1)^{\boldsymbol{x}\times C_{\mathrm{H}} \times \boldsymbol{v}_{\mathrm{E}}^T}\times (-1)^{\boldsymbol{x}\cdot\boldsymbol{p}_{\mathrm{ph}}}.
    \end{align*}
Therefore, if $v_{\mathrm{E}}$ is the solution of Eq.~\eqref{EnergyEquation}, the trace is simply $\sum_{\boldsymbol{x}}1=2^{d_r}$. Otherwise, it is $\sum_{\boldsymbol{x}}(-1)^{\boldsymbol{x}\cdot\boldsymbol{y}}=0$, where $\boldsymbol{y}=C_{\mathrm{H}}\times \boldsymbol{v}_{\mathrm{E}}^T+\boldsymbol{p}_{\mathrm{ph}}^T\neq 0$.   Thus, only if $v_{\mathrm{E}}$ is a solution of Eq.~\eqref{EnergyEquation} do there exist $2^{d_r}$ distinct eigenstates with the corresponding eigenenergies.

In conclusion, solving Eq.~\eqref{EnergyEquation} provides complete information about the spectrum of $H$. Let $\operatorname{span}\{a_{\mathrm{E}}^{[i]\} \in \mathbb{Z}_2^N,\ i = 1, \cdots, N - d_r}$ denote a basis for the right null space of $C_{\mathrm{H}}$, and let $s_{\mathrm{E}}$ be a particular solution to Eq.~\eqref{EnergyEquation}. Then the full spectrum of the Hamiltonian is obtained by evaluating all $2^{N - d_r}$ solutions of the form $v_{\mathrm{E}} = s_{\mathrm{E}} + \sum_{i=1}^{N - d_r} x_i a_{\mathrm{E}}^{[i]}$ where $x \in \mathbb{Z}_2^{N - d_r}$ are binary coefficients. Each such binary $N$-bit string encodes a valid eigenvalue configuration, with entries $0$ and $1$ corresponding to $E = +1$ and $E = -1$, respectively. Counting the number of $0$s and $1$s in each solution yields the spectrum of $H$.
\end{proof}

Solving for the ground-state energy of a commuting Pauli Hamiltonian can therefore be reduced to searching within the solution space of Eq.~\eqref{EnergyEquation} for a binary vector with the maximal number of $1$ entries. Equivalently, this is the well-known Maximum-Likelihood Decoding Problem, which has been proven to be NP-complete~\cite{NPC1055873}. In our context, this corresponds to transforming $p_{\mathrm{ph}} \rightarrow p_{\mathrm{ph}} + C_{\mathrm{H}} \times (1, \cdots, 1)^\mathrm{T}$ and then identifying the binary string in the solution space with the maximal number of $0$ entries. This reduction is manifestly polynomial in time, leading to the following proposition.

\begin{prop}
\label{prop:npc}
    Determining the ground-state energy of a commuting Pauli Hamiltonian is NP-complete. 
\end{prop}

Since this is an NP-complete problem, it is widely believed that no efficient algorithm exists for solving it in general, even on a quantum computer. Therefore, our goal is not to establish quantum advantage in solving this problem. Rather, the key insight is that it is possible to steer a quantum system toward its ground state without requiring explicit knowledge of that state—though the process may, in the worst case, require exponential time.

\subsection{Steering operators}\label{app:Steering Operators}

In this section, we introduce the set of steering operators for commuting Pauli Hamiltonians and analyze their effects in the Heisenberg picture. Our focus is on the evolution of the energy expectation value throughout the steering process. We begin by characterizing the evolution of the Hamiltonian group, denoted $\mathbb{H}$ (cf. Eq.~\eqref{eq:Hgroup}), under the action of these steering operators, and subsequently represent the transformations using a binary matrix formalism.

The superoperators are chosen as follows:
\begin{align}
\label{eq:sup_op_CPH}
    \mathcal{P}^{(i)}_{\pm}(\cdot;V_{+})&=\frac{\mathbbm{1}\pm H^{(i)} }{2} \cdot \frac{\mathbbm{1}\pm H^{(i)} }{2}\nonumber\\
    &+ \frac{\mathbbm{1}\mp H^{(i)} }{2} V^\dagger_{+} \cdot V_{+}\frac{\mathbbm{1}\mp H^{(i)} }{2},\nonumber\\
    \mathcal{P}^{(i)}_{\pm}(\cdot;V_{-})&=\frac{\mathbbm{1}\pm H^{(i)} }{2} \cdot \frac{\mathbbm{1}\pm H^{(i)} }{2}\nonumber\\
    &+ \frac{\mathbbm{1}\pm H^{(i)} }{2} V^\dagger_{-} \cdot V_{-}\frac{\mathbbm{1}\pm H^{(i)} }{2},
\end{align}
where $V_{\pm}=V_{\pm}^\dagger$ are some Pauli operators and $V_{\pm} H^{(i)} V_{\pm}=\pm H^{(i)}$. 

Intuitively, the effect of $\mathcal{P}^{(i)}_{\pm}$ keeps the subspace $(1\pm H^{(i)})/2$ invariant and   flip the orthogonal subspaces with $V_{\pm}$. Thus $\mathcal{P}^{(i)}_-(\cdot; V_-)$ locally cools $H^{(i)}$ while $\mathcal{P}^{(i)}_+(\cdot; V_-)$ locally heats $H^{(i)}$. This is manifested in the following equations:
Intuitively, the effect of $\mathcal{P}^{(i)}_{\pm}$ is to leave the subspace $(1 \pm H^{(i)})/2$ invariant while flipping the orthogonal subspace using $V_{\pm}$. Therefore, $\mathcal{P}^{(i)}-(\cdot; V-)$ acts to locally cool $H^{(i)}$, whereas $\mathcal{P}^{(i)}_+(\cdot; V-)$ locally heats $H^{(i)}$. This behavior is captured in the following equations:
\begin{align*}
    \mathcal{P}^{(i)\dagger}_-(H^{(i)};V_-)&=-\mathbbm{1},\\
    \mathcal{P}^{(i)\dagger}_+(H^{(i)};V_-)&=\mathbbm{1}.
\end{align*}

Since $V_{\pm}$ either commute or anti-commute with each $H^{(i)}$, this can be represented by a $N$-dimensional binary vector $g_{i,\mathrm{VH}}$ as
$$
VH^{(i)}V=(-1)^{g_{i,\mathrm{VH}}}H^{(i)}.
$$
Then the effect of superoperators on Hamiltonians can be simply written as
$$
    \mathcal{P}^{(j)\dagger}_{\pm}(H^{(i)})=(\pm H^{(j)})^{g_{j,\mathrm{VH}}}H^{(i)}.
$$
Noteworthy, $\mathbb{H}$ (cf. Eq.~\eqref{eq:Hgroup}) is closed under the application of $\mathcal{P}_{\pm}$ as
\begin{align}
&\mathcal{P}^{(j)\dagger}_{\pm}(H^{(i_1)}H^{(i_2)}\cdots H^{(i_n)})\nonumber\\
=&(\pm H^{(j)})^{g_{i_1,\mathrm{VH}}+\cdots+g_{i_n,\mathrm{VH}}}H^{(i_1)}H^{(i_2)}\cdots H^{(i_n)}.
\end{align}
This is the reason we introduced $\mathbb{H}$, and in the following, we will keep track of $H^{(i)}$ under a sequence of steering operators. Above we have ignored the dependence on $V_{\pm}$ since using $g_{\mathrm{VH}}$ is more convenient, and whether we use $V_+$ or $V_-$ depends on $g_{i,\mathrm{VH}}$. Nevertheless, the value of $g_{\mathrm{VH}}$ cannot be arbitrarily assigned, as we shall see in the following lemma. 
This is precisely why we introduced $\mathbb{H}$—to track the evolution of each $H^{(i)}$ under a sequence of steering operators. In the discussion above, we have omitted the explicit dependence on $V_{\pm}$, as it is more convenient to work with the binary vector $g_{\mathrm{VH}}$. The choice between $V_+$ and $V_-$ is determined by the value of $g_{i,\mathrm{VH}}$. However, as we will show in the following lemma, the values of $g_{\mathrm{VH}}$ cannot be assigned arbitrarily.

\begin{lemma}
\label{lemma:V}
    For arbitrary Pauli operator $V$, with its commuting relation with $H^{(i)}$ represented by vector $g_{\mathrm{VH}}$, must satisfy that $C_{\mathrm{H}} \times g_{\mathrm{VH}}^T=0$, where the matrix $C_{\mathrm{H}}$ is defined in Appendix~\ref{app:Spectrum of Commuting Pauli Hamiltonian}. Conversely for arbitrary $g_{\mathrm{VH}}$ in right null space of $C_{\mathrm{H}}$, we can construct a Pauli operator $V$.
\end{lemma}

This lemma arises from the fact that the group $\mathbb{H}$ contains the identity operator, i.e. the relations, and any operator $V$ must commute with the identity. We defer the proof of this lemma to Appendix~\ref{sec:Locality of superoperator}, where we also discuss the locality constraints involved. For now, let us proceed to analyze the evolution of the Hamiltonian terms $H^{(i)}$ in the Heisenberg picture.

As we have seen, the action of $\mathcal{P}^{(i)\dagger}_\pm$ simply maps one element of $\mathbb{H}$ to another. This transformation process terminates when the element becomes $\pm \mathbbm{1}$, since $\mathcal{P}^{(i)\dagger}_\pm(\mathbbm{1}) = \mathbbm{1}$. This implies that all initial states are ultimately steered into a subspace where each $H^{(i)}$ has a definite expectation value. To analyze this systematically, it is useful to decompose elements of $\mathbb{H}$ in terms of independent generators (see below), such that all nontrivial generators can eventually be mapped to the identity operator. Fortunately, this evolution can be described as a sequence of linear operations over the finite field $\mathbb{Z}_2$.

Without loss of generality, we consider that 
$$
H^{(i)}=(H^{(1)}\cdots H^{(i-1)})^0 (H^{(i)})^1 (H^{(i+1)}\cdots H^{(N)})^0
$$ 
can be represented by a vector $x_{\mathrm{H}}$ with $x_{j,\mathrm{H}}=\delta_{ij}$. As a further example, the operator $H^{(i_1)}H^{(i_2)}$ can be represented by a vector $x_{j,\mathrm{H}}=\delta_{i_1 j}+\delta_{i_2 j}$. Therefore, all of local Hamiltonian terms $\{H^{(i)}\}$ can be represented by a $N\times N$ identity matrix $\mathbbm{1}_{N\times N,\mathrm{H}}$ with its $i$th row representing $H^{(i)}$, where the subscript $H$ indicates that it represents the powers of $\{H^{(i)}\}$ operator. Without loss of generality, we can assume $C_{\mathrm{H}}$ is a upper triangular matrix (otherwise we can make use of Gaussian elimination and permute orders of $H^{(i)}$), so the orthogonal subspace of row subspace of $C_{\mathrm{H}}$ can have a simple form, spanned by row subspace of $(N-d_r)\times N$ matrix $D_{\mathrm{H}}$ with $D_{ij,\mathrm{H}}= \delta_{j-i,d_r}$. In this bases, we have the following decomposition
\begin{widetext}
        \begin{align}
        \mathbbm{1}_{N\times N,\mathrm{H}}=&\begin{pmatrix}
        1 & &\\
          &\ddots&\\
          & &1
        \end{pmatrix}_{\mathrm{H}}=\left(\begin{array}{@{}c|c@{}}
  B^T_{0,\mathrm{C}}&A^T_{0,\mathrm{C}}\\
\end{array}\right)\left(\begin{array}{c}
  C_{\mathrm{H}}\\
\hline
  D_{\mathrm{H}}
\end{array}\right)\nonumber\\
=&\left(\begin{array}{@{}c|c@{}}
  M^{-1}&-M^{-1}K\\
  \hline\\
   & \mathbbm{1}_{(N-d_r)\times (N-d_r)}
\end{array}\right)_{\mathrm{C}}\left(\begin{array}{@{}c|c@{}}
  M&K\\
  \hline\\
   & \mathbbm{1}_{(N-d_r)\times (N-d_r)}
\end{array}\right)_{\mathrm{H}}=\left(\begin{array}{c}
  C_{\mathrm{H}}\\
\hline
  D_{\mathrm{H}}
\end{array}\right)\left(\begin{array}{@{}c|c@{}}
  B^T_{0,\mathrm{C}}&A^T_{0,\mathrm{C}}\\
\end{array}\right),
    \end{align}
\end{widetext}
where the subscript 0 denotes this is the initial state and subscript $C$ denotes that it is a coefficient matrix. $M$ is a $d_r\times d_r$ upper triangular matrix, which can be given by $C_{\mathrm{H}}$, thus $M^{-1}$ is also an upper triangular matrix. Noteworthy, the column subspaces of $A^T_{0,\mathrm{C}}$ or the row subspaces of $A_{0,\mathrm{C}}$ is just the right null subspace of $C_{\mathrm{H}}$, and columns of $B^T_{0,\mathrm{C}}$, denoted as $\{\boldsymbol{b}^{[i]}_{\mathrm{C}}, i=1,\cdots,d_r\}$, spans the orthogonal subspaces such that $(C_{\mathrm{H}}\times \boldsymbol{b}^{[i]}_{\mathrm{C}})_j=\delta_{ij}, i,j=1,\cdots, d_r$. In the following we denote $a^{[i]}_{\mathrm{C}}$, $b^{[i]}_{\mathrm{C}}$ as rows of $A_{\mathrm{C}}$, $B_{\mathrm{C}}$ respectively while $a^{(i)}_{\mathrm{C}}$ and $b^{(i)}_{\mathrm{C}}$ as columns of $A_{\mathrm{C}}$ and $B_{\mathrm{C}}$. In other words, they can be represented as follows:
\begin{equation*}
    \left(\begin{array}{c}
  B_{\mathrm{C}}\\
\hline
  A_{\mathrm{C}}
\end{array}\right)=\left(\begin{array}{c}
  b_{\mathrm{C}}^{[1]}\\
  \vdots\\
  b_{\mathrm{C}}^{[d_r]}\\
\hline
  a_{\mathrm{C}}^{[1]}\\
  \vdots\\
  a_{\mathrm{C}}^{[N-d_r]}
\end{array}\right)=\left(\begin{array}{ccc}
  b_{\mathrm{C}}^{(1)} &\cdots &b_{\mathrm{C}}^{(N)}\\
\hline
   a_{\mathrm{C}}^{(1)} &\cdots &a_{\mathrm{C}}^{(N)}
\end{array}\right).
\end{equation*}

Without loss of generality, for an given element from $\mathbb{H}/\{1,-1\}$(cf. Eq.~\eqref{eq:Hgroup}), it is sufficient to use coefficient vectors with respect to the relation and non-trivial generator of $\mathbb{H}$, which are represented with the matrix $C_{\mathrm{H}}$ and $D_{\mathrm{H}}$, respectively. Accordingly, the coefficient vector can be separated into $b$ part, representing the trivial or relation part, and $a$ part, representing the non-trivial part, and they are $d_r$-dimensional, $N-d_r$-dimensional, respectively. For instance, $H^{(i)}$ is represented by the  vector $x_{\mathrm{H}}=(b^{(i)}_{0,\mathrm{C}},a^{(i)}_{0,\mathrm{C}})\times (C^T_{\mathrm{H}},D^T_{\mathrm{H}})^T$.

At the same time, since $\mathcal{P}_-$ can potentially introduce a minus sign, we also need a $N$-dimensional vector tracking for this dynamical phase, say $p_{\mathrm{dp}}$. So during the process, we can trace the dynamics of each local term $H^{(i)}$ with the coefficient vector and phase number, say it is $(b^{(i)}_{\mathrm{C}},a^{(i)}_{\mathrm{C}})$ and $p_{i,\mathrm{dp}}$ for $H^{(i)}$ at some time. Here we use parentheses for indices to emphasize that it ranges from $1$ to $N$. In summary, the dynamical phase and powers of Hamiltonians for $\mathcal{P}^{(j_1)\dagger}(\mathcal{P}^{(j_{2})\dagger}(\cdots\mathcal{P}^{(j_q)\dagger}(H^{(i)})))$ is $-1^{p_{i,\mathrm{dp}}}$ and $(b^{(i)}_{\mathrm{C}},a^{(i)}_{\mathrm{C}})\times (C^T_{\mathrm{H}},D^T_{\mathrm{H}})^T$. 

Then let us determine whether an element $H$ from $\mathbb{H}$ commutes or anti-commutes with a given Pauli operator $V$. Let us assume $H$ is represented by the coefficients $(b_{\mathrm{C}},a_{\mathrm{C}})$. In this situation, note that for any given Pauli operator $V$ represented with $g_{\mathrm{VH}}$, $C_{\mathrm{H}}\times g_{\mathrm{VH}}^T=0$ implies that $g_{\mathrm{VH}}$ lies in the subspace spanned by $a_{0,\mathrm{C}}^{[i]}$, see also Lemma~\ref{lemma:V}. Therefore, we can rewrite it as the linear combination $\boldsymbol{g}_{\mathrm{VH}}=\sum_i e_{i,\mathrm{CVH}} a^{[i]}_{0,\mathrm{C}}$, where $e_{\mathrm{CVH}}$ is a $N-d_r$-dimensional binary vector as a coefficient vector for the bases $\{a^{[i]}_{0,\mathrm{C}}\}$. Then whether $H$ commutes or anti-commutes with $V$ is determined by
\begin{align*}
    &\left(\begin{array}{@{}c|c@{}}
  b_{\mathrm{C}}&a_{\mathrm{C}}\\
\end{array}\right)\left(\begin{array}{c}
  C_{\mathrm{H}}\\
\hline
  D_{\mathrm{H}}
\end{array}\right)\left(\begin{array}{@{}c|c@{}}
  B^T_{0,\mathrm{C}}&A^T_{0,\mathrm{C}}\\
\end{array}\right)\left(\begin{array}{c}
  0\\
\hline
  e_{\mathrm{CVH}}^T
\end{array}\right)\\
=&\left(\begin{array}{@{}c|c@{}}
  b_{\mathrm{C}}&a_{\mathrm{C}}\\
\end{array}\right)\left(\begin{array}{c}
  0\\
\hline
  e_{\mathrm{CVH}}^T
\end{array}\right)\\
=& \boldsymbol{a}_{\mathrm{C}}\cdot \boldsymbol{e}_{\mathrm{CVH}},
\end{align*}
which is $0$ if commuting and $1$ otherwise. 

Given the discussions above, let us now summarize the effect of a superoperator $\mathcal{P}^{(j)}_{\pm}$, which we characterize using a three-tuple $\{(j, \pm, e_{\mathrm{CVH}})\}$, where $j \in \{1, \cdots, N\}$ denotes the Hamiltonian index, and the Pauli operator $V_{\pm}$ associated with the superoperator is specified by the binary vector $e_{\mathrm{CVH}}$ (see the discussion above). We consider an element $H \in \mathbb{H}$ (cf. Eq.~\eqref{eq:Hgroup}) characterized by its coefficient vectors $b_{\mathrm{C}}$, $a_{\mathrm{C}}$, and a dynamical phase $p_{\mathrm{dp}}$ (a scalar). Under the action of the superoperator, these components transform as follows:
\begin{align}
    &H\rightarrow \mathcal{P}^{(j)}_{\pm}(H),\\
    &b_{\mathrm{C}} \rightarrow b_{\mathrm{C}}+(\boldsymbol{a}_{\mathrm{C}}\cdot \boldsymbol{e}_{\mathrm{CVH}})b^{(j)}_{0,\mathrm{C}},\\
    &a_{\mathrm{C}} \rightarrow a_{\mathrm{C}}+(\boldsymbol{a}_{\mathrm{C}}\cdot \boldsymbol{e}_{\mathrm{CVH}})a^{(j)}_{0,\mathrm{C}}\label{a_change},\\
    \label{p_change}&p_{\mathrm{dp}} \rightarrow\begin{cases}
        p_{\mathrm{dp}}, & \text{if with }\mathcal{P}^{(j)}_+;\\
        p_{\mathrm{dp}}+(\boldsymbol{a}_{\mathrm{C}}\cdot \boldsymbol{e}_{\mathrm{CVH}}),& \text{if with }\mathcal{P}^{(j)}_-.
    \end{cases}
\end{align}

Let us now determine how the set of Hamiltonians $\{H^{(i)}\}$ evolves in the Heisenberg picture under the action of a sequence of steering superoperators $\{\mathcal{P}^{(j_1)}, \mathcal{P}^{(j_2)}, \cdots, \mathcal{P}^{(j_q)}\}$. During the steering process, we represent the evolution of each $H^{(i)}$ by its coefficient vectors $a_{\mathrm{C}}^{(i)}$, $b_{\mathrm{C}}^{(i)}$, and the dynamical phase scalar $p_{i,\mathrm{dp}}$. Collectively, the evolution of the full set $\{H^{(i)}\}$ can be described using the coefficient matrices $A_{\mathrm{C}}$, $B_{\mathrm{C}}$ (whose columns are the individual $a_{\mathrm{C}}^{(i)}$ and $b_{\mathrm{C}}^{(i)}$, respectively), along with the dynamical phase vector $p_{\mathrm{dp}} = (p_{1,\mathrm{dp}}, \cdots, p_{N,\mathrm{dp}})$.

According to the discussion above, under the action of a superoperator $\mathcal{P}^{(j_q)}$ characterized by the tuple $\{j_q, \pm_q, e^q_{\mathrm{CVH}}\}$, the matrices $A_{\mathrm{C}}$, $B_{\mathrm{C}}$, and the vector $p_{\mathrm{dp}}$ evolve as follows:
\begin{widetext}
\begin{align}
    &B_{\mathrm{C}}\rightarrow B_{\mathrm{C}}+ N^{q}\times A_{\mathrm{C}}=B_{\mathrm{C}}+\left(\begin{pmatrix}
    \\
     b^{(j_{q})}_{0,\mathrm{C}}\\
     \\  
    \end{pmatrix}\begin{pmatrix}
     & e^q_{\mathrm{CVH}} & 
    \end{pmatrix}\right)\times A_{\mathrm{C}},\\
    &A_{\mathrm{C}}\rightarrow M^{q}\times A_{\mathrm{C}}=\left(\mathbbm{1}_{(N-d_r)\times(N-d_r)}+\begin{pmatrix}
    \\
     a^{(j_{q})}_{0,\mathrm{C}}\\
     \\  
    \end{pmatrix}\begin{pmatrix}
     & e^q_{\mathrm{CVH}} & 
    \end{pmatrix}\right)\times A_{\mathrm{C}},\\
    &p_{\mathrm{dp}} \rightarrow\begin{cases}
        p_{\mathrm{dp}}, & \text{if with }\mathcal{P}^{(j_{q})}_+;\\
        p_{\mathrm{dp}} + \begin{pmatrix}
     e^q_{\mathrm{CVH}}
    \end{pmatrix}\times \begin{pmatrix}
     & A_{\mathrm{C}} & 
    \end{pmatrix},& \text{if with }\mathcal{P}^{(j_{q})}_-,
    \end{cases} 
\end{align}    
\end{widetext}
where $N^{(j_{q})}$ is the $d_r\times (N-d_r)$ binary matrix defined above and $M^{(j_{q})}$ is a $(N-d_r)\times (N-d_r)$ binary matrix defined above.

Furthermore, let us consider the whole evolution of $\{H^{(i)}\}$ under the sequence of steering superoperators, $\{\mathcal{P}^{(j_1)},\mathcal{P}^{(j_2)},\cdots,\mathcal{P}^{(j_{q})}\}$. Note that in the Heisenberg picture, the Hamiltonians evolve in a reverse order in time, namely
\begin{align}
    H^{(i)}&\rightarrow \mathcal{P}^{(j_1)\dagger}(H^{(i)})\rightarrow \mathcal{P}^{(j_1)\dagger}(\mathcal{P}^{(j_2)\dagger}(H^{(i)}))\rightarrow\cdots\nonumber\\
    &\rightarrow \mathcal{P}^{(j_1)\dagger}(\mathcal{P}^{(j_2)\dagger}(\cdots \mathcal{P}^{(j_q)\dagger}(H^{(i)})\cdots)).
\end{align}
Suppose these superoperators are characterized by the sequence of tuples $S=\{(j_1,\pm_1,e^{1}_{\mathrm{CVH}}),(j_2,\pm_2,e^{2}_{\mathrm{CVH}})\cdots,(j_q,\pm_q,e^{q}_{\mathrm{CVH}})\}$, then 
the parameters $A_{\mathrm{C}}, B_{\mathrm{C}}$ and $p_{\mathrm{dp}}$ defined above will transform as follows:
\begin{widetext}
    \begin{align}
    B_{0,\mathrm{C}}&\rightarrow B_{0,\mathrm{C}}+ N^{1}\times A_{0,\mathrm{C}}\rightarrow B_{0,\mathrm{C}}+ (N^{2}+N^{1}\times M^{2})\times A_{0,\mathrm{C}}\nonumber\\
    &\rightarrow B_{\mathrm{C}}^q= B_{0,\mathrm{C}}+ (N^{q}+N^{q-1}\times M^{q}+\cdots+N^{1}\times M^{2}\times\cdots\times M^{q})\times A_{0,\mathrm{C}},\\
    A_{0,\mathrm{C}}&\rightarrow M^{1}\times A_{0,\mathrm{C}}\rightarrow M^{1}\times M^{2}\times A_{0,\mathrm{C}}\rightarrow A_{\mathrm{C}}^q= M^{1}\times\cdots\times M^{q}\times A_{0,\mathrm{C}},\\
    p_{\mathrm{dp}}^0=0&\rightarrow e^{1}_{\mathrm{CVH}}\times A_{0,\mathrm{C}}\rightarrow (e^{2}_{\mathrm{CVH}}+e^{1}_{\mathrm{CVH}}\times M^{2})\times A_{0,\mathrm{C}}\nonumber\\
    & \rightarrow p_{\mathrm{dp}}^q= (e^{q}_{\mathrm{CVH}}+e^{q-1}_{\mathrm{CVH}}\times M^{q}+\cdots+e^{1}_{\mathrm{CVH}}\times M^{2}\times\cdots\times M^{q})\times A_{0,\mathrm{C}}.    
\end{align}
\end{widetext}
Here we assume all superoperators are defined with a minus sign, while a similar expression of $p_{\mathrm{dp}}$ can be found if there are also $\mathcal{P}_+^\dagger$, by subtracting corresponding terms. For the convenience of discussion, let us further define two quantities:
\begin{align}
    f^q=&e^{q}_{\mathrm{CVH}}+e^{q-1}_{\mathrm{CVH}}\times M^{q}+\cdots+e^{1}_{\mathrm{CVH}}\nonumber\\
    &\times M^{2}\times\cdots\times M^{q},\nonumber\\
    P^q=&N^{q}+N^{q-1}\times M^{q}+\cdots+N^{1}\nonumber\\
    &\times M^{2}\times\cdots\times M^{q}.\nonumber
\end{align}
In this situation, the evolution maps above can be rewritten as follows:
\begin{align}
    B^q_{\mathrm{C}}=&B_{0,\mathrm{C}}+P^q A_{0,\mathrm{C}},\nonumber\\
    A^q_{\mathrm{C}}=&M^1 M^2\cdots M^q A_{0,\mathrm{C}},\nonumber\\
    p_{\mathrm{dp}}^q=&f^q\times A_{0,\mathrm{C}},\nonumber\\
    P^{q+1}=&N^{q+1}+P^q M^{q+1},\nonumber\\
    f^{q+1}=&e^{q+1}_{\mathrm{CVH}}+f^q M^{q+1}.\label{eq:evo_map}
\end{align}

\subsection{Proof of steerability}

In this section, we present and prove the main theorem regarding the steerability of commuting Pauli Hamiltonians. Additionally, we provide a randomized algorithm for generating a sequence of steering operators, such that these operators ensure that the energy of the system is non-increasing, hence driving the system toward the ground-state manifold.

\begin{theorem}
    GS manifold of a Commuting Pauli Hamiltonian is a steerable subspace.
    \label{thm:CPH}
\end{theorem}

To prove this theorem, we explicitly construct a steering protocol that drives the system toward the GS manifold. We exploit the superoperators introduced in the previous section and track the evolution of energy expectation values. The protocol proceeds in two stages:
(i) all initial states are steered into a subspace that is a common eigenspace of all $H^{(i)}$; (ii) steering operators are then generated using randomly chosen Pauli operators $V$, with the constraint that the energy of the system does not increase. Once the system has converged to the GS manifold, these superoperators further ensure ergodic transitions within the ground-state subspace.

\begin{proof}

For the first stage of steering, let us consider randomly generating a sequence of superoperators defined by $S=\{(j_1,\pm_1,e^{1}_{\mathrm{CVH}}),(j_2,\pm_2,e^{2}_{\mathrm{CVH}})\cdots,(j_q,\pm_q,e^{q}_{\mathrm{CVH}})\}$, until all $H^{(i)}$ evolve into $\pm \mathbbm{1}$. Equivalently,  all $N-d_r$ non-trivial generators of $\mathbb{H}$ (cf. Eq.~\eqref{eq:Hgroup})become $\pm \mathbbm{1}$, implying that the matrix $A_{\mathrm{C}}$ evolve into a zero matrix.  Note that the row vectors of $A_{0,\mathrm{C}}$, namely $a_{0,\mathrm{C}}^{[i]}$, are linearly independent, thus $A_{\mathrm{C}}=0$ if and only if $M^{1}\times\cdots\times M^{q}=0$. This is always possible since $M^{i}$ is a rank-$(N-d_r-1)$ matrices, so there always exists a sequence $S$ with length $q\ge N-d_r$, which makes the product to be a zero matrix. For instance, we can take $j_1=d_r+1,j_2=d_r+1,\cdots, j_{N-d_r}=N$ and proper $e^i_{\mathrm{CVH}}$.

Practically, this process can also be done in the following iterative way: firstly for $H^{(1)}$ we can find some superoperator such that $H^{(1)}\rightarrow \pm (H^{(1)})^2=\pm 1$, and then if during this process $H^{(2)}\rightarrow \pm H^{(1)} H^{(2)}\neq \pm 1$, then we can use another two superoperator such that $H^{(1)} H^{(2)}\rightarrow \rightarrow \pm H^{(1)} (H^{(2)})^2=\pm H^{(1)} \rightarrow\pm (H^{(1)})^2=\pm 1$ and repeat this process until all local terms are transformed into $\pm 1$. 

Then for the second stage, we will prove that under our steering protocol, the Hamiltonian is transformed to an identity operator multiplied by an energy constant, which will result in GS energy.

Before proceeding, we need to calculate the corresponding energy constant. After converging in the first stage, each $H^{(i)}$ is transformed into some linear combination of relations, encoded in $b^{i}_{\mathrm{C}}$, $i$th column $B_{\mathrm{C}}$. Thus obviously it is $H^{(i)}\rightarrow (-1)^{\boldsymbol{b}^{i}_{\mathrm{C}}\cdot \boldsymbol{p}_{\mathrm{ph}}+p_{i,\mathrm{dp}}}\mathbbm{1}$. In other words, the final states go into the common eigenspace of all local terms $H^{(i)}$, with the energy vector
\begin{equation}
v_{\mathrm{E}}=p_{\mathrm{ph}}\times B_{\mathrm{C}}+ p_{\mathrm{dp}}.
\end{equation}
Obviously it satisfy Eq.~\eqref{EnergyEquation}, by noting that $C_{\mathrm{H}}\times e_{\mathrm{CVH}}^T=0 $ and $C_{\mathrm{H}}\times B_{\mathrm{C}}^T= C_{\mathrm{H}}\times B_{0,\mathrm{C}}^T=\mathbbm{1}_{d_r\times d_r}$. 

Now let us study how does the energy constant evolve under the steering superoperators. Now we consider a subsequent superoperator defined by the parameters $(j_{q+1},\pm_{q+1},e^{q+1}_{\mathrm{CVH}})$. According to Eq.~\eqref{eq:evo_map}, the parameters $A_{\mathrm{C}},B_{\mathrm{C}},p_{\mathrm{dp}}$ evolve as follows:
\begin{align}
    A^q_{\mathrm{C}}\rightarrow & A^{q+1}_{0,\mathrm{C}}= 0\times M^{q+1}\times A_{0,\mathrm{C}}=0,\\
    p^q_{\mathrm{dp}}\rightarrow & p^{q+1}_{\mathrm{dp}}=p_{\mathrm{dp}}+(1+f^q\cdot a_{0,\mathrm{C}}^{(j_{q+1})})e^{q+1}_{\mathrm{CVH}}\times A_{0,\mathrm{C}},\\
    B^q_{\mathrm{C}}\rightarrow&B^{q+1}_{\mathrm{C}}= B^q_{\mathrm{C}}+(b^{(j_{q+1})}_{0,\mathrm{C}}+P^q\times a_{0,\mathrm{C}}^{(j_{q+1})})\nonumber\\
    &\otimes e_{\mathrm{CVH}}^{(0)}\times A_{0,\mathrm{C}}.
\end{align}
Therefore, the energy vector $v_{\mathrm{E}}$ transforms as follows:
\begin{widetext}
    \begin{equation}
    \label{eq:energy change}
    v^{q+1}_{\mathrm{E}}=\begin{cases}
         v^q_{\mathrm{E}}+ [p_{\mathrm{ph}}\times(b^{(j_{q+1})}_{0,\mathrm{C}}+P^q\cdot a_{0,\mathrm{C}}^{(j_{q+1})})]e_{\mathrm{CVH}}^{q+1}\times A_{0,\mathrm{C}},& \text{if with }\mathcal{P}^{(j_{q+1})}_+;\\
        v^q_{\mathrm{E}}+ [p_{\mathrm{ph}}\times(b^{(j_{q+1})}_{0,\mathrm{C}}+NM\cdot a_{0,\mathrm{C}}^{(j_{q+1})})+(1+f^q\cdot a_{0,\mathrm{C}}^{(j_{q+1})})]e_{\mathrm{CVH}}^{q+1}\times A_{0,\mathrm{C}},& \text{if with }\mathcal{P}^{(j_{q+1})}_-.
    \end{cases} 
\end{equation}
\end{widetext}
In other words, the energy vector might change by a vector in the right null subspace of $C_{\mathrm{H}}$, i.e. $e^{q+1}_{\mathrm{CVH}} \times A_{0,\mathrm{C}}$, or be invariant, controlled by the coefficients in the brackets. 

Before proceeding, let us further prove two facts for this evolution map: (i) when $p_{\mathrm{ph}}\neq 0$, $p_{\mathrm{ph}}\times(B_{0,\mathrm{C}}+P^q\cdot A_{0,\mathrm{C}})\neq f^q\times A_{0,\mathrm{C}}$. Note that $p_{\mathrm{ph}}\times(B_{0,\mathrm{C}}+P^q\cdot A_{0,\mathrm{C}})\times C^T_{\mathrm{H}}=p_{\mathrm{ph}}\times(\mathbbm{1}_{d_r\times d_r}+0)=p_{\mathrm{ph}}$, and $f\times A_{0,\mathrm{C}}\times C^T_{H}=0$. Therefore, assuming $p_{\mathrm{ph}}\neq 0$, $p_{\mathrm{ph}}\times(B_{0,\mathrm{C}}+P^q\cdot A_{0,\mathrm{C}})= f^q\times A_{0,\mathrm{C}}$, then $p_{\mathrm{ph}}=0$, contradicting with the assumption, finishing the proof. (ii) when $p_{\mathrm{ph}}\neq 0$, $p_{\mathrm{ph}}\times(B_{0,\mathrm{C}}+P^q\cdot A_{0,\mathrm{C}})\neq 0$. Otherwise $p_{\mathrm{ph}}=p_{\mathrm{ph}}B_{0,\mathrm{C}}C_{\mathrm{H}}^T=p_{\mathrm{ph}}\times(B_{0,\mathrm{C}}+P^q\cdot A_{0,\mathrm{C}})\times C^T_{\mathrm{H}}=0\times C^T_{\mathrm{H}}=0$, contradicting with the assumption that $p_{\mathrm{ph}}\neq 0$.

Based on this evolution map, in the following we will show that it is always possible to choose a subsequent proper superoperator $\mathcal{P}$ tied with the parameters $(j_{q+1},\pm_{q+1},e^{q+1}_{\mathrm{CVH}})$, such that the energy of the system does not increase. This can be discussed for the two cases, whether $p_{\mathrm{ph}}$ is a zero vector or not.

Let us first discuss the case $p_{\mathrm{ph}}\neq 0$. If the energy would increase after $e^{q+1}_{\mathrm{CVH}}A_{0,\mathrm{C}}$ is added to $v^q_{\mathrm{E}}$, our goal is then to choose $j_{q+1},\pm_{q+1}$ properly, such that the coefficients before $e^{q+1}_{\mathrm{CVH}}$ in Eq.~\eqref{eq:energy change} is $0$, hence $v^q_{\mathrm{E}}$ keeps invariant. First note that since we have proved $p_{\mathrm{ph}}\times(B_{0,\mathrm{C}}+P^{q}\cdot A_{0,\mathrm{C}})\neq f^{q}\times A_{0,\mathrm{C}}$, it is impossible that $p_{\mathrm{ph}}\times(b^{(j)}_{0,\mathrm{C}}+P^{q}\cdot a_{0,\mathrm{C}}^{(j)})=f^q\cdot a^{(j)}_{0,\mathrm{C}}=1$ is established for every $j$. Let us denote $x=p_{\mathrm{ph}}\times(b^{(j_{q+1})}_{0,\mathrm{C}}+P^q\cdot a_{0,\mathrm{C}}^{(j_{q+1})})$ and $y=f^q\cdot a^{(j_{q+1})}_{0,\mathrm{C}}$, then there exists at least one $j_{q+1}$ such that $(x,y)\neq(1,1)$. For these remaining possibilities and such a $j_{q+1}$, we can choose the sign $\pm_{q+1}$
\[
\pm_{q+1}=\begin{cases}
    -_{q+1}, & \text{if } x=1,y=0,\\
    \pm_{q+1}, & \text{if } x=0,y=1,\\
    +_{q+1}, & \text{if } x=0,y=0.
\end{cases}
\]
It can be verified that $v_{\mathrm{E}}$ is then invariant. 

On the other hand, if the energy would decrease or keep the same after $e^{q+1}_{\mathrm{CVH}}A_{0,\mathrm{C}}$ is added to $v^q_{\mathrm{E}}$, our goal is then to make the coefficient before $e^{q+1}_{\mathrm{CVH}}$ in Eq.~\eqref{eq:energy change} to be $0$. Note that we have proved that the components of $p_{\mathrm{ph}}\times(B_{0,\mathrm{C}}+P^qA_{0,\mathrm{C}})$ cannot all be zero, so there exist at least one $j_{q+1}$ such that $x=1$. Then we can choose the sign $\pm_{q+1}$ as follows:
\[
\pm_{q+1}=\begin{cases}
    \pm_{q+1}, & \text{if } x=1,y=1,\\
    +_{q+1}, & \text{if } x=1,y=0,\\
    -_{q+1}, & \text{if } x=0,y=0,
\end{cases}
\]
and it can be directly verified that $v^q_{\mathrm{E}}\rightarrow v^q_{\mathrm{E}}+e^{q+1}_{\mathrm{CVH}}\times A_{0,\mathrm{C}}$.

For $p_{\mathrm{ph}}=0$, now $v_{\mathrm{E}}=f\times A_{0,\mathrm{C}}$. So we can use $\mathcal{P}_+^\dagger$ if energy would increase and use $\mathcal{P}_-^\dagger$ if energy would decrease or keep the same. Note that the latter case fails when $v_{\mathrm{E}}=f\times A_{0,\mathrm{C}}=(1,1,\cdots,1)$. This means all $H^{(i)}$ can simultaneously take eigenvalue $-1$, thus it is Frustration-Free and does not have degeneracy in this binary representation (although $H$ might have some degeneracy due to symmetry), so it is natural that it terminates the energy decreasing.

We now describe a randomized procedure that generates an infinite sequence $S=\{(j_1,e^{(1)}_{\mathrm{CVH}},\pm_1), (j_2,e^{(2)}_{\mathrm{CVH}},\pm_2),\cdots\}$ which guarantees the convergence to GS manifold. (i) The first stage is to erase any information about the initial state by stabilizing all initial states into a subspace with definite local energies. This is achieved by evolving the system until the matrix $A_{\mathrm{C}} = 0$. (ii) The second stage is to further cool the system. At each step, this is done in a greedy fashion: we first randomly generate a binary vector $e^{(q+1)}_{\mathrm{CVH}}$, and then determine $j_{q+1}$ and $\pm_{q+1}$ based on the discussion above. The action of the superoperator $\mathcal{P}^{(j_{q+1})}$ is to update the energy configuration via $v_{\mathrm{E}}\rightarrow v_{\mathrm{E}} + e_{\mathrm{CVH}}\times A_{0,\mathrm{C}}$ if the resulting binary vector has the same or a greater number of $1$s—i.e., if the energy does not increase. If the number of $1$s decreases, the update is rejected. As a result, the GS subspace becomes the only possible fixed point of the evolution, ensuring convergence of the protocol.
\end{proof}

In conclusion, we have shown that for a commuting Pauli Hamiltonian $H$, one can randomly generate binary vectors $e_{\mathrm{CVH}}$ to construct superoperators that either lower the system’s energy or drive it into another eigensubspace of $H$ with the same energy. As a result, in the long-time limit, this sequence of superoperators converges to the GS manifold and enables transitions within the GS subspace. In general, the convergence time may be exponentially long, as the underlying problem is classically NP-complete.

\subsection{Locality of superoperator}\label{sec:Locality of superoperator}

Another important aspect of steerability is the locality of the superoperators introduced above. To demonstrate this locality, we begin by proving Lemma~\ref{lemma:V}, which is stated in Appendix~\ref{app:Steering Operators}.

\addtocounter{theorem}{-2}
\begin{lemma}
    For arbitrary Pauli operator $V$, with its commuting relation with $H^{(i)}$ represented by vector $g_{\mathrm{VH}}$, must satisfy that $C_{\mathrm{H}} \times g_{\mathrm{VH}}^T=0$, where the matrix $C_{\mathrm{H}}$ is defined in Appendix~\ref{app:Spectrum of Commuting Pauli Hamiltonian}. Conversely for arbitrary $g_{\mathrm{VH}}$ in right null space of $C_{\mathrm{H}}$, we can construct a Pauli operator $V$.
\end{lemma}
\addtocounter{theorem}{1}

\begin{proof}
    As introduced in Appendix~\ref{app:Steering Operators}, the commutation relation of $V$ with the elements from $\mathbb{H}$ (cf. Eq.~\eqref{eq:Hgroup}) is determined by the component of the binary vector $g_{\mathrm{VH}}$ as follows:
    $$
    VH^{(i)}V=(-1)^{g_{i,\mathrm{VH}}}H^{(i)}.
    $$
    However, these $N$ commuting relations are not independent, since $\{H^{(i)}\}$ are not independent generators of $\mathbb{H}$. In other words, $V$ must commute with those product of $\{H^{(i)}\}$ that constitutes a relation of group $\mathbb{H}$, since this product is proportional to the identity operator. 

    In general, let us consider an element $G\in \mathbb{H}/\{1,-1\}$, that can be written as $G=\prod_i (H^{(i)})^{c_{i,\mathrm{H}}}$, with the binary vector $c_{\mathrm{H}}$ representing powers of $\{H^{(i)}\}$. It can then be verified that the commutation relation is as follows:
    $$
    VGV=(-1)^{\boldsymbol{c}_{\mathrm{H}}\cdot\boldsymbol{g}_{\mathrm{VH}}}G.
    $$
Therefore, when $G$ is proportional to the identity operator, or equivalently $c_{\mathrm{H}}$ lies in the row space of the matrix $C_{\mathrm{H}}$, then $\boldsymbol{c}_{\mathrm{H}}\cdot\boldsymbol{g}_{\mathrm{VH}}=0$. Since this commutation relation hold for all binary vectors in the row space of $C_{\mathrm{H}}$, we conclude that $C_{\mathrm{H}}\times g_{\mathrm{VH}}^T=0$.

On the other hand, let us prove that for any binary vector $g_{\mathrm{VH}}$, such that $C_{\mathrm{H}}\times g_{\mathrm{VH}}^T=0$, there always exists a Pauli operator that the corresponding commutation relation holds. Note that if we consider two Pauli operator $V,V'$ defined with the parameters $g_{\mathrm{VH}},g_{\mathrm{VH}}'$, respectively, then $VV' G V'V = (-1)^{\boldsymbol{c}_{\mathrm{H}}\cdot(\boldsymbol{g}_{\mathrm{VH}}+\boldsymbol{g}'_{\mathrm{VH}})}G$. Therefore, it suffices to construct it for bases of the subspace spanned by $g_{\mathrm{VH}}$, i.e. consider $V^{[i]}$ defined with $\{a^{[i]}_{0,\mathrm{C}}\}$. Note that $D_{\mathrm{H}}\times A_{0,\mathrm{C}}^T=\mathbbm{1}_{(N-d_r)\times(N-d_r)}$, thus $V^{(i)}$ anti-commutes with $H^{(d_r+i)}$, $i$th independent generators of $\mathbb{H}$, while commuting with other independent generators.
    
In the following, let us solve for the binary representation of the Pauli operator $V^{(i)}$. Note that for two Pauli operator $\sigma_x^m\sigma_z^n$ and $\sigma_x^p\sigma_z^q$, if we permute them we will obtain 
    $$
    \sigma_x^m\sigma_z^n\times \sigma_x^p\sigma_z^q=(-1)^{mq+np}\sigma_x^p\sigma_z^q\times\sigma_x^m\sigma_z^n
    $$
thus induces a symplectic structure (note that $-1\equiv 1\mod{2}$). More generally, the phase obtained by permuting two Pauli operator $G$, $G'$ represented with $g_{P}$ and $g'_{P}$ can be determined as
    \begin{align}
    G\times G'=&(-1)^{g_{P}\times J\times g'^{T}_{\mathrm{P}}}G'\times G,\nonumber\\ 
   \text{where } J=&\begin{pmatrix}
    0 & \mathbbm{1}_{n\times n}\\
    \mathbbm{1}_{n\times n} & 0
    \end{pmatrix}.
    \end{align}
Now suppose the binary vectors encoding powers of $H^{(d_r+i)}$ in terms of $\sigma_x$ and $\sigma_z$ is denoted as $u^{(i)}_{\mathrm{P}}$, and $v^{(i)}_{\mathrm{P}}$ encodes $V^{(i)}$, then the commutation relation
    $$
    V^{(i)}H^{(d_r+j)}V^{(i)}=(-1)^{\delta_{ij}} H^{(d_r+j)},
    $$
where $i,j \in\{1,\cdots, N-d_r\}$, is translated as follows:
    $$
    u^{(j)}_{\mathrm{P}}\times J \times v^{(i)}_{\mathrm{P}} =\delta_{ij}, 
    $$
which is a linear system of equations. Since $u^{(i)}_{\mathrm{P}}$ are all independent, $v^{(i)}_{\mathrm{P}}$ all have many solutions, determined up to right null space of $U\times J$ where $i$th row of $U$ is $u^{(i)}_{\mathrm{P}}$. The Pauli operators $\{V^{(i)}\}$ is then solved given the powers of $\sigma_x$ and $\sigma_z$, i.e. $v^{(i)}_{\mathrm{P}}$. For a general Pauli operator $V$ defined with the binary vector $g_{\mathrm{VH}}$, it can be similarly solved.
\end{proof}

In the following, let us discuss the locality of the superoperator. Note that support of $\mathcal{P}^{(j)}$ is determined by $\operatorname{supp}(H^{(j)})\cup\operatorname{supp}(V)$. Above we did not make any assumption about the locality of the Hamiltonian, so in principle with global commuting Pauli Hamiltonian, we might need global superoperators. So to further discuss locality, here we consider translationally invariant $k-$local commuting Pauli Hamiltonian. And we assume there is only one qubit on each lattice point for simplicity. The discussion here can be simplified by introducing a $\mathbb{F}_2$-field representation~\cite{haah2013commuting}, where $\mathbb{F}_2=\{0,1\}$. Suppose we are given a spatial lattice with translation symmetry group $\Lambda$, for instance in a $2$-dimensional infinite lattice it is $\mathbb{Z}\otimes\mathbb{Z}$. Then we can introduce a  $R-$module, where $R=\mathbb{F}_2(\Lambda)$, as the representation of the Pauli operators.  For an element $a\in R$, we can expand it as $\sum_{g\in\Lambda}a_g g$, where $a_g$ are coefficients. We define trace of $a\in R$ as $\operatorname{tr}(a)=a_1$. Then we can represent a Pauli operator $O$ as follows:
$$
O= \left(\bigotimes_{g\in\Lambda} (\sigma_x^{(g)})^{x_g}\right)\left(\bigotimes_{g\in\Lambda} (\sigma_z^{(g)})^{z_g}\right),
$$
where $(x,z)\in R^2$. For two Pauli operators $O_1,O_2$ represented by $(x_1,z_1),(x_2,z_2)\in R^2$, $O_1$ commute with $O_2$ if and only if 
$$
\operatorname{tr}\left(\begin{pmatrix}
 \bar x_1& \bar z_1
\end{pmatrix}\begin{pmatrix}
 0& 1\\
 1& 0
\end{pmatrix}\begin{pmatrix}
 x_2\\
 z_2
\end{pmatrix}
\right)=0,
$$
where $\bar x= \sum_{g\in\Lambda}x_{g^{-1}}g$. 

Based on this representation, a commuting Pauli Hamiltonian can be written as $H=\sum_{g\in\Lambda}h_{1,g}+\cdots+h_{t,g}$ where $t$ is the number of types of interaction. For instance, in one dimension $\Lambda=\mathbb{Z}$, then the group algebra is $R=\mathbb{F}_2[x,\bar x]$, namely a polynomial ring over $\mathbb{F}_2$. Consider the Ising model $H=-\sum_{i\in\mathbb{Z}}\sigma^z_i\otimes \sigma^z_{i+1}$, then the Pauli module is $R^2$ and stabilizer module $S$ is generated by $(0,1+x)^T$, since $\sigma^z_i\otimes \sigma^z_{i+1}$ is represented by $(0,x^i+x^{i+1})$. Then in general, the Hamiltonian is represented by $2$ by $t$ matrix over $R$. Every element of the Hamiltonian group is then generated by this module, namely
$$
S\times a=\begin{pmatrix}
 x_1 & \cdots &x_t\\
 z_1 & \cdots &z_t
\end{pmatrix}\begin{pmatrix}
 a_{1}\\
 \vdots\\
 a_{t}
\end{pmatrix},
$$
where $a\in R^t$. Then the relations of the Hamiltonian group $\mathbb{H}$ here is just the right kernel of $S$. Suppose the generators of constraints are $c_{i,\mathrm{H}}\in R^t$, then the energy vectors can be solved as
$$
\operatorname{tr}\begin{pmatrix}
c_{1,\mathrm{H}}\\
c_{2,\mathrm{H}}\\
\vdots
\end{pmatrix}\begin{pmatrix}
 v_{1,\mathrm{E}}\\  v_{2,\mathrm{E}} \\\cdots \\v_{t,\mathrm{E}}
\end{pmatrix}= p_{\mathrm{ph}}^T,
$$
where $v_{i,\mathrm{E}}=\sum_{g\in\Lambda} v_{g,i,\mathrm{E}}g$ denotes the eigenvalue of $h_{i,g}=v_{g,i,\mathrm{E}}$, and $p_{\mathrm{ph}}$ are binary numbers as defined above. Then we can also represented $V$ by a $g_{\mathrm{VH}}\in R^t$, where
\begin{equation}
V h_{i,g} V= (-1)^{g_{g,i,\mathrm{VH}}} h_{i,g}.
\end{equation}

Therefore, the locality of $V$ and $H^{(i)}$ can be established when the entries of $g_{\mathrm{VH}}$ and the representation of $H^{(i)}$ are restricted to operators of bounded support. For example, in two dimensions with $R = \mathbb{F}_2[x, y, \bar{x}, \bar{y}]$, a term like $g_{1,\mathrm{VH}} = x^2 + x^{50} + xy^{-50}$ implies that $V$ anti-commutes with $h_{1,x^2}$, $h_{1,x^{50}}$, and $h_{1,xy^{-50}}$, which is hardly non-local. Thus, if a randomly generated $g_{\mathrm{VH}}$ is found to be non-local, it can be decomposed into a sum of local components, each corresponding to a local operator $V_k$. By slightly modifying Def.~\ref{def:stabilizability}—replacing a single local superoperator with a product of local superoperators, each realized by a local Pauli operator $V_k$—we conclude that the GS manifold of any commuting Pauli Hamiltonian is a steerable subspace.

\subsection{Alternative superoperators}\label{app:Alternative superoperators}

In addition to the family of steering operators discussed above, we now consider an alternative class of steering operators.

Above we have chosen the the superoperators such that the system is always stabilized in subspace with definite local energy $\langle H^{(i)}\rangle=\pm 1$, which works just like flipping spins classically. In general, steering superoperators can exhibit more complex behaviors. For instance consider the anti-ferromagnetic 1-dimensional Ising model $H=\sum_i H^{(i)}=\sum_i \sigma^x_i\otimes \sigma^x_{i+1}$. This Hamiltonian possesses global $\mathbb{Z}_2$ symmetry, namely $[H,Z]=0$, where $Z=\bigotimes_i\sigma^z_i$. We can choose superoperator as follows:
\begin{align*}
    &\mathcal{P}^{(i)}(\cdot)\\
    =&\frac{\mathbbm{1}-H^{(i)}}{2}\cdot\frac{\mathbbm{1}-H^{(i)}}{2}\\
    &+ \frac{\mathbbm{1}-H^{(i)}}{2} V \frac{\mathbbm{1}-P^{(i)}}{2}\cdot\frac{\mathbbm{1}-P^{(i)}}{2} V \frac{\mathbbm{1}-H^{(i)}}{2}\\
    &+ \frac{\mathbbm{1}-H^{(i)}}{2} V \frac{\mathbbm{1}+P^{(i)}}{2}\cdot\frac{\mathbbm{1}+P^{(i)}}{2} V \frac{\mathbbm{1}-H^{(i)}}{2},
\end{align*}
where $V=\mathbbm{1}_i\otimes \sigma^z_{i+1}$ and $P^{(i)}=\sigma^z_i\otimes \sigma^z_{i+1}$. Note that $\mathcal{P}^{(i)\dagger}(Z)=Z$ is invariant, and $\mathcal{P}^{(i)\dagger}(H^{(i)})=-1$, establishing a local cooling of $H^{(i)}$ while respecting the $\mathbb{Z}_2$ symmetry. On the other hand, $\mathcal{P}^{(i)\dagger}(H^{(i+1)})=\mathcal{P}^{(i)\dagger}(\mathbbm{1}_i\otimes \sigma^x_{i+1})\otimes \sigma^x_{i+2}=(\mathbbm{1}_i\otimes \sigma^x_{i+1}-\sigma^x_{i}\otimes\mathbbm{1}_{i+1})/2\otimes \sigma^x_{i+2}$ is not an element of the Hamiltonian group. 

As a result, when using this alternative family of superoperators, the total energy of the system decreases exponentially with the number of applications, converging toward the GS manifold. Specifically, the energy expectation value evolves as $\langle H\rangle \sim E_{\mathrm{GS}}+E \exp -O(N)$, where $E$ is a constant and $N$ denotes the number of applied superoperators. Notably, the resulting asymptotic state does not, in general, possess definite local energy with respect to each $H^{(i)}$, and therefore lies beyond the scope of a simple classical description.

\subsection{Realization of superoperator with Clifford unitary}

We now explain how the superoperators $\mathcal{P}^{(i)}_{\pm}$, defined in Eq.~\eqref{eq:sup_op_CPH}, can be physically realized through the following procedure. First, we introduce an auxiliary qubit (serving as a "detector"), and then apply a Clifford unitary operation~\cite{FTQCPhysRevA.57.127}—which maps Pauli operators to Pauli operators—on a segment of the system \textit{together with} the auxiliary qubit. Finally, we trace out the auxiliary qubit, thereby obtaining a superoperator acting solely on the corresponding segment of the system.

Consider the unitary operator
\begin{align*}
    U=&|0\rangle\langle 0|\otimes \frac{1-H^{(i)}}{2} + |0\rangle\langle 1|\otimes \frac{1+H^{(i)}}{2}\\
    &+|1\rangle\langle 0|\otimes \frac{1-H^{(i)}}{2}V+|1\rangle\langle 1|\otimes \frac{1+H^{(i)}}{2}V,
\end{align*}
where $\{V,H^{(i)}\}=0$, and the first qubit is the auxiliary qubit. It can be verified that $UU^\dagger=U^\dagger U=\mathbbm{1}$ is a unitary operator. Then note that if we initialize the auxiliary qubit as $|0\rangle$, apply $U$ and then trace out auxiliary qubit, we then realize $\mathcal{P}^{(i)}_-(\cdot)$, since the Kraus operators are $\langle 0|U|0\rangle=(1-H^{(i)})/2$ and $\langle 1|U|0\rangle=(1-H^{(i)})/2\times V$. Direct calculation of $U^\dagger (\sigma_{x,z}\otimes P)U$ prove it is a Clifford unitary operator, where $P$ is a Pauli operator. For instance, $U^\dagger (\sigma_x\otimes P)U=\sigma_z\otimes P$ when $\{V,P\}=\{H,P\}=0$. Similar unitary operator can also be constructed for other cases of the superoperator.

\subsection{Stability against heating}\label{app:Stability against heating}
 
In this section, we explain how environmental heating can be suppressed with the superoperators. This suppression is evident from the fact that the system does not leak out of the GS manifold, as shown previously. Here, we focus on illustrating the dynamical process by which this elimination of environmental errors (i.e., heating) occurs.

Let us first consider the situation that the system is subjected to a Kraus operator $K$, originating from environmental decoherence. Without loss of generality, let us assume $K$ is a local Pauli operator. Then we focus on the situation that $K$ does not commute with all of $\{H^{(i)}\}$, or equivalently, $K$ anti-commutes with a subset of $\{H^{(i)}\}$, denoted as $\mathbb{H}_K$, that is spatially local. In this situation, since $K H^{(i)} K=-H^{(i)}$ for some $i$, this $K$ error can induce the energy flip of $H^{(i)}$, potentially cause heating of the system. In this situation, we can also design a local Pauli operator $V$, such that $V$ anti-commutes with the element from $\mathbb{H}_K$. Then for a $H^{(i)}\in\mathbb{H}_{K}$, we define a superoperator $\mathcal{P}$ defined with the parameters $(i,\pm,V)$. When the system is in the eigensubspace of $H^{(i)}$ with eigenvalue $+1$, we use the sign $+$ to define the superoperator, and we use $-$ otherwise. We apply this superoperator after the error occurs. Then in the the Heisenberg picture, an element $H^{(j)}\in\mathbb{H}_{K}$ evolve as follows:
\begin{align}
    H^{(j)}\rightarrow& K\mathcal{P}^\dagger(H^{(j)})K\nonumber\\
    = & \pm K H^{(i)}H^{(j)} K\nonumber\\
    =& \pm H^{(i)}H^{(j)}.
\end{align}
Therefore, when the system is driven into eigensubspaces of all $\{H^{(i)}\}$, the expectation value $\langle H^{(j)}\rangle$ is invariant,  as a manifestation that the error $K$ is corrected.

On the other hand, we may also investigate the decoherence effect by examining the evolution of expectation values of other operators. Under certain errors, expectation values of Pauli operators in GS manifold, other than $\{H^{(i)}\}$, might be altered, so here we study the expectation value of a Pauli operator $O$. It evolves as follows: 
\begin{widetext}
\[
\mathcal{P}^{(i)\dagger}_+(O), \mathcal{P}^{(i)\dagger}_-(O)=\begin{cases}
    O(\mathbbm{1}\pm H^{(i)}), O, & \text{if } [O,H^{(i)}]=[O,V^{(i)}]=0;\\
    0, \pm O H^{(i)}, & \text{if } [O,H^{(i)}]=\{O,V^{(i)}\}=0;\\
    0,0, & \text{if } \{O,H^{(i)}\}=0,
\end{cases}
\]
\end{widetext}
where for the former two cases, the two results are obtained by assuming $[V^{(i)},H^{(i)}]=0$ and $\{V^{(i)},H^{(i)}\}=0$, respectively. Therefore, we observe that $\langle O\rangle =0$ when (i) $[O,H^{(i)}]=[V^{(i)},H^{(i)}]=\{O,V^{(i)}\}=0$ or  (ii)$\{O, H_i\}=0$. Therefore, the errors which change values of such $O$ must be suppressed. 

\section{Important facts of the Non-Steerable case}\label{Discussions about non-steerable case}

In this section, we first study the non-steerability of target state(s) with a full Schmidt rank. We then discuss the calculation of $1 - p(\Pi_{\mathrm{GS}})$, introduced in Eq.~\eqref{eq: upper bound degenerate} in Sec.~\ref{sec:A bound on the distance between the surrogate state and the target state employing local measure}. Finally, we present the numerical estimation of the effective temperature lower bound for the Fermi–Hubbard model and the SYK model, as introduced in Sec.~\ref{sec:Implementations of NFFNS model}.

\subsection{Non-steerability for states with full-Schmidt-rank}\label{sec:Non-steerability for states with full-rank reduced states} 

As discussed in Sec.~\ref{sec:Definition of stabilizability}, a prerequisite for steerability is that the target GS can either be stabilized by certain non-trivial steering operations or be transformed from other intermediate states. In this section, we demonstrate that for GSs with full Schmidt rank, such steering operations do not exist.

We first prove a corollary stating that states with full Schmidt rank are fragile under any non-trivial operation. In particular, such states violate Condition~\ref{cond:2}. In general, we consider partitioning the whole system into two parts, $S$ and $\bar S$, such that $S$ is a local region coupled to the local steering operator $\mathcal{P}_S$, and $\bar S$ is the rest of the system. The Schmidt decomposition will be performed over this bipartition.
\begin{corollary}
\label{corollary:full rank}
    Given a (mixed) state $\rho\in\mathcal{L}(\mathcal{H}_{S})\otimes\mathcal{L}(\mathcal{H}_{\bar S})$ ($S$ is smaller than $\bar S$), if the Schmidt rank of $|\rho\rangle\rangle$ calculated for this bipartition is full, then the local superoperator $\mathcal{P}_S$ defined on $\mathcal{H}_S$ satisfying $\mathcal{P}_S(\rho)=\rho$ has only one trivial solution: $\mathcal{P}_S(\cdot)=\cdot$, namely an identity map.
\end{corollary}

\begin{proof}
The proof of it is straightforward by calculating the constraints on $\mathcal{P}_S$ explicitly. Let us write down the Schmidt decomposition as 
    $$
    |\rho\rangle\rangle=\sum_i c_i |\alpha_i\rangle\rangle_S\otimes|\beta_i\rangle\rangle_{\bar S}. 
    $$ 
Since $|\alpha_i\rangle\rangle_S$ forms a complete basis for $\mathcal{L}(\mathcal{H}_S)$, the equation $\mathcal{P}_S|\rho\rangle\rangle_S=|\rho\rangle\rangle_S$ implies
\begin{align*}
    &{}_{\bar S}\langle\langle \beta_i|\mathcal{P}_S|\rho\rangle\rangle=c_i |\alpha_i\rangle\rangle_S  \nonumber\\
    \Rightarrow &\mathcal{P}_S|\alpha_i\rangle\rangle_S= |\alpha_i\rangle\rangle_S\nonumber\\
    \Rightarrow &{}_S\langle\langle\alpha_i|\mathcal{P}_S|\alpha_j\rangle\rangle_S =\delta_{ij}.
\end{align*}
In other words, $\mathcal{P}_S$ is an identity map.
\end{proof}

We now investigate non-steerability based on the corollary above. Let us first consider the case of a non-degenerate target GS $|\psi\rangle$. Suppose that for every bipartition where subsystem $S$ is accessible by some steering operations, the state $|\psi\rangle\langle\psi|$ always has full Schmidt rank. Then, no local superoperator can leave $|\psi\rangle$ invariant, implying that $|\psi\rangle$ is non-steerable. Moreover, such a $|\psi\rangle$ has a RDM with full rank and supports no non-trivial SCQ, thus violating Condition~\ref{cond:2}. An example of such a case is the ground state of the half-filled SYK model, as discussed in Sec.~\ref{sec: SYK model}.

Next, consider a degenerate GS manifold $\mathcal{H}_{\mathrm{GS}}$, where no trivial SCQ exists in any subspace of $\mathcal{H}_{\mathrm{GS}}$, cf. Condition~\ref{cond:2} in Sec.~\ref{sec:Necessary conditions for NFFJS Hamiltonian}. Then, for any pure state $|\psi\rangle \in \mathcal{H}_{\mathrm{GS}}$, the Schmidt rank of $|\psi\rangle\langle\psi|$ is full, which means any non-trivial operation would necessarily change $|\psi\rangle$. However, the Schmidt ranks of mixed states in $\mathcal{H}_{\mathrm{GS}}$ are not necessarily full,  see the example below. In such cases, for a mixed state $\rho$, it may be possible to construct a non-trivial superoperator $\mathcal{P}_S$ that keeps $\rho$ invariant, i.e., $\mathcal{P}_S(\rho) = \rho$. But the degrees of freedom in $\mathcal{P}_S$ are limited to eliminating off-diagonal elements of the density matrix without enabling population transfer. Such superoperators do not help stabilize the system, as populations in excited states cannot be transferred into the GS manifold.
Therefore, in this situation, there does not exist any non-trivial superoperator that both preserves at least one state in $\mathcal{H}_{\mathrm{GS}}$ and helps stabilizing the subspace. Moreover, if Condition~\ref{cond:3} is also violated, then even superoperators that merely eliminate off-diagonal elements cannot prevent leakage from $\mathcal{H}_{\mathrm{GS}}$, thereby establishing non-steerability.

As an example that mixed state in $\mathcal{H}_{\mathrm{GS}}$ is not necessarily Schmidt-rank-full, let us take $\mathcal{H}_S=\operatorname{span}\{|\alpha_i\rangle_S\}_{i=0,1,2,3}$, and $\mathcal{H}_{\mathrm{GS}}=\operatorname{span}\{|\psi_0\rangle,|\psi_1\rangle\}$ is a two-dimensional GS manifold. The Schmidt decomposition  of two GSs is given by
\begin{align*}
    |\psi_0\rangle=&1/2(|\alpha_0\rangle_S\otimes|\beta_0\rangle_{\bar S}+|\alpha_1\rangle_S\otimes|\beta_1\rangle_{\bar S}\\
    &+|\alpha_2\rangle_S\otimes|\beta_2\rangle_{\bar S}+|\alpha_3\rangle_S\otimes|\beta_3\rangle_{\bar S}),\\
    |\psi_1\rangle=&1/2(e^{i \theta}|\alpha_0\rangle_S\otimes|\beta_1\rangle_{\bar S}+|\alpha_1\rangle_S\otimes|\beta_2\rangle_{\bar S}\\
    &+|\alpha_2\rangle_S\otimes|\beta_3\rangle_{\bar S}+|\alpha_3\rangle_S\otimes|\beta_0\rangle_{\bar S}),
\end{align*}
where $\theta\neq 2n\pi, n\in\mathbb{Z}$. It can be verified that any linear combination $\alpha|\psi_0\rangle+\beta|\psi_1\rangle$ also does not support an SCQ, therefore the Condition~\ref{cond:2} is violated. However the mixed state $1/2(|\psi_0\rangle\langle\psi_0|+|\psi_1\rangle\langle\psi_1|)$ has a Schmidt rank of $12$ instead of a full rank.

As shown above in Corollary~\ref{corollary:full rank}, for target states with full Schmidt rank, it is not possible to preserve the state under the action of any superoperator. One may ask whether it is nevertheless possible to transform other intermediate states into the target state. Here, we demonstrate that this is also impossible, as formalized in the following corollary.

\begin{corollary}
For a target GS $|\psi\rangle=\sum_{i} c_i|\alpha_{i}\rangle_S|\beta_{i}\rangle_{\bar S}$ that has full Schmidt rank, the local superoperator $\mathcal{P}_S$, supported on $\mathcal{H}_S$, and the intermediate state $\rho$ satisfying $\mathcal{P}_S(\rho)=|\psi\rangle\langle \psi|$ has only trivial solution: $\rho=|\varphi\rangle\langle \varphi|$, where $|\varphi\rangle=U^\dagger_S\otimes \mathbbm{1}|\psi\rangle$, $U_S$ is an unitary operator on $S$, and $\mathcal{P}_S(\cdot)=U_S\cdot U_S^\dagger$. If there is no such other GS $|\varphi\rangle$, then the superoperator must be trivial.
\end{corollary}

\begin{proof}
Intuitively, $\mathcal{P}_S(\rho)=|\psi\rangle\langle \psi\rangle$ is possible if and only if $\mathcal{P}_S$ is some unitary. Suppose the intermediate state can be decomposed as $\rho=\sum_i p_i |\varphi_i\rangle\langle\varphi_i|$, then $\sum_i p_i\mathcal{P}_S( |\varphi_i\rangle\langle\varphi_i|)=|GS\rangle\langle GS|$. Since $|\psi\rangle\langle \psi|$ is a pure state, i.e. a limit point of the set of density matrices, then $\mathcal{P}_S( |\varphi_i\rangle\langle\varphi_i|)=|\psi\rangle\langle \psi|$ for all $i$, implying $|\varphi_i\rangle$ also has full Schmidt rank since local superoperator does not increase Schmidt rank. So matrix element of $\mathcal{P}_S$ is fully specified as in Corollary~\ref{corollary:full rank}, namely 
\begin{align*}
\mathcal{P}_S({}_{\bar S}\langle\beta_{j}|\varphi_i\rangle\langle\varphi_i|\beta_{k}\rangle_{\bar S})&={}_{\bar S}\langle\beta_{j}|\psi\rangle\langle\psi|\beta_{k}\rangle_{\bar S}\\
&= c_j c_k^*|\alpha_{j}\rangle_S{}_S\langle\alpha_{k}|.
\end{align*}
These matrix elements should be the same for each $|\varphi_i\rangle$, thus all $|\varphi_i\rangle$ are the same, denoted as $\rho=|\varphi\rangle\langle \varphi|$. It is easy to see, the only solution to $\mathcal{P}_S(|\varphi\rangle\langle\varphi|)=|\psi\rangle\langle \psi|$ is then, $|\varphi\rangle=U^\dagger_S\otimes \mathbbm{1}|\psi\rangle$ and $\mathcal{P}_S(\cdot)=U_S\cdot U_S^\dagger$, where $U_S$ is some unitary operator acting on $\mathcal{H}_S$. 
\end{proof}

Thus, even when allowing non-strong superoperators, which is introduced in Appendix~\ref{app:Steering superoperators} and employing it to derive bounds on steerability in Sec.~\ref{sec:A bound on the distance between the surrogate state and the target state employing local measure}, there exists no non-trivial mechanism to generate the target state. This is because the unitary superoperators considered above merely induce rotations within the total Hilbert space, without reducing it to a subspace or enabling convergence to a specific GS. Therefore, it is sufficient to restrict our analysis to strong superoperators.

In the situation that GS manifold is non-degenerate, this corollary can also be used for proof of non-steerability, incorporated with continuous symmetry. This is established by proving the reduced desity matrices are full-rank. For instance, consider a density matrix $\rho$ defined on a qubit, and it is invariant under arbitrary $SU(2)$ transformation, i.e. $\exp(-i\alpha \boldsymbol{n\cdot \sigma}) \times \rho \times \exp(i\alpha \boldsymbol{n\cdot \sigma})=\rho$, then it must be $\rho=\frac{1}{2}\mathbbm 1$. One-dimensional antiferromagnetic Heisenberg model, which is discussed in Sec.~\ref{sec:Anti-ferromagnetic Heisenberg model}, is one model with such local symmetry, $SU(2)^{\otimes N}$. Its Hamiltonian is defined as
$$
H_{\mathrm{AFH}} = \sum_i S_x^i S_x^{i+1} + S_y^i S_y^{i+1} + S_z^i S_z^{i+1},
$$
where $S_{x,y,z}^i$ represent the spin operators for the $i$-th spin, and it is invariant under any joint $SU(2)^{\otimes N}$ rotation. Then for non-degenerate eigenstates, their $n$-body RDM must be invariant under $SU(2)^{\otimes n}$, i.e. $\exp -i \alpha \boldsymbol{n}\cdot\sum_i S^{(i)}$. Since center of $SU(2)$ is $\mathbb{Z}_2$($\pm 1$), a spin-$J$ state invariant under arbitrary $SU(2)$ rotation must be proportional to identity matrix.  Therefore, when we combine multiple spins, their RDM must be proportional to the identity matrix in each collective spin-$J$ sector (it can also be zero). In this situation, when we inspect the RDMs of the GS, if we know that the RDM has a finite population for each collective spin sector, then the RDM on each sector is proportional to the identity matrix, implying that the RDM is full-rank. Therefore, there does not exist SCQ for the GS manifold of $H_{\mathrm{AFH}}$, violating the necessary Condition~\ref{cond:2}.

Another extreme example would be Absolute Maximally Entangled (AME) states~\cite{PhysRevA.92.032316ame}, where any RDM of less than half of the whole system is proportional to identity matrices. In this situation, one must get access to half of the whole system to stabilize this highly entangled states.

\subsection{Calculation of function \texorpdfstring{$p(\Pi_{\mathrm{GS}})$}{p(ΠGS)}}

In this section, we consider the calculation of $1 - p(\Pi_{\mathrm{GS}})$, as introduced in Eq.~\eqref{eq: upper bound degenerate} in Sec.~\ref{sec:A bound on the distance between the surrogate state and the target state employing local measure}. We then illustrate the necessity of computing this quantity exactly by examining the example of Dicke states.

Let us first discuss the calculations of $p(\Pi_{\mathrm{GS}})$. Given the projection operator for the GS manifold, $\Pi_{\mathrm{GS}}$. the function is defined as follows:
\begin{equation}
1-p(\Pi_{\mathrm{GS}})=\sup_{\tilde{\Pi}^{\mathrm{surr}}_i}\lambda(\tilde{\Pi}^{\mathrm{surr}}_i \Pi_{\mathrm{GS}}\tilde{\Pi}^{\mathrm{surr}}_i),
\end{equation}
which is maximized over all possible projection operators $\{\tilde{\Pi}^{\mathrm{surr}}_i\}$ defined on all possible local regions $S_i$, and $\lambda(M)$ is the largest eigenvalue of the operator $M$. As we discussed in Sec.~\ref{sec:Non-Frustration-Free Steadily Steerable Hamiltonians}, when there exist an SCQ in (a subspace of) the GS manifold, then (in this subspace) there exist a GS $|\psi\rangle$ and a local projection operator $\Pi$ such that $\Pi_{\mathrm{GS}}|\psi\rangle=\Pi|\psi\rangle=|\psi\rangle$. This implies $\Pi\Pi_{\mathrm{GS}}\Pi|\psi\rangle=|\psi\rangle$, hence $p(\Pi_{\mathrm{GS}})=0$. Note that this prerequisite satisfies the necessary Condition~\ref{cond:2}. 

Then let us consider the situation that the necessary Condition~\ref{cond:2} is violated, leading to the non-zero $p(\Pi_{\mathrm{GS}})$ function value. For a non-degenerate GS manifold, defined with a single target state $|\psi\rangle$, the calculation of the function $p(\Pi_{\mathrm{GS}})$ is already discussed in Sec.~\ref{sec:A bound on the distance between the surrogate state and the target state employing local measure}. It is given by the smallest eigenvalue of the RDM of $|\psi\rangle$, minimized over local region $S_i$. Equivalently, this value can be obtained as the eigenvalue of $\Pi^{\mathrm{surr}}_i\Pi_{\mathrm{GS}}\Pi^{\mathrm{surr}}_i$, with the eigenvector $\Pi^{\mathrm{surr}}_i|\psi\rangle$. Then let us generalize it to the degenerate case, defined with a set of GS $\{|\psi_i\rangle\}$. In this situation, the GSs projected by the local projection operator $\Pi^{\mathrm{surr}}_i$ might not be orthogonal, i.e. $(\langle \psi_{j}|\Pi^{\mathrm{surr}}_i)(\Pi^{\mathrm{surr}}_i|\psi_k\rangle) \not \propto \delta_{jk}$. Without loss of generality, assume we take a proper basis of GS manifold, such that $\Pi^{\mathrm{surr}}_i|\psi_j\rangle$ are mutually orthogonal for different $j$. Then in this basis, the eigenvectors of $\Pi^{\mathrm{surr}}_i\Pi_{\mathrm{GS}}\Pi^{\mathrm{surr}}_i$ are given by $\Pi^{\mathrm{surr}}_i|\psi_j\rangle$, with eigenvalue $\langle \psi_j|\Pi^{\mathrm{surr}}_i|\psi_j\rangle$. Similar to the discussion of non-degenerate case, running over $\Pi^{\mathrm{surr}}_i$ but for fixed $S_i$, the maximal eigenvalue of $\Pi^{\mathrm{surr}}_i\Pi_{\mathrm{GS}}\Pi^{\mathrm{surr}}_i$ is given by the smallest eigenvalue of the set of RDM $\{\operatorname{Tr}_{S_i}|\psi_j\rangle\langle\psi_j|\}_j$. So in summary, the function $p(\Pi_{\mathrm{GS}})$ can be calculated as the smallest eigenvalue of $\operatorname{Tr}_{S_i}|\psi\rangle\langle\psi|$, running over all possible GS and all local region $S_i$ where the steering operator is applied to.

In general, consider the GSs are not strongly entangled and contain components with small Schmidt coefficients. According to the discussions above, the value of $p(\Pi_{\mathrm{GS}})$ function is given by these small Schmidt coefficients. When we numerically solve for the GSs' wavefunctions approximately, these components might be discarded for numerical efficiency. Let us demonstrate an example that these components should not be discarded when we talk about steerability. 

Consider Dicke states $|D^n_k\rangle$~\cite{PhysRev.93.99}, which is a $n$-qubit state with Hamming Weight $k$. The Dicke state with presence of single $1$ can be written as $|D^n_1\rangle= \frac{1}{\sqrt{n}}(|100\cdots 0\rangle+|010\cdots 0\rangle +|0010\cdots 0\rangle +\cdots + |000\cdots 01\rangle)$. If we view it as a one-dimensional system, then the spatially $k$-local RDM is just $\rho_k=\operatorname{Tr}_{k+1,k+2\cdots N}|D^n_1\rangle\langle D^n_1|= \frac{1}{n}((n-k)|0\rangle\langle 0|^{\otimes k}+ (|100\cdots 0\rangle + |010\cdots 0\rangle + \cdots |00\cdots 1\rangle )(\langle 100\cdots 0|+\langle 010\cdots 0|+\cdots +\langle 00\cdots 1| ))$, namely $\rho_k$ only have two non-zero eigen values, $(n-k)/n$ and $k/n$. So long as $k$ is not macroscopically large, for instance $k=O(\log n)$, $k/n$ vanishes in the thermodynamic limit. 
In this case, the naive approach of discarding eigenstates of the RDM with very small probabilities leads to an approximate RDM of the form $|0\rangle^{\otimes k}$. Designing a steering protocol based on this approximate RDM would then result in the completely trivial state $|0\rangle^{\otimes N}$. Consequently, even for a very large qubit chain with $n = 2 \times 10^{10}$, and the smaller eigenvalue of $\rho_2$ is $10^{-10}$, it should still not be neglected. Moreover, such an approximation may artificially reduce the value of the $p(|0\rangle^{\otimes N})$ function from a tiny but non-zero quantity to exactly zero. This stresses the need for caution when approximating quantum states that exhibit non-local correlations, even if those correlations appear locally irrelevant.

\subsection{Implementation of examples}\label{sec:Implementation of Examples}

In this section, we first discuss the notion of locality for indistinguishable particles. We then present the numerical calculation of the effective temperature lower bound introduced in Sec.~\ref{sec:Derivation of the Temperature lower bound}, for the SYK and Fermi–Hubbard models discussed in Sec.~\ref{sec:Implementations of NFFNS model}. 

Before proceeding, let us clarify the notion of locality for indistinguishable particles. In such systems, the notion of locality can be subtle, as one can perform mode transformations, e.g. Bogoliubov transformations~\cite{valatin1958comments}, that change the support of operators. In other words, an operator that is local in one mode basis can become non-local in another, and vice versa. For example, evaporative cooling is implemented using jump operators like $c_{k > k_0}$, which annihilate particles with momentum above a threshold $k_0$~\cite{PhysRevLett.61.935}. This operator is clearly local in momentum space, but highly non-local in position space, since $c_k = (1/\sqrt{V}) \sum_r e^{-ikr} c_r$. Therefore, we define locality in this context as follows: an operator is said to be local if there exists a physically realizable mode transformation under which the operator acts non-trivially only on a finite number of modes. For instance, any non-interacting Hamiltonian can be diagonalized via such a transformation; in the resulting basis, the Hamiltonian becomes Frustration-Free, and its GS can be steered by annihilating quasi-particle excitations.

\begin{figure*}[htbp]
    \centering
    \begin{subfigure}[b]{0.48\textwidth}
         \centering
         \includegraphics[width=\textwidth]{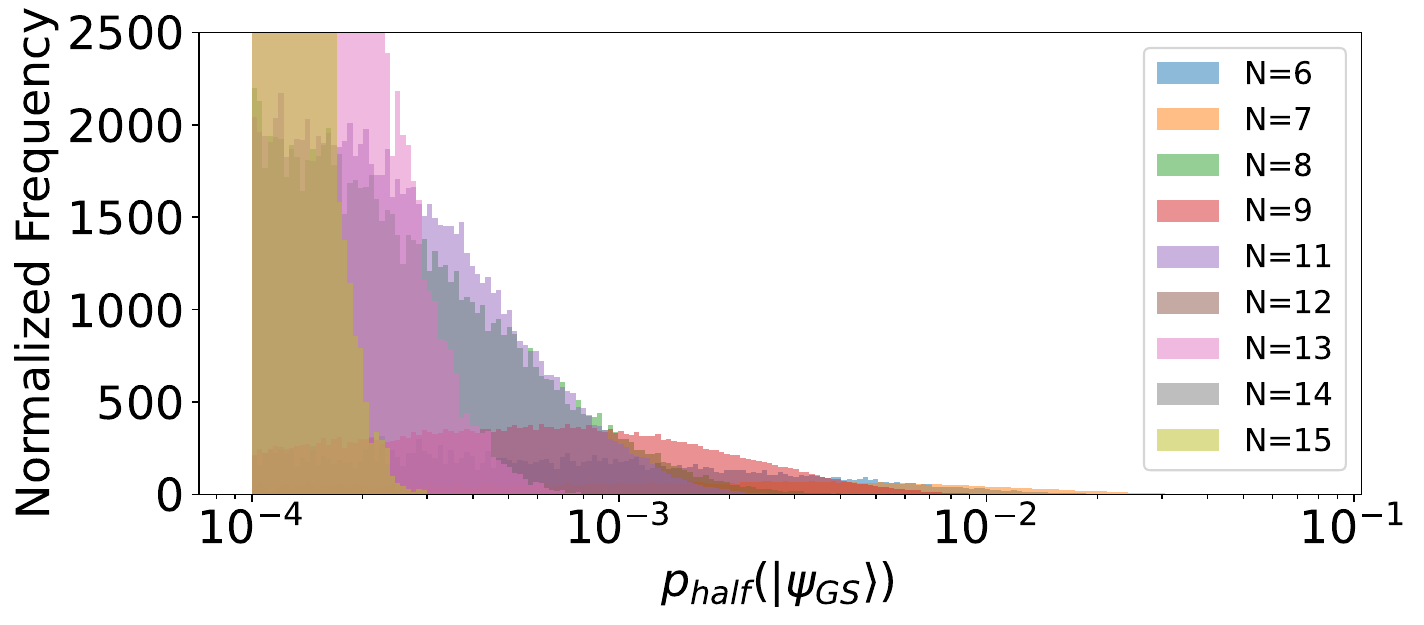}
         \caption{}
     \end{subfigure}
     \hfill
     \begin{subfigure}[b]{0.45\textwidth}
         \centering
         \includegraphics[width=\textwidth]{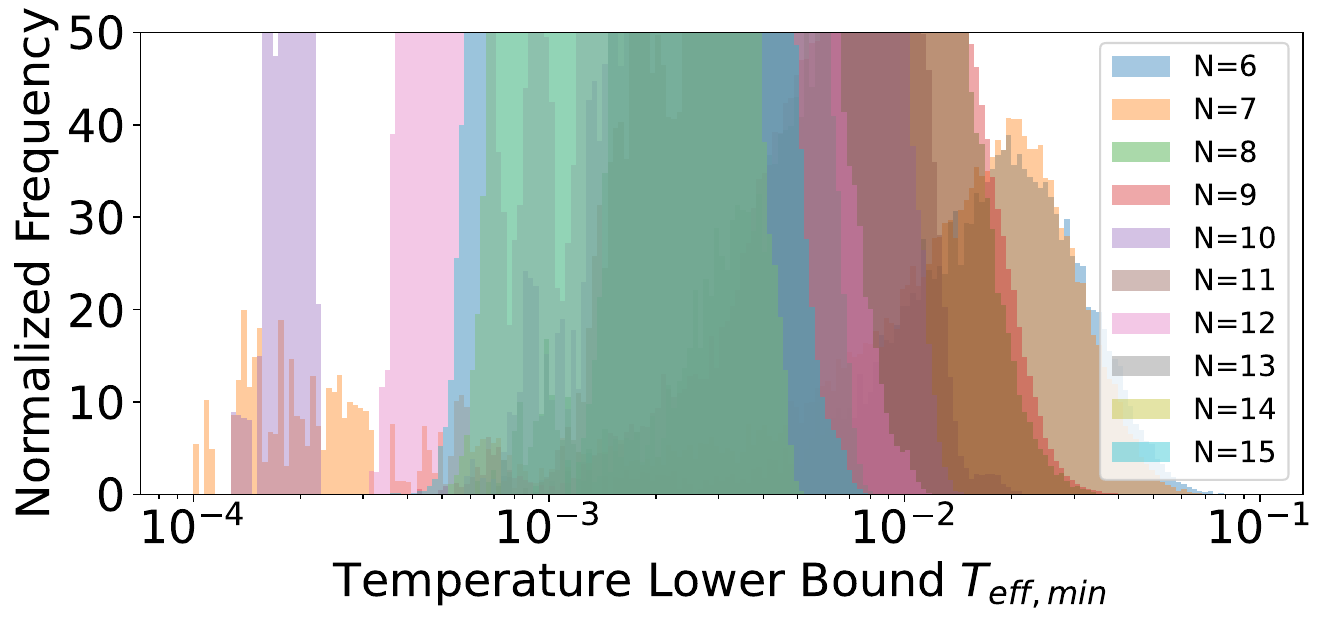}
         \caption{}
     \end{subfigure}
     \hfill
     \begin{subfigure}[b]{0.48\textwidth}
         \centering
         \includegraphics[width=\textwidth]{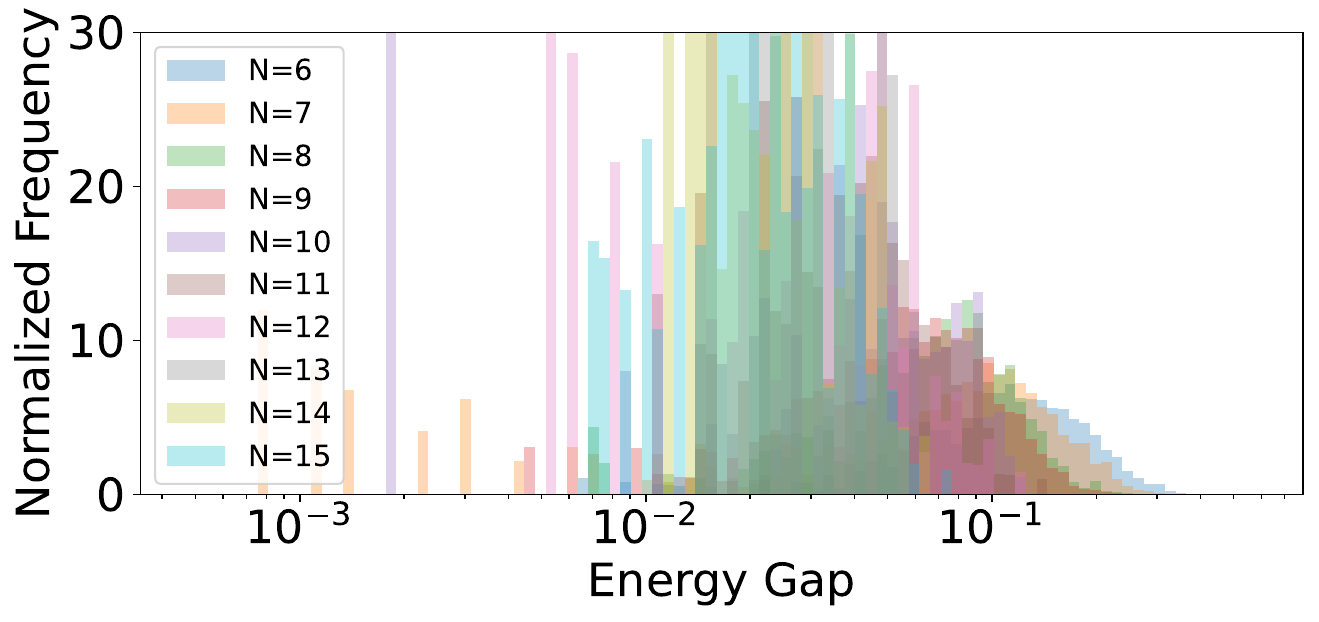}
         \caption{}
     \end{subfigure}
     \hfill
      \begin{subfigure}[b]{0.45\textwidth}
         \centering
         \includegraphics[width=\textwidth]{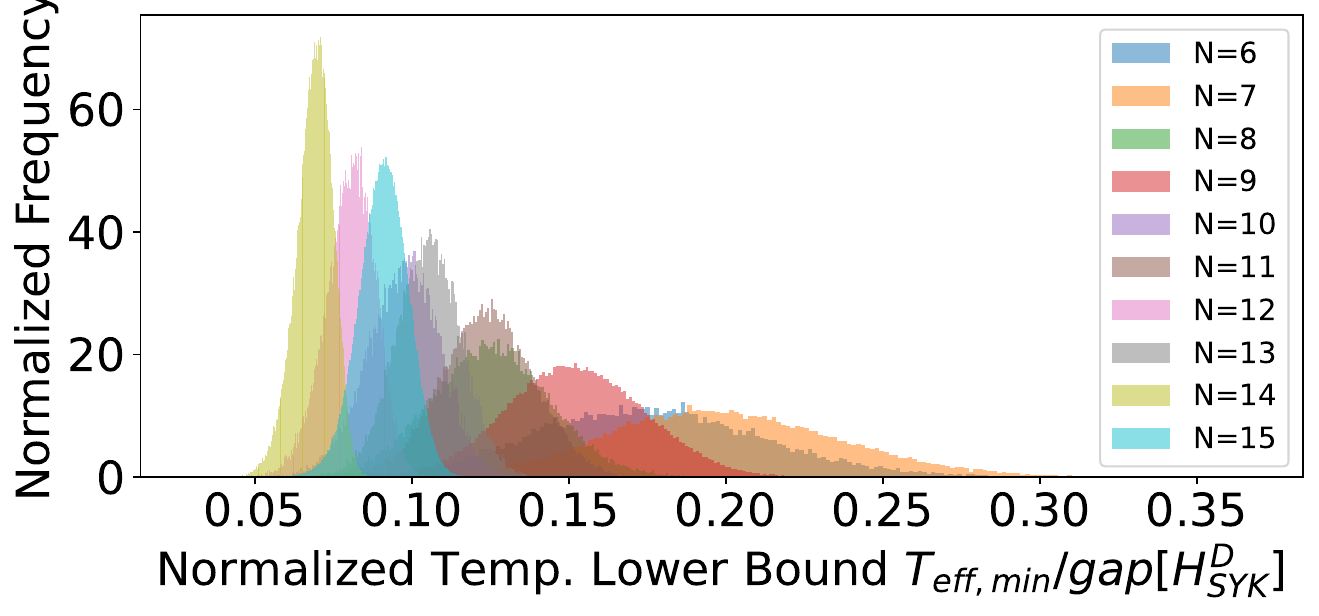}
         \caption{}
     \end{subfigure}
     \hfill
      \begin{subfigure}[b]{0.5\textwidth}
         \centering
         \includegraphics[width=\textwidth]{SYK/SYK-mean.pdf}
         \caption{}
     \end{subfigure}
    \caption{{\it SYK model for Dirac Fermion.} Shown are normalized histograms of several physical quantities for different numbers of fermionic modes $N$. Here the occupation (number of Dirac fermions) is  $\lfloor N/2 \rfloor$.  Steering operators consist of the detector's degree of freedom and $m=\lfloor N/2 \rfloor$ fermionic degrees of freedom. Hence $p_{half}(|\psi\rangle_{\mathrm{GS}})$ (Eq.~\eqref{eq:upper bound}) is obtained by running over all partitions of the system into {\it $m$ fermionic modes} and {\it the rest}. These data are computed with different disorder configurations of $H^{\mathrm{D}}_{\mathrm{SYK}}$. (For the sake of clarity, small tails of the distributions are cut off)  (a) Upper bound on the overlap of the presumed surrogate target state and the ground state $p_{half}(|\psi\rangle_{\mathrm{GS}})$ (cf. Eq.~\eqref{eq:upper bound}). (b)  The minimal effective temperature, $T_{\mathrm{eff,min}}$  (cf. Eq.~\eqref{eq:temp}).  (c) The spectral gap from the ground state energy, $\operatorname{gap}[H^{\mathrm{D}}_{\mathrm{SYK}}]$. (d) The ratio of the minimal effective temperature over the gap, $T_{\mathrm{eff,min}}/\operatorname{gap}[H^{\mathrm{D}}_{\mathrm{SYK}}]$. Note the Gaussian-like nature of the distributions.  (e) Mean and standard deviation for the Gaussian-like distributions in (d). Note that the means and standard deviations for  $N$ odd are systematically larger than for $N$ even.}
    \label{fig:SYK-app}
\end{figure*}

\begin{figure}[htbp]
    \centering
    \begin{subfigure}[b]{0.4\textwidth}
         \centering
         \includegraphics[width=\textwidth]{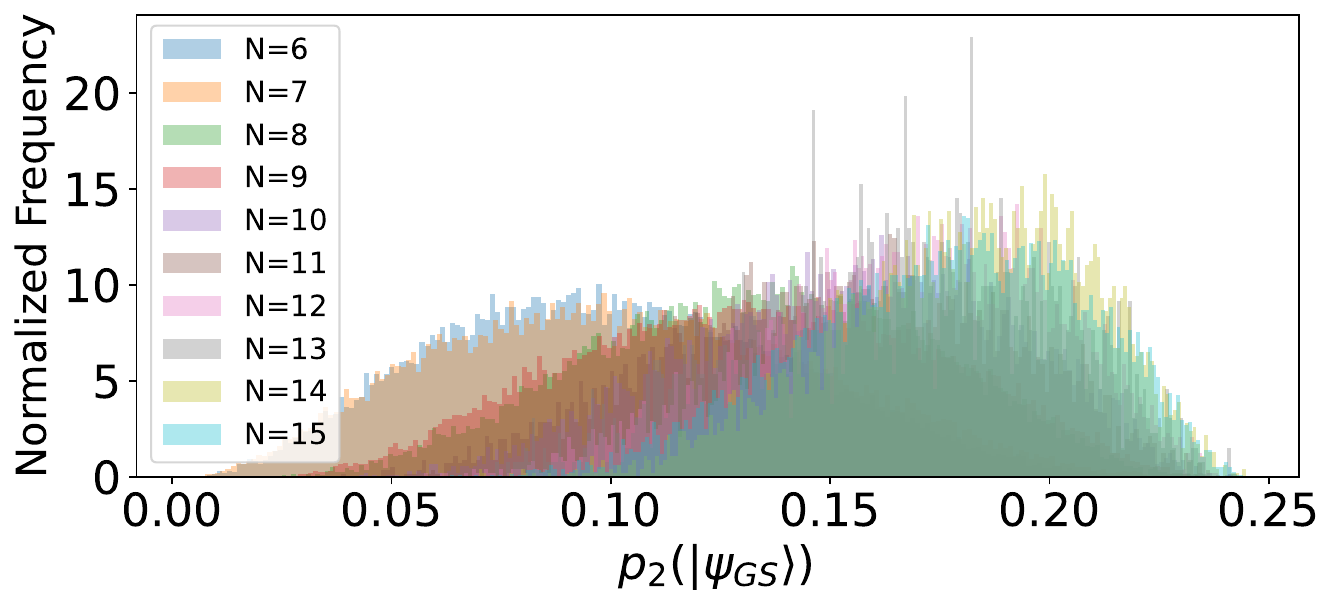}
         \caption{}
     \end{subfigure}
     \hfill
     \begin{subfigure}[b]{0.4\textwidth}
         \centering
         \includegraphics[width=\textwidth]{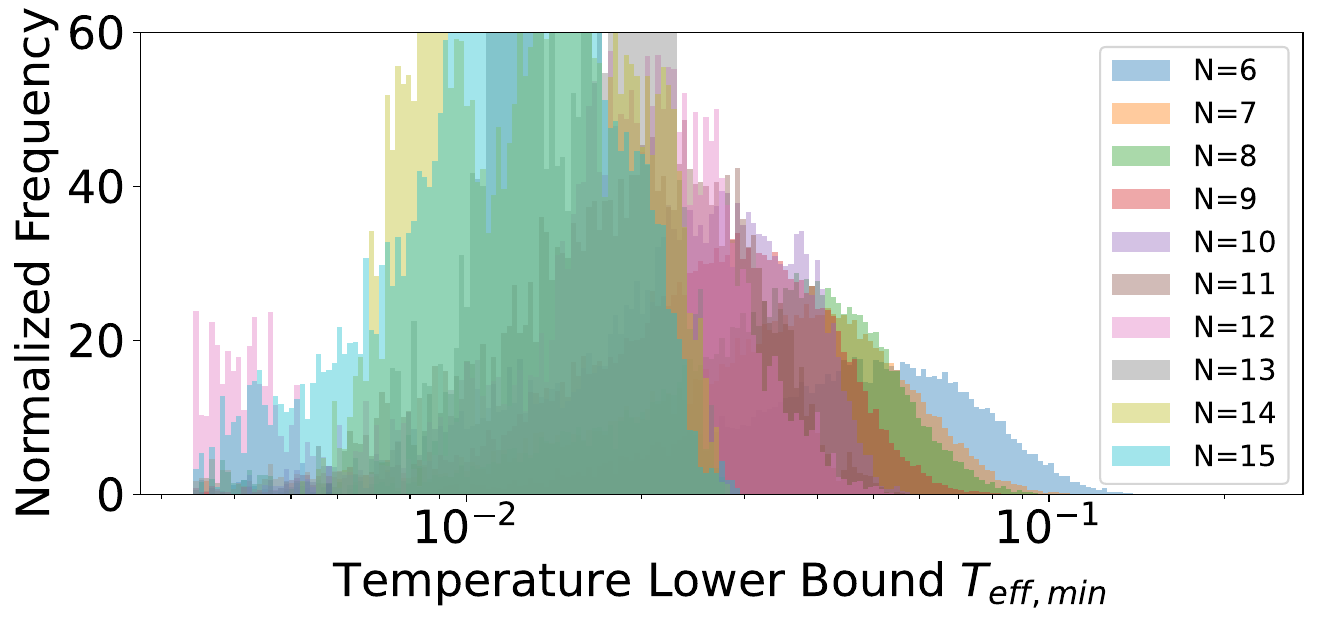}
         \caption{}
     \end{subfigure}
     \hfill
      \begin{subfigure}[b]{0.4\textwidth}
         \centering
         \includegraphics[width=\textwidth]{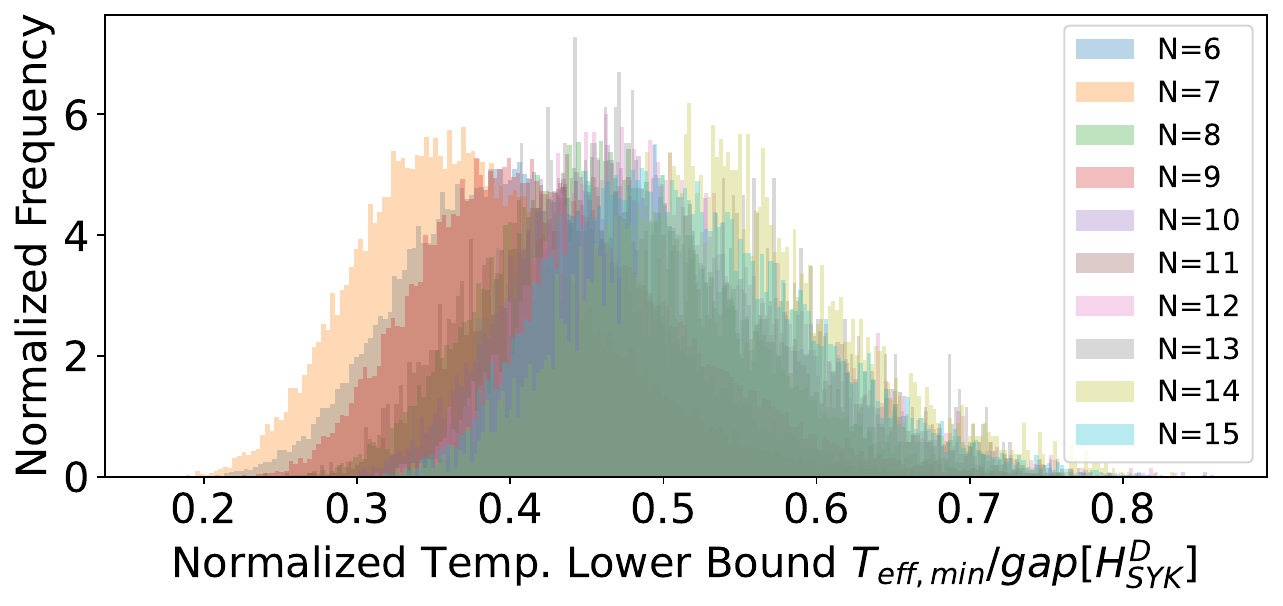}
         \caption{}
     \end{subfigure}
    \caption{{\it SYK model for Dirac fermions}: shown are normalized histograms of several physical quantities for different numbers of fermionic modes $N$. Here the occupation (number of Dirac fermions) is  $\lfloor N/2 \rfloor$.  Steering operators consist of the detector's degree of freedom and  $m=2$ fermionic degrees of freedom.    (a) Maximal overlap of the presumed surrogate target state and the true ground state $p_2(|\psi\rangle_{\mathrm{GS}})$ (cf. Eq.~\eqref{eq:upper bound}). It decreases with $N$. The rank of the RDM coupling to steering operator is 4, hence the function $p_2(|\psi\rangle_{\mathrm{GS}})$ cannot exceed $1/4$. (b)  The minimal effective temperature, $T_{\mathrm{eff,min}}$  (cf. Eq.~\eqref{eq:temp}).  (c) The ratio of the minimal effective temperature over the gap, $T_{\mathrm{eff,min}}/\operatorname{gap}[H^{\mathrm{D}}_{\mathrm{SYK}}]$.  Note deviations from a Gaussian form.}
    \label{fig:SYK2-app}
\end{figure}

The SYK model is introduced in Appendix~\ref{Discussions about steerability}, defined either with Majorana fermions or Dirac fermions. Here, we present the details of the numerical calculation of the lower bound on effective temperature for the SYK model, cf. Sec.~\ref{sec: SYK model}. The calculations are performed using exact diagonalization, where the fermionic Hamiltonian is mapped to qubits via the Jordan–Wigner transformation~\cite{Jordan1928}.

We first consider the case of Dirac fermions, as discussed in Sec.~\ref{sec: SYK model}. The SYK model with Dirac fermions is particle number preserving, i.e. it commutes with the total particle number operator: $[H^{\mathrm{D}}_{\mathrm{SYK}}, N_{\mathrm{par}}] = 0$, where $N_{\mathrm{par}} = \sum_i c_i^\dagger c_i$. As a result, $N_{\mathrm{par}}$ can be used to generate trivial SCQs. For example, in the one-particle subspace ($N_{\mathrm{par}} = 1$), the operator $|1_i\rangle\langle 1_i| \otimes |1_j\rangle\langle 1_j|$ for modes $i, j$ is a trivial bipartite SCQ, since the probability of observing two particles is zero. To avoid such trivial contributions, we perform exact diagonalization in the half-filling subspace $N_{\mathrm{par}} = \lfloor N/2 \rfloor$, where $N$ is the number of fermionic modes.
For each realization of the coupling constants $\{J_{ij;kl}\}$, we compute the function $p(|\psi_{\mathrm{GS}}\rangle)$ using all $m$-body RDMs, where $m = 2$ or $m = \lfloor N/2 \rfloor$, over the range $N = 6, 7, \dots, 15$. We then estimate the lower bound on the effective temperature from the spectrum of the SYK model restricted to the subspace $N_{\mathrm{par}} = \lfloor N/2 \rfloor$, implicitly assuming that steering operations act only within this particle-conserving subspace. The resulting distributions of effective temperature lower bounds are shown in Fig.~\ref{fig:SYK-app} for $\lfloor N/2 \rfloor$-body steering operators and Fig.~\ref{fig:SYK2-app} for $2$-body steering operators. In the former case, the distributions exhibit approximately Gaussian behavior over different realizations of the random couplings $J_{ij;kl}$.

Following the same approach, we also consider the SYK model defined with Majorana fermions, denoted $H^{\mathrm{M}}_{\mathrm{SYK}}$ (see Sec.~\ref{sec: SYK model}). In contrast to $H^{\mathrm{D}}_{\mathrm{SYK}}$, $H^{\mathrm{M}}_{\mathrm{SYK}}$ does not respect particle conservation symmetry, and thus we analyze its full spectrum. For each realization of the random couplings $J_{ijkl}$, we compute $p(|\psi_{\mathrm{GS}}\rangle)$ using exact diagonalization and all $m$-body RDMs with $m = 2$ or $m = \lfloor N/2 \rfloor$, for system sizes $N = 5, 6, \dots, 12$. The resulting distributions of effective temperature lower bounds are shown in Fig.~\ref{fig:MFSYK-app} for $\lfloor N/2 \rfloor$-body steering operators and Fig.~\ref{fig:MFSYK2-app} for $2$-body steering operators. The former distributions also exhibit approximately Gaussian behavior.

\begin{figure*}[htbp]
    \centering
    \begin{subfigure}[b]{0.48\textwidth}
         \centering
         \includegraphics[width=\textwidth]{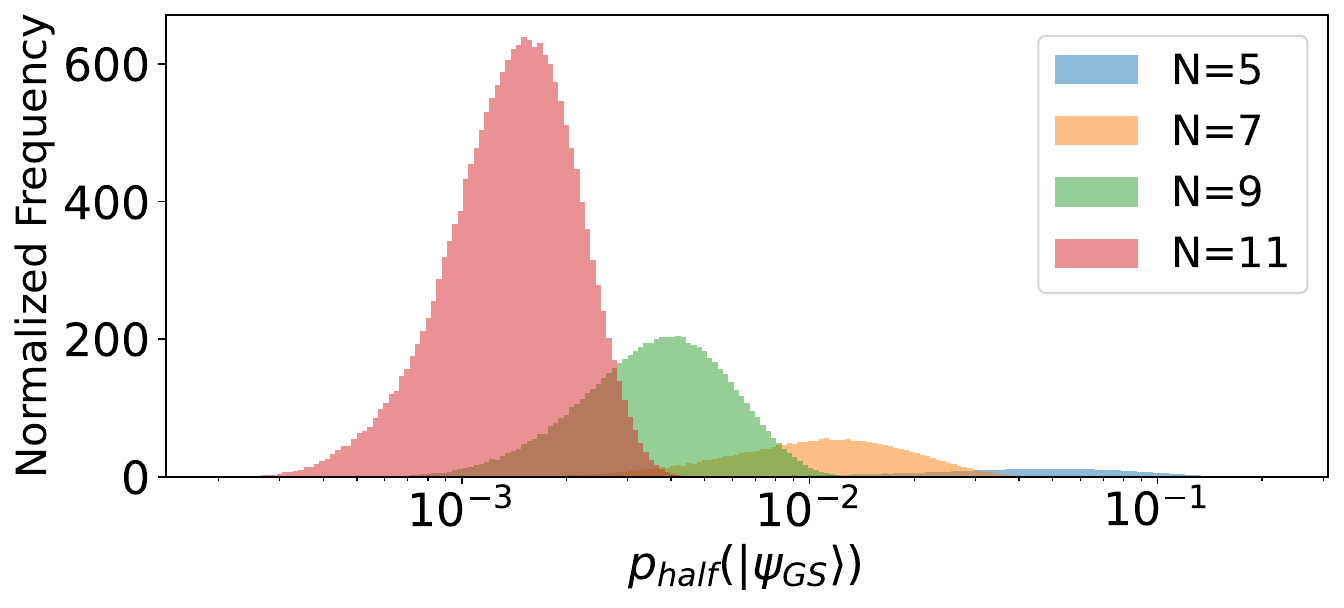}
         \caption{}
     \end{subfigure}
     \hfill
     \begin{subfigure}[b]{0.45\textwidth}
         \centering
         \includegraphics[width=\textwidth]{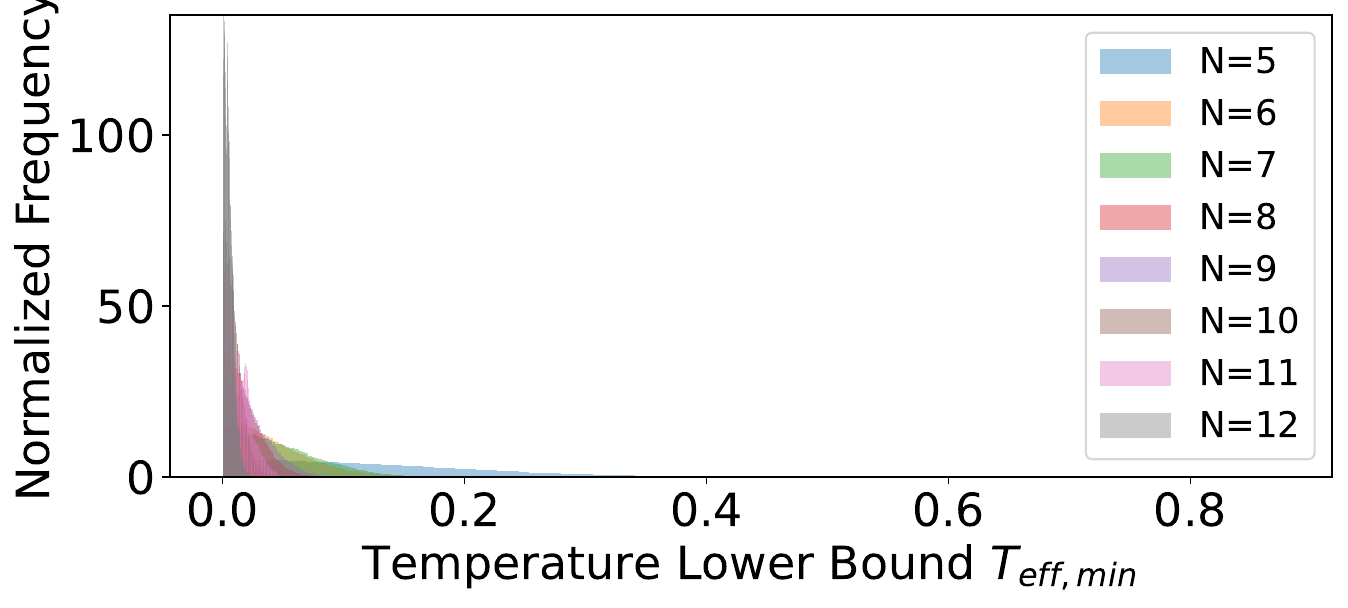}
         \caption{}
         \label{fig:MFSYK-T}
     \end{subfigure}
     \hfill
     \begin{subfigure}[b]{0.45\textwidth}
         \centering
         \includegraphics[width=\textwidth]{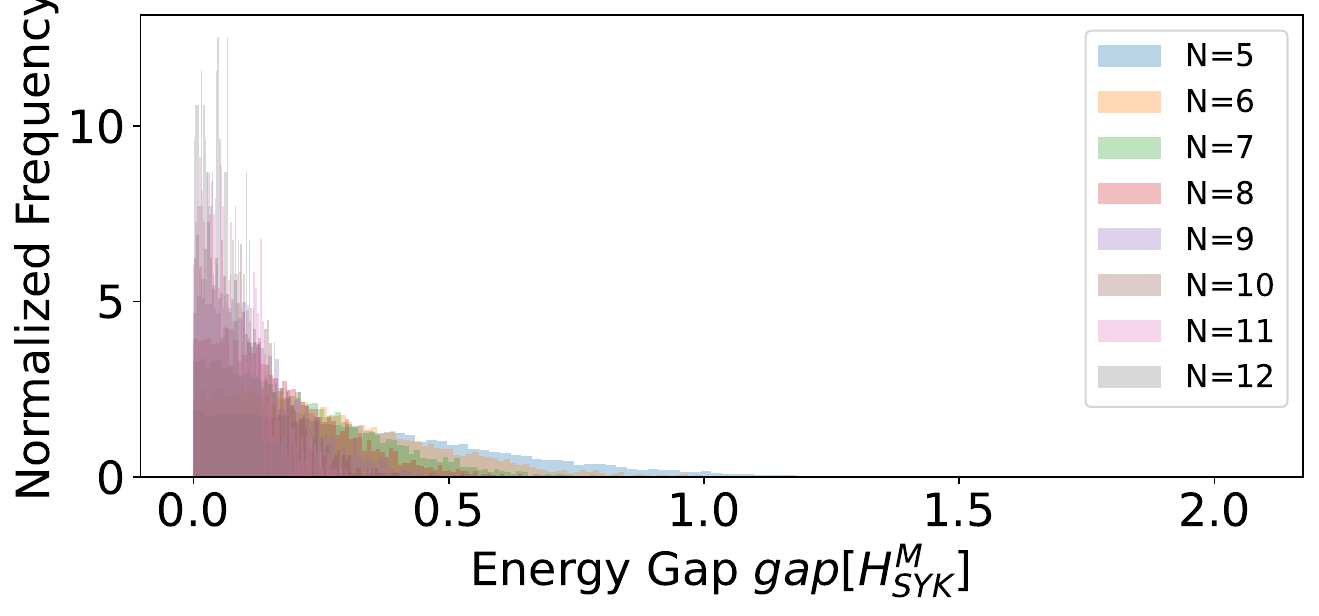}
         \caption{}
         \label{fig:MFSYK-gap}
     \end{subfigure}
     \hfill
      \begin{subfigure}[b]{0.45\textwidth}
         \centering
         \includegraphics[width=\textwidth]{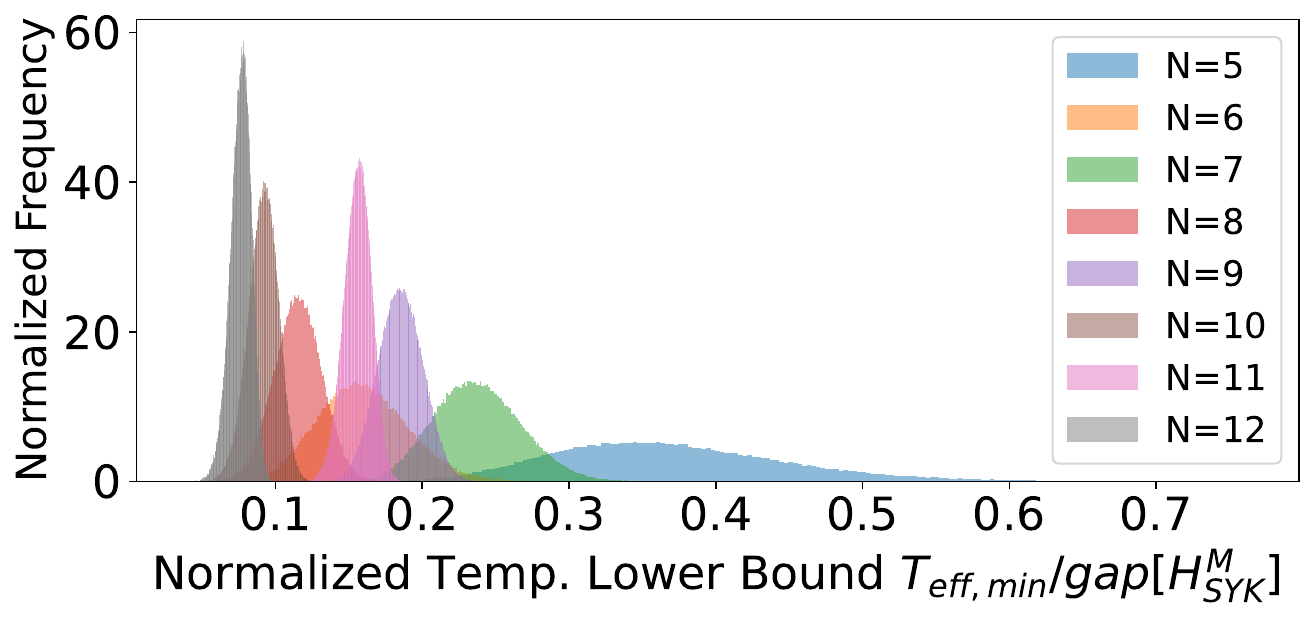}
         \caption{}
     \end{subfigure}
     \hfill
      \begin{subfigure}[b]{0.5\textwidth}
         \centering
         \includegraphics[width=\textwidth]{MFSYK/MFSYK-mean.pdf}
         \caption{}
     \end{subfigure}
    \caption{{\it SYK model for Majorana Fermion.} Shown are normalized histograms of several physical quantities for different numbers of fermionic modes $N$. Steering operators consist of the detector's degree of freedom and $m=\lfloor N/2 \rfloor$ fermionic degrees of freedom. Hence $p_{half}(|\psi\rangle_{\mathrm{GS}})$ (Eq.~\eqref{eq:upper bound}) is obtained by running over all partitions of the system into {\it $m$ fermionic modes} and {\it the rest}. These data are computed with different disorder configurations of $H^{\mathrm{M}}_{\mathrm{SYK}}$.  (a) Upper bound on the overlap of the presumed surrogate target state and the ground state $p_{half}(|\psi\rangle_{\mathrm{GS}})$ (cf. Eq.~\eqref{eq:upper bound}). (b)  The minimal effective temperature, $T_{\mathrm{eff,min}}$  (cf. Eq.~\eqref{eq:temp}).  (c) The spectral gap from the ground state energy, $\operatorname{gap}[H^{\mathrm{M}}_{\mathrm{SYK}}]$. (d) The ratio of the minimal effective temperature over the gap, $T_{\mathrm{eff,min}}/\operatorname{gap}[H^{\mathrm{M}}_{\mathrm{SYK}}]$. Note the Gaussian-like nature of the distributions.  (e) Mean and standard deviation for the Gaussian-like distributions in (d). Note that the means and standard deviations for  $N$ odd are systematically larger than for $N$ even.}
    \label{fig:MFSYK-app}
\end{figure*}

\begin{figure}[htbp]
    \centering
    \begin{subfigure}[b]{0.4\textwidth}
         \centering
         \includegraphics[width=\textwidth]{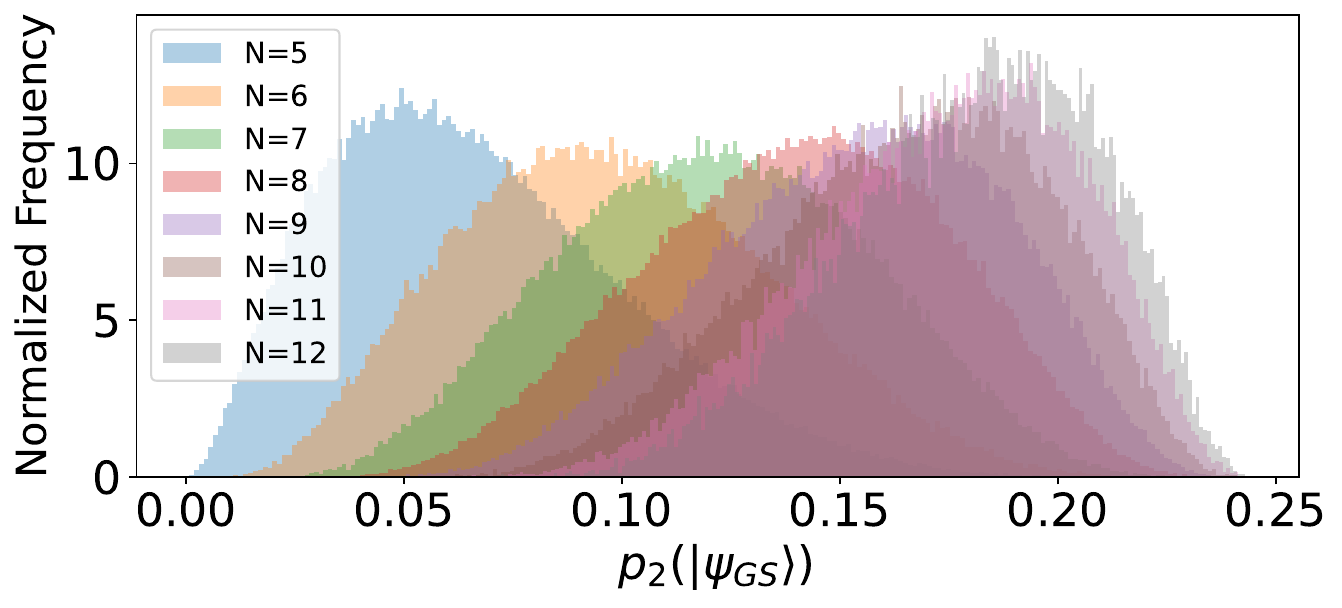}
         \caption{}
     \end{subfigure}
     \hfill
     \begin{subfigure}[b]{0.4\textwidth}
         \centering
         \includegraphics[width=\textwidth]{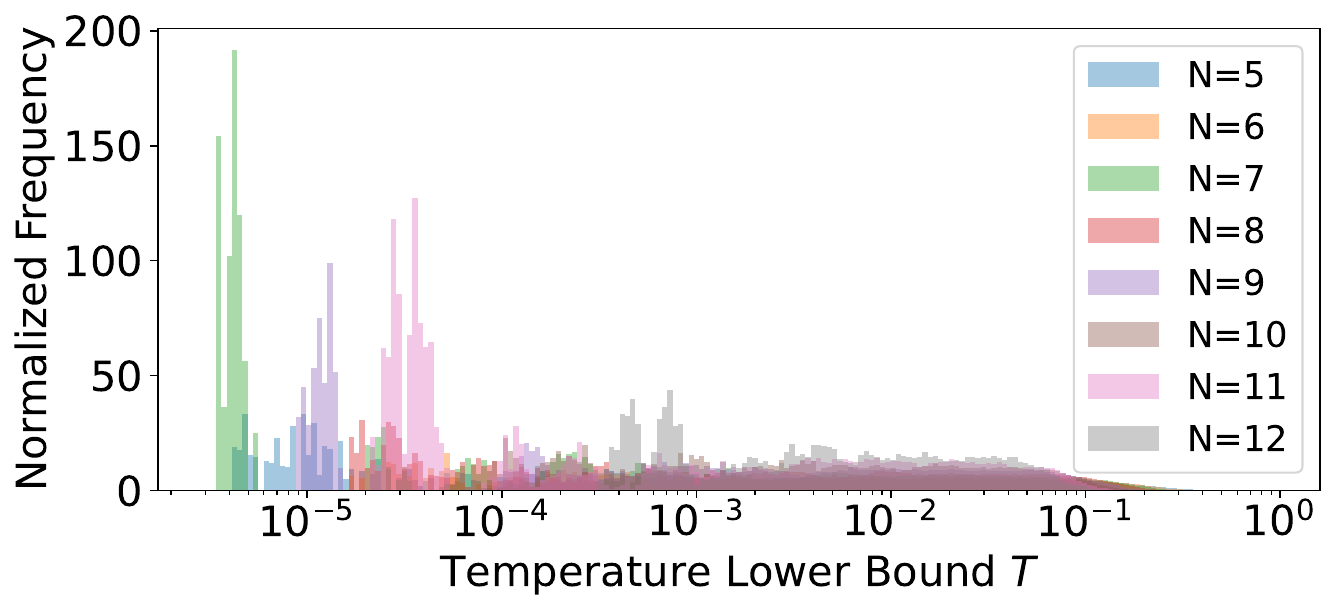}
         \caption{}
     \end{subfigure}
     \hfill
      \begin{subfigure}[b]{0.4\textwidth}
         \centering
         \includegraphics[width=\textwidth]{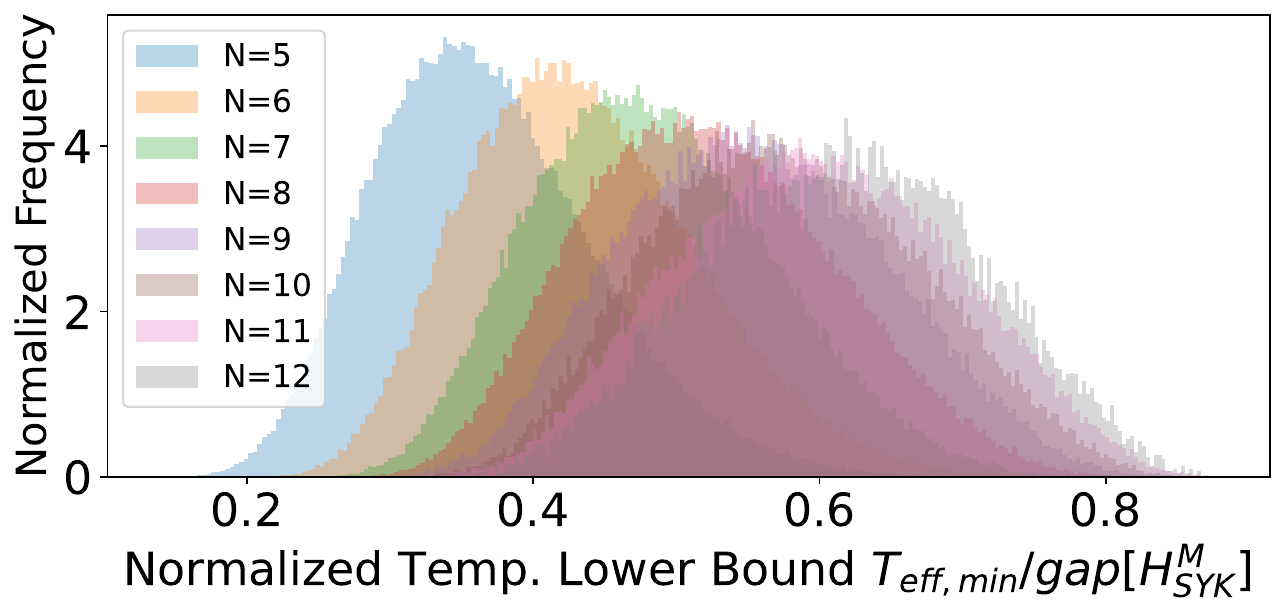}
         \caption{}
     \end{subfigure}
    \caption{{\it SYK model for Fermion fermions}: shown are normalized histograms of several physical quantities for different numbers of fermionic modes $N$. Steering operators consist of the detector's degree of freedom and  $m=2$ fermionic degrees of freedom.    (a) Maximal overlap of the presumed surrogate target state and the true ground state $p_2(|\psi\rangle_{\mathrm{GS}})$ (cf. Eq.~\eqref{eq:upper bound}). It decreases with $N$. The rank of the RDM coupling to steering operator is 4, hence the function $p_2(|\psi\rangle_{\mathrm{GS}})$ cannot exceed $1/4$. (b)  The minimal effective temperature, $T_{\mathrm{eff,min}}$  (cf. Eq.~\eqref{eq:temp}).  (c) The ratio of the minimal effective temperature over the gap, $T_{\mathrm{eff,min}}/\operatorname{gap}[H^{\mathrm{M}}_{\mathrm{SYK}}]$.  Note deviations from a Gaussian form.}
    \label{fig:MFSYK2-app}
\end{figure}

To analyze the Fermi–Hubbard model discussed in Sec.~\ref{sec:Fermi-Hubbard model}, where the GS manifold may be degenerate in certain regimes, we require an additional proposition to determine whether the GS subspace is bipartite distinguishable, i.e. whether it violates Condition~\ref{cond:3}. For a target subspace $\mathcal{H}_{\mathrm{target}}$, let ${|\psi_i\rangle}$ denote an orthonormal basis. As discussed in Definition~\ref{def:qgd} in Appendix~\ref{sec:Necessary conditions for Non-Frustration-Free Jittery Steerable (NFFJS) class}, bipartite distinguishability of $\mathcal{H}_{\mathrm{target}}$ is equivalent to the linear independence of the set ${\operatorname{Tr}_S(|\psi_i\rangle\langle\psi_j|)}$.
In non-steerable systems, the degeneracy of the GS manifold often arises from global symmetries. Consequently, the states $\{|\psi_i\rangle\}$ are globally very different, and are thus expected to be distinguishable on $\mathcal{H}_{\bar S}$. In such cases, we can formulate the following proposition as a practical criterion for bipartite distinguishability.
\begin{prop}
\label{prop:qgd}
    $\mathcal{H}_{\mathrm{target}}=\operatorname{span}\{|\psi_i\rangle\}$ is bipartite distinguishable for the bipartition $\mathcal{H}_S\otimes \mathcal{H}_{\bar S}$, if $\operatorname{span}(\operatorname{Tr}_S(|\psi_{i}\rangle\langle\psi_{i}|))\not \subset \bigcup_{j<i}\operatorname{span}(\operatorname{Tr}_S(|\psi_{j}\rangle\langle\psi_{j}|))$. Here by the span of a matrix (the RDM) we mean the linear space spanned by the eigenvectors of the matrix.
\end{prop}
This proposition means, when there is a way to distinguish $|\psi_i\rangle$ from $|\psi_j\rangle$ for all $j<i$, then the subspace is bipartite distinguishable. It is sufficient that, if $\dim{\bigcup_{j}\operatorname{span}(\operatorname{Tr}_S(|\psi_{j}\rangle\langle\psi_{j}|))}=\sum_j\dim{\operatorname{span}(\operatorname{Tr}_S(|\psi_{j}\rangle\langle\psi_{j}|))}$, then the condition in Prop.~\ref{prop:qgd} is satisfied.
\begin{proof}
    Let us take a bases of $\mathcal{H}_S$ as $|i_S\rangle$ and a bases of $\mathcal{H}_{\bar S}$ as $|j_{\bar S}\rangle$, and denote $d_1=\dim\mathcal{H}_{\mathrm{target}}$ as the dimension of $\mathcal{H}_{\mathrm{target}}$. Then any state in the subspace can be spanned as $|\psi_i\rangle=\sum_{jk} c^i_{jk}|j_S\rangle|k_{\bar S}\rangle$. Then it suffices to prove that  $\operatorname{Tr}_S(|\psi_i\rangle\langle\psi_j|)=\sum_{klm}c^i_{mk}c^{j*}_{ml}|k_{\bar S}\rangle\langle l_{\bar S}|$ are linearly independent. Suppose they are not linearly independent, then there exist a coefficient matrix $e_{ij}\neq 0$ such that $\sum_{ijm}e_{ij}c^i_{mk}c^{j*}_{ml}=0$.  Let us reselect the bases $|j_{\bar S}\rangle$ as the sequential Schmidt orthogonalization of the set of bases, $\{\langle i_S|\psi_1\rangle\}_i,\{\langle i_S|\psi_2\rangle\}_i,\cdots,\{\langle i_S|\psi_{d_1}\rangle\}_i$. The condition in this proposition implies that the eigenstates of $\operatorname{Tr}_S|\psi_i\rangle\langle\psi_i|$ is not included in the bases $\operatorname{span}_{j,k<i}\{\langle j_S|\psi_k\rangle\}$, then $c_{ij}$ is a lower triangle matrix if we treat $i$ as row indices and $j$ as column indices. Let us denote $ d_2= \dim{\operatorname{span}|j_{\bar S}\rangle}$, then the condition also implies that $c^{d_1}_{mk}c^{d_1*}_{ml}$ is the only matrix has non-zero matrix element on $k=d_2,l=d_2$. This implies $e_{d_1, d_1}=0$. Then one can similarly assert that $e_{d_1, d_1-1}= e_{d_1-1, d_1}=0,e_{d_1, d_1-2}= e_{d_1-2, d_1}=0, \cdots$ thus $e_{ij}=0$. So we prove $c^{i}_{mk}c^{j*}_{ml}$ are linearly independent.
\end{proof}

In this paragraph, we discuss numerical calculations related to the Fermi–Hubbard model introduced in Sec.~\ref{sec:Fermi-Hubbard model}, which is widely used as a theoretical model for high-$T_c$ superconductivity~\cite{anderson2013twenty}. The Hamiltonian is given by
$$
H_{\mathrm{FH}}=-t\sum_{\langle ij\rangle}(c^\dagger_{i\uparrow} c_{j\uparrow}+c^\dagger_{i\downarrow} c_{j\downarrow})+U\sum_i n_{i\uparrow}n_{i\downarrow},
$$
where $c_{i\uparrow,\downarrow}$ is the annihilation operator for an electron at site $i$ with spin up or down, and $n_{i\uparrow,\downarrow} = c^\dagger_{i\uparrow,\downarrow} c_{i\uparrow,\downarrow}$ is the corresponding number operator. The first summation runs over nearest-neighbor pairs. This Hamiltonian does not possess any local conserved quantities, but it does have trivial SCQs due to particle-number conservation, i.e., $[H_{\mathrm{FH}}, N_{\uparrow}] = [H_{\mathrm{FH}}, N_{\downarrow}] = 0$, where $N_{\uparrow,\downarrow} = \sum_i n_{i\uparrow,\downarrow}$.

We use exact diagonalization to compute the GS manifold of $H_{\mathrm{FH}}$ on a $3 \times 3$ lattice with periodic boundary conditions (see Fig.~\ref{fig:FHL}) and analyze its few-body RDMs to calculate the effective temperature lower bound discussed in Sec.~\ref{sec:Derivation of the Temperature lower bound}. We consider two cases: (i) Half-filling, i.e., the antiferromagnetic phase;(ii) Near-half-filling, corresponding to a regime close to $d$-wave superconductivity with doping around $0.15$~\cite{sigrist2005introduction}. The temperature lower bound is calculated with the grand canonical ensemble.

For case (i), we focus on the subspace with $N_{\uparrow} = 4$, $N_{\downarrow} = 5$, where the GS of $H_{\mathrm{FH}}$ is four-fold degenerate. This degeneracy arises from translation invariance. Denoting the translation operators by one lattice site along the $x$ and $y$ directions as $T_x$ and $T_y$, with eigenvalues $1, e^{2\pi i/3}, e^{4\pi i/3}$, the four GSs belong to subspaces labeled by $(T_x,T_y)=(1,e^{\frac{2i\pi}{3}}),(1,e^{\frac{4i\pi}{3}}),(e^{\frac{2i\pi}{3}},1),(e^{\frac{4i\pi}{3}},1)$. These states are globally distinct, and direct investigations of their wavefunctions confirms that the GS manifold is bipartite distinguishable, as established by Prop.~\ref{prop:qgd}. The corresponding temperature lower bound is shown in Fig.~\ref{fig:FHAFM}.
Assuming spatial locality of the superoperators and exploiting translation invariance, it suffices to consider representative regions $S$ as shown in Fig.~\ref{fig:FHL}. For the function $p(\Pi_{\mathrm{GS}})$ defined in Eq.~\eqref{eq: upper bound degenerate}, we consider: (1) two-body RDMs, where superoperators can either flip the spin on a single site or transfer a particle between neighboring sites without spin change; (2) three-body and four-body RDMs, where only particle transfers (no spin flips) are considered.

\bibliography{main}

\end{document}